\documentclass[
final, nomarks
]{dmtcs-episciences}

\usepackage[utf8]{inputenc}
\usepackage{subfigure}

\usepackage[round]{natbib}

\usepackage{graphicx}
\usepackage{amsmath, amssymb, amsfonts, amstext, mathdots, latexsym}
\usepackage{algorithm, algorithmic}
\usepackage{caption, subcaption}
\usepackage{multirow, slashbox, rotating}
\usepackage{tipa, arcs, fix-cm}
\usepackage{enumerate, setspace, makeidx, epsfig, verbatim}

\newcommand{\ignore}[1]{}

\def\trs{\mathcal{D}}
\def\tr{\Delta}
\newcommand{\norm}[1]{\left\lVert#1\right\rVert}
\DeclareMathOperator*{\argmax}{arg\,max}

\newtheorem{theorem}{Theorem}[section]

\newtheorem{lemma}[theorem]{Lemma}
\newtheorem{corollary}[theorem]{Corollary}

\newtheorem{observation}[theorem]{Observation}

\author[Konstantinos Georgiou et al.]{Konstantinos Georgiou\affiliationmark{1}
  \and Somnath Kundu\affiliationmark{1}
  \and Pawe\l{} Pra\l{}at\affiliationmark{1}}
\title[Makespan Trade-offs for Visiting Triangle Edges]{Makespan Trade-offs for Visiting Triangle Edges\thanks{Research supported in part by NSERC.}}
\affiliation{
  Toronto Metropolitan University, Toronto, Canada}

\begin{document}
\publicationdata{vol. 26:3}{2024}{22}{10.46298/dmtcs.8729}{2021-11-19; 2021-11-19; 2024-07-22}{2024-12-10}
\maketitle
\begin{abstract}
We study a primitive vehicle routing-type problem in which a fleet of $n$ unit speed robots start from a point within a non-obtuse triangle $\tr$, where $n \in \{1,2,3\}$. 
The goal is to design robots' trajectories so as to visit all edges of the triangle with the smallest visitation time makespan.
We begin our study by introducing a framework for subdividing $\tr$ into regions with respect to the type of optimal trajectory that each point $P$ admits, pertaining to the order that edges are visited and to how the cost of the minimum makespan $R_n(P)$ is determined, for $n\in \{1,2,3\}$. 
These subdivisions are the starting points for our main result, which is to study makespan trade-offs with respect to the size of the fleet. 
In particular, we define $\mathcal R_{n,m} (\tr)= \max_{P \in \tr} R_n(P)/R_m(P)$, and we prove that, over all non-obtuse triangles $\tr$: (i) 
$\mathcal R_{1,3}(\tr)$ ranges from $\sqrt{10}$ to $4$, 
(ii) $\mathcal R_{2,3}(\tr)$ ranges from $\sqrt{2}$ to $2$, and
(iii) $\mathcal R_{1,2}(\tr)$ ranges from $5/2$ to $3$. 
In every case, we pinpoint the starting points within every triangle $\tr$ that maximize $\mathcal R_{n,m} (\tr)$, as well as we identify the triangles that determine all $\inf_\tr \mathcal R_{n,m}(\tr)$ and $\sup_\tr \mathcal R_{n,m}(\tr)$ over the set of non-obtuse triangles. 
\end{abstract}

\section{Introduction}
\label{sec: intro}

Vehicle routing problems form a decades old paradigm of combinatorial optimization questions. In the simplest form, the input is a fleet of robots (vehicles) with some starting locations, together with stationary targets that need to be visited (served). Feasible solutions are robots' trajectories that eventually visit every target, while the objective is to minimize either the total length of traversed trajectories or the time that the last target is visited. 

Vehicle routing problems are typically NP-hard in the number of targets. 
The case of 1 robot in a discrete topology corresponds to the celebrated Traveling Salesman Problem whose variations are treated in numerous papers and books. Similarly, numerous vehicle routing-type problems have been proposed and studied, varying with respect to the number of robots, the domain's topology and the solutions' specs, among others. 

We deviate from all previous approaches and we focus on efficiency trade-offs, with respect to the fleet size, of a seemingly simple geometric variation of a vehicle routing-type problem in which targets are the edges of a non-obtuse triangle. The optimization problem of visiting all these three targets (edges), with either 1, 2 or 3 robots, is computationally degenerate. Indeed, even in the most interesting case of 1 robot, an optimal solution for a given starting point can be found by comparing a small number of candidate optimal trajectories (that can be efficiently constructed geometrically). From a combinatorial geometric perspective, however, the question of characterizing the points of an arbitrary non-obtuse triangle with respect to optimal trajectories they admit when served by 1 or 2 robots, e.g.\ the order that targets are visited, is far from trivial (and in fact it is still eluding us in its generality). 

In the same direction, we ask a more general question: Given an arbitrary non-obtuse triangle, what is the worst-case trade-off ratio of the cost of serving its edges with different number of robots, over all starting points?
Moreover, what is the smallest and what is the largest such value as we range over all non-obtuse triangles? Our main contributions pertain to the development of a technical geometric framework that allows us to pinpoint exactly the best-case and worst-case non-obtuse triangles, along with the worst-case starting points that are responsible for the extreme values of these trade-off ratios. 
To the best of our knowledge, the study of efficiency trade-offs with respect to fleet sizes is novel, at least for vehicle routing type problems or even in the realm of combinatorial geometry.

\subsection{Motivation \& Related Work}

The main motivation for our work stems from its classification as a vehicle routing-type problem, first introduced by~\cite{dantzig1959truck}. In vehicle routing problems (VRPs), the primary objective is to minimize either the visitation time (makespan) or the total distance traveled to serve a set of targets using a fleet of (usually capacitated) robots. Early results on this topic are detailed in surveys such as the one by~\cite{laporte2,toth2002vehicle}.

While VRPs are typically studied in discrete domains, geometric vehicle routing problems have also been extensively explored, as seen by~\cite{das2015quasipolynomial}. The long list of VRP variations proposed over time has led to numerous subject-focused surveys; see~\cite{kumar2012survey,mor2020vehicle,ritzinger2016survey} for three relatively recent examples.

Famously, VRPs generalize the celebrated Traveling Salesman Problem (TSP), where a single vehicle must efficiently tour a set of targets. Similar to VRPs, TSP has numerous variations, including geometric ones, e.g. by~\cite{Arora98,DumitrescuT16}, where targets are lines in the latter work. The natural extension of TSP to multiple vehicles is known as the Multiple Traveling Salesman Problem (MTSP), see~\cite{bektas2006multiple}, a variant of VRPs where vehicles are un-capacitated. MTSP has also been studied with variations in the initial deployment of vehicles, either from a single location (single depot), as in our problem, or from multiple locations.

Apart from being related to vehicle routing problems, the geometric traveling salesman problem, and search and exploration games, our problem also relates to the so-called shoreline search problem, first introduced by~\cite{baeza1988searching}. Our research question arose from attempting to derive new lower bounds for this problem.

In the shoreline search problem, a unit robot searches for a hidden line on the plane, unlike our problem where the triangle edges are visible. The objective is to visit the line as quickly as possible relative to the distance from the robot's initial placement. The best known algorithm for this problem has a performance ratio of roughly 13.81, with only very weak (unconditional) lower bounds known~\cite{baeza1995parallel}. Recently, the problem of searching with multiple robots was revisited, resulting in new lower bounds by~\cite{AcharjeeGKS19,DKP20}.

In typical online problems, a lower bound argument involves allowing an arbitrary algorithm to run for a certain time until the hidden item is placed at a location that the robot has not yet visited. The lower bound is then obtained by adding the elapsed time to the distance from the robot to the hidden item (the line), as it is assumed that the online algorithm has full knowledge of the input at this point. Applying this strategy to the shoreline problem involves identifying a number of lines as close as possible to the robot's starting point and computing the shortest trajectory for the robot to visit all of them, which mirrors our problem. This reasoning also applies to the case of multiple agents.

In the simplest configuration that could yield strong bounds, three lines forming a non-obtuse triangle are identified. The question then arises: what is the shortest trajectory that allows one (or multiple) agent(s) to visit all these edges? This question led to the research we present in this work. The motivation for restricting our attention to non-obtuse triangles is twofold: firstly, to address this specific configuration that arose from our motivating search problem, and secondly, because the optimal visitation cost for three robots in non-obtuse triangles is defined as the maximum distance over all triangle edges, treated as lines.
Although our quantified results do not have immediate implications for the shoreline problem (or its lower bounds), we hope that the techniques we developed to compare optimal visitation trajectories by one or more agents will provide new insights for improving the lower bounds of the shoreline problem.


\subsection{A Note on our Contributions \& Paper Organization}

We introduce and study a novel concept of efficiency/fleet size trade-offs in a special geometric vehicle routing-type problem that we believe is interesting in its own right. Deviating from the standard combinatorial perspective of the problem, we focus on the seemingly simple case of visiting the three edges of a non-obtuse triangle with $n\in \{1,2,3\}$ robots. Interestingly, the problem of characterizing the starting points within arbitrary non-obtuse triangles with respect to structural properties of the optimal trajectories they admit is a challenging question. 
More specifically, one would expect that the latter characterization is a prerequisite in order to analyze efficiency trade-offs when serving with different number of robots, over all triangles. 
Contrary to this intuition, and without fully characterizing the starting points of arbitrary triangles, 
we develop a framework that allows us 
(a) to pinpoint the starting points of any triangle at which these (worst-case) trade-offs attain their maximum values, and 
(b) to identify the extreme cases of non-obtuse triangles that set the boundaries of the inf and sup values of these worst-case trade-offs.

In Section~\ref{sec: definition and contributions} we give a formal definition of our problem, as well as we quantify our main results.
Section~\ref{sec: terminology obesrvations} introduces some basic terminology, together with some preliminary and important observations. 
Then, in Section~\ref{sec: visiting 1,2 edges}, we address first the basic question of visiting optimally two triangle edges, and we move in Section~\ref{sec: visiting 3 edges} to the problem of visiting optimally three triangle edges in a specific order. 
Section~\ref{sec: Rn regions} is the beginning of our technical contribution, were we introduce a framework for characterizing triangle points with respect to optimal solutions they admit when serving with $n=3,2,1$ robots, see Sections~\ref{sec: regions 3}, \ref{sec: regions 2} and \ref{sec: regions 1}, respectively. 
Using that framework, we expand our technical contribution by computing in Section~\ref{sec: special visitations} the visitation cost with 1, 2 robots of some special triangle starting points. 
Finally, in Section~\ref{sec: trade-offs} we quantify the efficiency trade-offs with respect to the fleet size where, in particular, Sections~\ref{sec: 13 sup}, \ref{sec: 23 sup} and \ref{sec: 12 sup} focus on the cases of serving with 1 vs.\ 3 robots, 2 vs.\ 3 robots, and 1 vs.\ 2 robots, respectively.

\section{Our Results \& Basic Terminology and Observations}
\label{sec: terminology results}

\subsection{Problem Definition \& Main Contributions}
\label{sec: definition and contributions}

We consider the family of non-obtuse triangles $\trs$, equipped with the Euclidean distance.
For any $n\in \{1,2,3\}$, any given triangle $\tr \in \trs$, and any point
$P$ in the triangle, denoted by $P \in \tr$,
we consider a fleet of $n$ unit speed robots starting at point $P$. 
A feasible solution to the triangle $\tr$ visitation problem with $n$ robots starting from $P$ is given by robots' trajectories that eventually visit every edge of $\tr$, that is, each edge needs to be touched by at least one robot in any of its points including the endpoints. 
The visitation cost of a feasible solution is defined as the makespan of robots' trajectory lengths, or equivalently as the first time by which every edge is touched by some robot. By $R_n(\tr,P)$ we denote the optimal visitation cost of $n$ robots, starting from some point $P\in \tr$. When the triangle $\tr$ is clear from the context, we abbreviate $R_n(\tr,P)$ simply by $R_n(P)$. 

In this work we are interested in determining visitation cost trade-offs with respect to different fleet sizes. In particular, for some triangle $\tr \in \trs$ (which is a compact set as a subest of $\reals^2$), and for $1\leq n<m\leq 3$, we define 
$$
\mathcal R_{n,m} (\tr) := \max_{P \in \tr} \frac{R_n(\tr,P)}{R_m(\tr,P)}.
$$
Our main technical results pertain to the study of $\mathcal R_{n,m} (\tr)$ as $\tr$ ranges over all non-obtuse triangles $\trs$. In particular, we determine $\inf_{\tr \in \trs} \mathcal R_{n,m} (\tr)$ and $\sup_{\tr \in \trs} \mathcal R_{n,m} (\tr)$ for all pairs $(n,m)\in \{(1,3),(2,3),(1,2)\}$. Our contributions are summarized in Table~\ref{tbl: main contributions}.\footnote{Note that the entries in column 1 are not obtained by multiplying the entries of columns 2,3. This is because the triangles that realize the $\inf$ and $\sup$ values are not the same in each column.} 

\begin{table}[h]
\begin{center}
\begin{tabular}{|l|ccc|}
\hline
	 		& 	$\mathcal R_{1,3}(\tr)$	& 	$\mathcal R_{2,3}(\tr)$ 	& 	$\mathcal R_{1,2}(\tr)$ \\
\hline
$\inf_{\tr \in \trs}$			&		$\sqrt{10}$	&		$\sqrt{2}$ &		$2.5$	\\
$\sup_{\tr \in \trs}$		&			$4$	&			2	&	 	$3$	\\
\hline
\end{tabular}
\caption{Our main contributions.}
\label{tbl: main contributions}
\end{center}
\end{table}

For establishing the claims above, we observe that $\inf_{\tr \in \trs} \max_{P \in \tr} \frac{R_n(\tr,P)}{R_m(\tr,P)}
= \alpha$ is equivalent to that 
$$
\forall \tr \in \trs, \exists P \in \tr, \frac{R_n(\tr,P)}{R_m(\tr,P)}
\geq \alpha 
~ \quad \text{ and } \quad ~~
\forall \epsilon >0, \exists \tr \in \trs, \forall P \in \tr, \frac{R_n(\tr,P)}{R_m(\tr,P)}
\leq \alpha+\epsilon.
$$
Similarly, $\sup_{\tr \in \trs} \max_{P\in \tr} \frac{R_n(\tr,P)}{R_m(\tr P)}
= \beta$ is equivalent to that 
$$
\forall \tr \in \trs, \forall P \in \tr, \frac{R_n(\tr,P)}{R_m(\tr,P)}
\leq \beta
 ~ \quad \text{ and } \quad ~~
\forall \epsilon >0, \exists \tr \in \trs, \exists P \in \tr, \frac{R_n(\tr,P)}{R_m(\tr,P)}
\geq \beta - \epsilon.
$$
Therefore, as a byproduct of our analysis, we also determine the best and the worst triangle cases of ratios $\mathcal R_{n,m}(\tr)$, as well as the starting points that determine these ratios. In particular we show that 
(i) the extreme values of $\mathcal R_{1,3}(\tr)$ are attained as $\tr$ ranges between ``thin'' isosceles and equilateral triangles, and the worst starting point is the incenter, 
(ii) the extreme values of $\mathcal R_{2,3}(\tr)$ are attained as $\tr$ ranges between right isosceles and equilateral triangles, and the worst starting point is again the incenter, and 
(iii) the extreme values of $\mathcal R_{1,2}(\tr)$ are attained as $\tr$ ranges between equilateral and right isosceles triangles, and the worst starting point is the middle of the shortest altitude.


\subsection{Basic Terminology \& Some Useful Observations}
\label{sec: terminology obesrvations}

The length of segment $AB$ is denoted by $\norm{AB}$.
An arbitrary non-obtuse triangle will be usually denoted by $\triangle ABC$, which we assume is of bounded size. 
More specifically, without loss of generality, we often consider $\triangle ABC$ represented in the Cartesian plane in \textit{standard analytic form}, with $A=(p,q), B=(0,0)$ and $C=(1,0)$ (certain conditions imposed on $p,q$ for the triangle to be non-obtuse and for $AC$ to be the largest edge will be invoked when necessary). The following will be used repeatedly.

\begin{observation}
\label{obs: A,I coordinates angles}
For $\triangle ABC$ in standard analytic form, where $A=(p,q)$, we have that 
\begin{equation}
\label{equa: A coordinates}
p=\frac{\cos(B)\sin(C)}{\sin(B+C)}, 
~~
q=\frac{\sin(B)\sin(C)}{\sin(B+C)}.
\end{equation}
Therefore, for the incenter $I=(p_I,q_I)$ (the intersection of angle bisectors), we have 
\begin{equation}
\label{equa: I coordinates}
p_I=\frac{\cos(B/2)\sin(C/2)}{\sin((B+C)/2)}, 
~~
q_I=\frac{\sin(B/2)\sin(C/2)}{\sin((B+C)/2)}.
\end{equation}
\end{observation}

\begin{proof}
Consider the projection $D$ of $A$ onto $BC$. We have that 
$\tan(B)=\norm{AD}/\norm{BD}$, as well as 
$\tan(C)=\norm{AD}/\norm{CD}$. Therefore
\begin{align*}
\norm{BC}
& = \norm{BD}+\norm{CD} \\
& = \norm{AD} \left( \frac1{\tan(B)}+\frac1{\tan(C)}\right) \\
& = \norm{AD} \frac{\sin(C)\cos(B)+\cos(B)\sin(B)}{\sin(B)\sin(C)} \\
& = \norm{AD} \frac{\sin(B+C)}{\sin(B)\sin(C)}.
\end{align*}
Since $\norm{BC}=1$, it follows that 
$$
q=\norm{AD}=
\frac{\sin(B)\sin(C)}{\sin(B+C)}.
$$ 
Finally, we have
$$
p=\norm{BD}
=\frac{\norm{AD}}{\tan(B)}
=
\frac{\cos(B)\sin(C)}{\sin(B+C)},
$$
and so~\eqref{equa: A coordinates} follows.
Note that~\eqref{equa: I coordinates} is obtained as an immediate corollary, since $\triangle IBC$ is in analytic form too. 
  \end{proof}

The next corollary is obtained after elementary algebraic manipulations. 
\begin{corollary}
\label{cor: incenter}
For $\triangle ABC$ in standard analytic form, its incenter $I=(p_I,q_I)$ is given by the formula 

$$
p_I= \frac{1}{2} \left(\sqrt{p^2+q^2}-\sqrt{(p-1)^2+q^2}+1\right),
~~
q_I=\frac{q}{\sqrt{p^2+q^2}+\sqrt{(p-1)^2+q^2}+1}.
$$
\end{corollary}

 The cost of optimally visiting a collection of line segments $\mathcal C$ (triangle edges) with 1 robot starting from point $P$ is denoted by $d(P,\mathcal C)$. For example, when $\mathcal C=\{AB,BC\}$ we write $d(P,\{AB,BC\})$. When, for example, $\mathcal C=\{AB\}$ is a singleton set, we slightly abuse the notation and for simplicity write $d(P,AB)$ instead of $d(P,\{AB\})$. Note that if the projection $P'$ of $P$ onto the line defined by points $A,B$ lies in segment $AB$, then $d(P,AB)=\norm{PP'}$, and otherwise $d(P,AB)=\min\{\norm{PA},\norm{PB}\}$. The following observation follows immediately from the definitions, and the fact that we restrict our study to non-obtuse triangles. 

\begin{observation}
\label{obs: cost 1,2,3}
For any non-obtuse triangle $\tr =\triangle ABC$, and $P\in \tr$, we have 
\begin{enumerate}[(i)]
\item $R_3(\tr,P) = \max\{d(P,AB), d(P,BC),d(P,CA) \}$.
\item $R_2(\tr,P) = 
\min\left\{
\begin{array}{l}
\max\{ d(P,AB), d(P,\{BC,CA\}) \} \\
\max\{ d(P,BC), d(P,\{AB,CA\}) \} \\
\max\{ d(P,CA), d(P,\{BC,AB\}) \}
\end{array}
\right\}$.
\item $R_1(\tr,P) = d(P,\{AB,BC,CA\})$. 
\end{enumerate}
\end{observation}

Motivated by our last observation, we also introduce notation for the cost of \emph{ordered visitations}. Starting from point $P$, we may need to visit an \emph{ordered list} of (2 or 3) line segments in a specific order. For example, we write $d(P,[AB,BC,AC])$ for the optimal cost of visiting the list of segments $[AB,BC,AC]$, in this order, with 1 robot. As we will be mainly concerned with $\triangle ABC$  edge visitations, and due to the already introduced standard analytic form, we refer to the trajectory realizing $d(P,[AB,BC,AC])$ as the (optimal) \emph{LDR strategy} (L for ``Left'' edge $AB$, D for ``Down'' edge $BC$, and R for ``Right'' edge AC). We introduce analogous terminology for the remaining 5 permutations of the edges, i.e.\ LRD, RLD, RDL, DRL, DLR. Note that it may happen that in an optimal ordered visitation, robot visits a vertex of the triangle edges. In such a case we interpret the visitation order of the incident edges arbitrarily. 
For ordered visitation of 2 edges, we introduce similar terminology pertaining to (optimal) LD, LR, RL, RD, DR and DL strategies. 

In order to obtain the results reported in Table~\ref{tbl: main contributions}, it is necessary to subdivide any triangle $\tr$ into sets of points that admit the same optimal ordered visitations (e.g.\ all points $P$ in which an optimal $R_1(\tr,P)$ strategy is LRD). 
For $n\in \{2,3\}$ robots, the subdivision is also with respect to the cost $R_n(\tr,P)$. Specifically for $n=2$, the subdivision is also with respect to whether the cost $R_2(\tr,P)$ is determined by the robot that is visiting one or two edges (see Observation~\ref{obs: cost 1,2,3}). 
We will refer to these subdivisions as the \emph{$R_1,R_2,R_3$ regions}. For each $n\in\{1,2,3\}$, the $R_n$ regions will be determined by 
collection (loci) of points between neighbouring regions that admit more than one optimal ordered visitations.

Angles are read counter-clockwise, so that for example for $\triangle ABC$ in standard analytic form, we have $\angle A = \angle BAC$. For aesthetic reasons, we may abuse notation and drop symbol $\angle$ from angles when we write trigonometric functions.  
Visitation trajectories will be denoted by a list of points $\langle A_1,\ldots,A_n\rangle$ ($n\geq 2$), indicating a movement along line segments between consecutive points. Hence, the cost of such trajectory would be $\sum_{i=2}^n\norm{A_{i}A_{i-1}}$.


\section{Preliminary Results}
\label{sec: preliminaries}

\subsection{Optimal Visitations of Two Triangle Edges}
\label{sec: visiting 1,2 edges}

We consider the simpler problem of visiting two distinguished edges of a triangle $\tr=ABC$, starting from a point within the triangle. The following preliminary observations will be useful, and the reader may refer to Figures~\ref{fig: Visit2Edges-i} and~\ref{fig: Visit2Edges-ii}. 
\begin{figure}
\centering
\includegraphics[width=5.5cm]{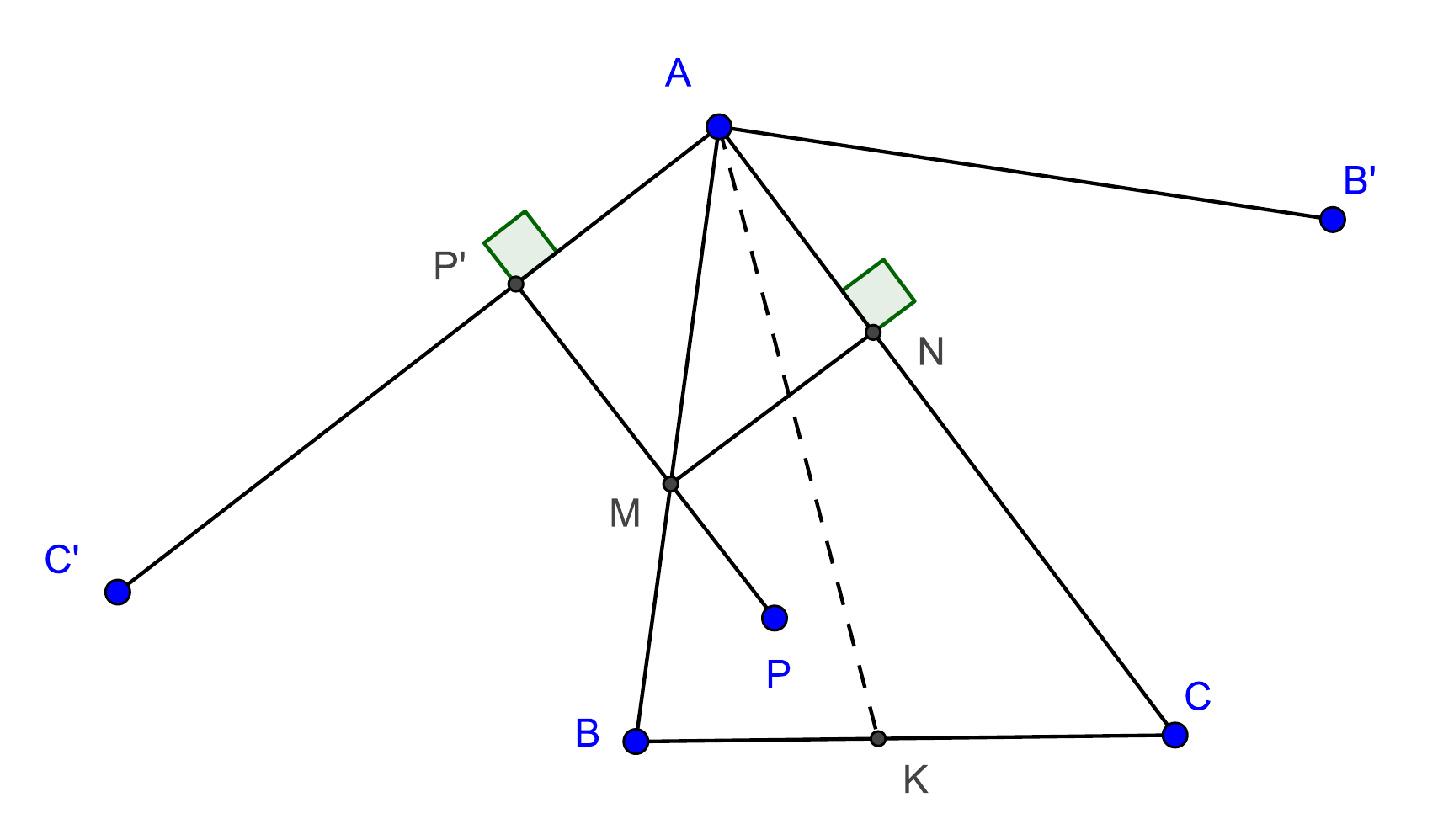}
\caption{
Optimal trajectory for visiting $\{AB,AC\}$ for $\angle A\leq \pi/3$.}
\label{fig: Visit2Edges-i}
\end{figure}
\begin{figure}
\centering
\includegraphics[width=7cm]{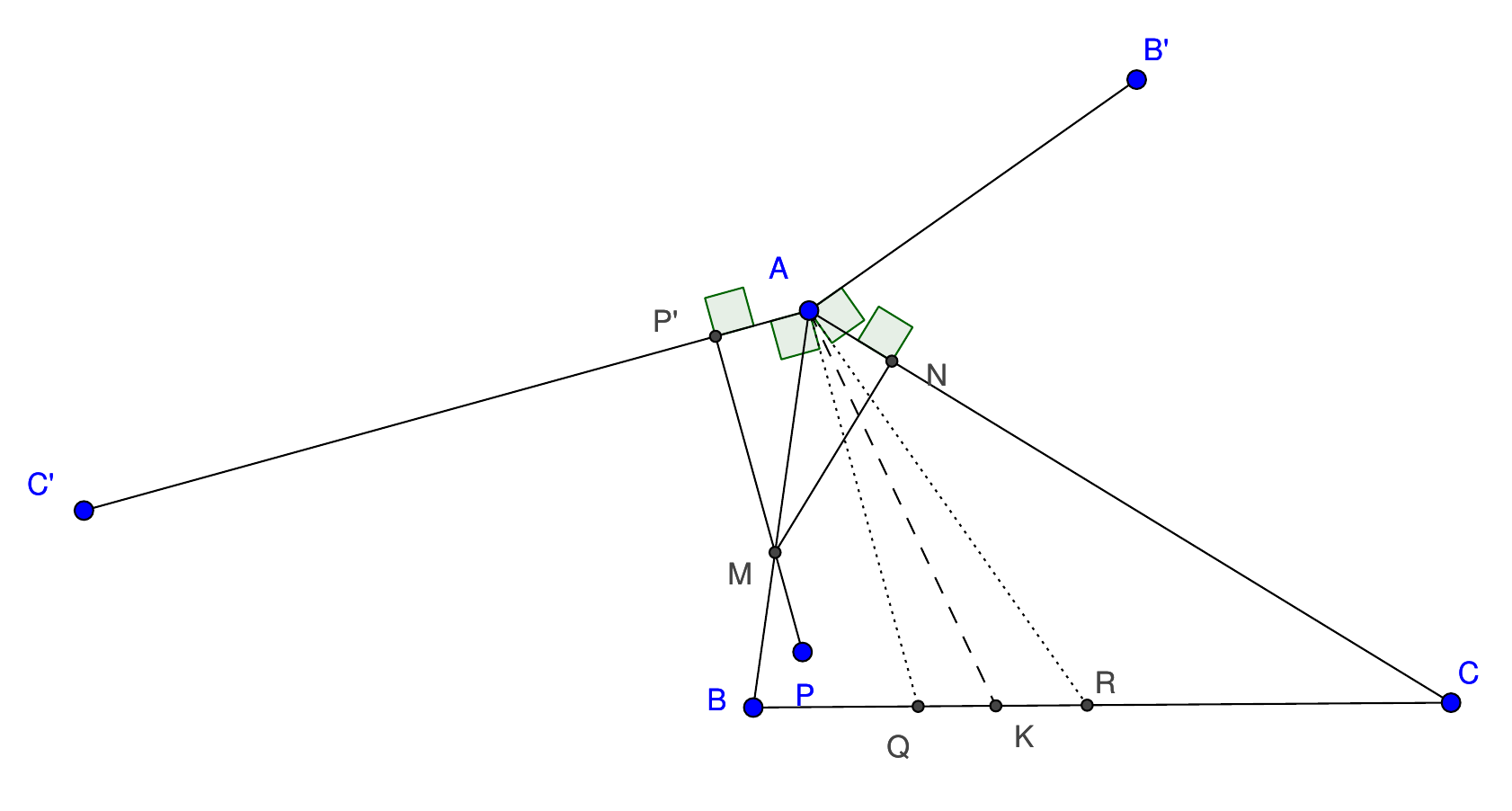}
\caption{
Optimal trajectory for visiting $\{AB,AC\}$ for $\angle A\geq \pi/3$, starting outside the optimal bouncing subcone.}
\label{fig: Visit2Edges-ii}
\end{figure}
Consider some $P\in \tr$,
and let $K$ be on $BC$ such that $AK$ is the angle bisector of $A$. Any point on $AK$ is equidistant from $AB,AC$. Moreover, for any $P\in ABK$ visiting $AB$ is not more costly than visiting $AC$. 

Now we consider the problem of visiting $AB,AC$ starting from $P\in \tr$. 
Let $C',B'$ be the reflections of $C,B$ around $AB,AC$, respectively. 
Clearly, 
\begin{equation}
\label{equa: cost 2 edges}
d(P,\{AB,AC\}) = \min\{ d(P,AC'), d(P,AB')\}.
\end{equation}
Since, in particular, $AK$ is also the angle bisector of $C'AB'$, we conclude that if $P \in ABK$, then $d(P,\{AB,AC\})$ is determined by visiting $AB$ no later than $AC$, that is, 
$$d(P,\{AB,AC\})=d(P,[AB,AC]).$$
In the remaining of this section, we fix such a $P$.

For the specifics of the optimal trajectory, we need some additional terminology. 
When $\angle A \geq \pi/3$, we define the concept of its \emph{optimal bouncing subcone}, which is defined as a cone of angle $3\angle A - \pi$ and tip $A$, so that $\angle A$ and the subcone have the same angle bisector. When $\angle A=\pi/3$, then the optimal bouncing subcone is a ray with tip $A$ that coincides with the angle bisector of $\angle A$.  Whenever $\angle A < \pi/3$ we define its optimal bouncing subcone as the degenerate empty cone. 

\begin{observation}
\label{obs: Visit2Edges-ii}
If $P$ is in the optimal bouncing subcone of $\angle A$, then 
$d(P,\{AB,AC\}) = \norm{PA}$. 
\end{observation}

\begin{proof}
Consider a line passing through $A$ that is perpendicular to $AC'$ that intersects $BC$ at $Q$ (see Figure~\ref{fig: Visit2Edges-ii}). Consider also a line passing through $A$ that is perpendicular to $AB'$ that intersects $BC$ at $R$. Then, the cone with tip $A$ and angle $\angle QAR$ is the optimal bouncing subcone of $\angle A$. 
Let $P$ be a point within the subcone. 
By construction, the projection of $P$ onto the line defined by points $A,C'$ falls outside the line segment $AC'$, and similarly for points $A,B'$. 
Therefore, for any point $P$ within the subcone, we have that $d(P,AB')=d(P,AC')=\norm{PA}$, so combined with~\eqref{equa: cost 2 edges}, the claim follows. 
  \end{proof}

For a point $P\in \triangle ABC$ outside the optimal bouncing subcone of $\angle A$, we define the (two) \emph{optimal bouncing points $M,N$ of the ordered $[AB,AC]$ visitation} as follows. Let $P'$ be the projection of $P$ onto $AC'$. Then, $M$ is the intersection of $PP'$ with $AB$ and $N$ is the projection of $M$ onto $AC$. Note that equivalently, $M,N$ are determined uniquely by requiring that (i) $\angle BMP = \angle NMA$, and (ii) $\angle ANM = \pi/2$.

\begin{observation}
\label{obs: Visit2Edges-i}
If $P$ is outside the optimal bouncing subcone of $\angle A$, then 
$d(P,\{AB,AC\}) = \norm{PM}+\norm{MN}$, where $M,N$ are the optimal bouncing points of ordered $[AB,AC]$ visitation. 
\end{observation} 

\begin{proof}
Since $P$ in $\triangle ABK$, we have that $d(P,\{AB,AC\})=d(P,[AB,AC])$. Also by~\eqref{equa: cost 2 edges}, we have that $d(P,[AB,AC])=\norm{PP'}$. The claim follows by noticing that $\norm{MP'}=\norm{MN}$.
  \end{proof}

\subsection{Optimal (Ordered) Visitation of Three Triangle Edges}
\label{sec: visiting 3 edges}

In this section we discuss optimal LRD visitations of non-obtuse $\triangle ABC$, together with optimality conditions (recall that optimality refers to the cost incured by one robot visiting all edges). 
Optimality conditions for the remaining 5 ordered visitations are obtained similarly. 
In order to determine the optimal LRD visititation, we obtain reflection $C'$ of $C$ across $AB$, and reflection $B'$ of $B$ across $C'A$, see also Figure~\ref{fig: 3edgesPredereminedOrderSmallC}.  
\begin{figure}
\centering
\includegraphics[width=4cm]{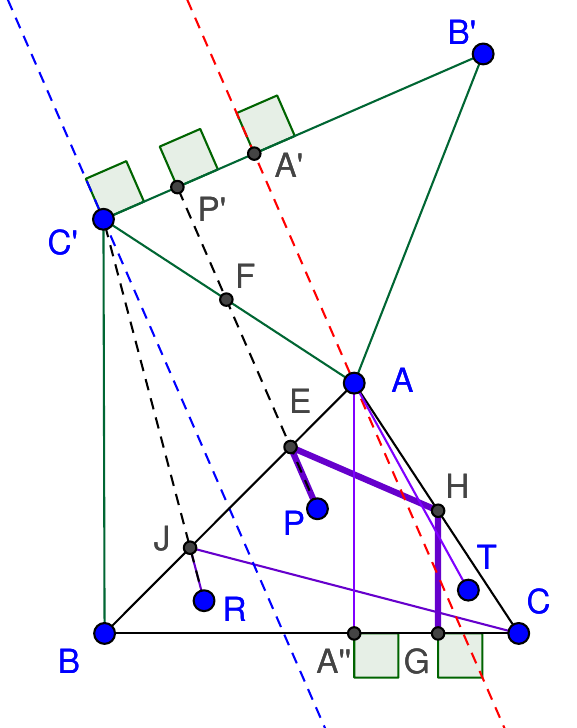}
\caption{
Arbitrary non-obtuse $\triangle ABC$ shown with its LRD bounce indicator line (blue dotted line) and its LRD subopt indicator line (red dotted line).}
\label{fig: 3edgesPredereminedOrderSmallC}
\end{figure}

From $C'$ and $A$, we draw a lines $\epsilon, \zeta$, both perpendicular to $C'B'$ which may (or may not) intersect $\triangle ABC$. 
We refer to line $\epsilon$ as the \emph{LRD bounce indicator line}.
We also refer to line $\zeta$ as the \emph{LRD subopt indicator line}.
Each of the lines identify a halfspace on the plane.
The halfspace associated with $\epsilon$ on the side of vertex $A$ will be called the \emph{positive halfspace of the $LRD$ bounce indicator line} or in short the \emph{positive LRD bounce halfspace}, and its complement will be called the \emph{negative LRD bounce halfspace}. 
The halfspace associated with $\zeta$ on the side of vertex $B$ will be called the \emph{positive halfspace of the $LRD$ subopt indicator line}, or in short the \emph{positive LRD subopt halfspace}, and its complement will be called the \emph{negative LRD subopt halfspace}. 

For a point $P$ in the positive LRD bounce and subopt halfspaces, let $P'$ be its projection onto $C'B'$. Let $E,F$ be the intersections of $PP'$ with $AB,AC'$, respectively. Let also $H$ be the reflection of $F$ across $AB$, and let $G$ be the projection of $H$ onto $BC$. 
Points $E,H,G$ will be called the \emph{optimal LRD bouncing points} for point $P$. The points are also uniquely determined by requiring that $\angle BEP = \angle HEA$ and that $HG$ is perpendicular to $BC$. 
For a point $R$ in the negative LRD bounce halfspace and in the positive subopt halfspace, let $J$ be the intersection of $RC'$ with $AB$. 
Point $J$ will be called the \emph{degenerate optimal LRD bouncing point}, which is also uniquely determined by the similar bouncing rule $\angle BJR = \angle CJA$.
Finally, let $A',A''$ be the projection of $A$ onto $B'C',BC$, respectively. 

The next lemma refers to such points $P,R$ together with the construction of Figure~\ref{fig: 3edgesPredereminedOrderSmallC}. Its proof follows immediately by noticing that 
the optimal LRD visitation is in 1-1 correspondence with the optimal visitation of segment $B'C'$ using a trajectory that passes from segment $AB$. 
\begin{lemma}
\label{lem: +bounce +subopt ordered visitiation}
The optimal LRD visitation trajectory, with starting points $P,R,T$, is: 
\begin{itemize}
\item Trajectory $\langle P,E,H,G\rangle$, provided that $P$ is in the positive LRD bounce and subopt halfspaces,
\item Trajectory $\langle R,J,C \rangle$, provided that $R$ is in the negative LRD bounce halfspace and in the positive subopt halfspace, 
\item Trajectory $\langle T,A,A''\rangle$, provided that $T$ is in the negative LRD subopt halfspace (see Figure~\ref{fig: 3edgesPredereminedOrderSmallC}). 
\end{itemize}
\end{lemma}


\section{Computing the $R_n$ Regions, $n=1,2,3$}
\label{sec: Rn regions}

By Observations~\ref{obs: Visit2Edges-ii},~\ref{obs: Visit2Edges-i} and Lemma~\ref{lem: +bounce +subopt ordered visitiation}, we see that optimal visitations of 2 or 3 edges have cost equal to (i) the distance of the starting point to a line (reflection of some triangle edge), or (ii) the distance of the starting point to some point (triangle vertex) or (iii) the distance of the starting point to some triangle vertex plus 
the length of some triangle altitude. 
In this section we describe the $R_n$ regions of certain triangles, $n \in \{1,2,3\}$. For this, we compare optimal ordered strategies, and the subdivisions of the regions are determined by loci of points that induce ordered trajectories of the same cost. As these costs are of type (i), (ii), or (iii) above (and considering all their combinations) the loci of points in which two ordered strategies have the same cost will be either some line (line bisector or angle bisector), or some conic section (parabola or hyperbola).


\subsection{Triangle Visitation with 3 Robots - The $R_3$ Regions}
\label{sec: regions 3}

Consider $\tr \in \trs$ with vertices $A,B,C$. For every $P\in \tr$, any trajectories require time at least the maximum distance of $P$ from all edges, in order to visit all of them. This bound is achieved by having all robots moving along the projection of $P$ onto the 3 edges, and so we have 
$$
R_3(P) = \max \{ d(P,AB),d(P,BC), d(P,CA) \},
$$
as also in Observation~\ref{obs: cost 1,2,3}.
Next we show how to subdivide the region of $\tr$ with respect to  which of the 3 projections is responsible for the optimal visitation cost. For this, we let $I$ denote the incenter (the intersection of angle bisectors) of $\tr$. Let also $K,L,M$ be the intersections of the bisectors with edges $BC, CA$ and $AB$, respectively, see also Figure~\ref{fig: R3 regions}. 

 \begin{figure}[h!]
\centering
  \includegraphics[width=5.5cm]{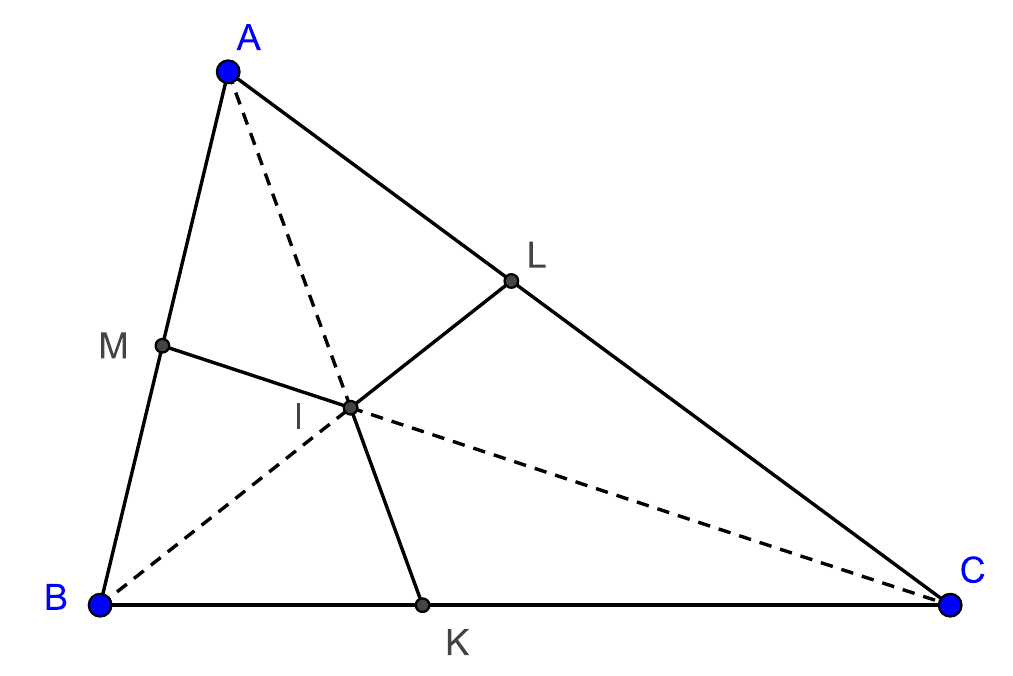}
\caption{
The $R_3$ regions of an arbitrary non-obtuse $\triangle ABC$. $AK, BL, CM$ are the angle bisectors of $\angle A, \angle B, \angle C$, respectively. Recall that the incenter $I$ is equidistant from all triangle edges.}
\label{fig: R3 regions}
\end{figure}

\begin{lemma}
\label{lem: r3 regions}
For every starting point $P \in \tr$, we have that 
$$
R_3(\tr,P) =
\left\{
\begin{array}{ll}
d(P,AB)& \mbox{, provided that } P\in  CLIK\\
d(P,BC)& \mbox{, provided that } P\in  AMIL\\
d(P,CA)& \mbox{, provided that } P\in  BKIM
\end{array}
\right..
$$
\end{lemma}

\begin{proof}
The subdivision of the $R_3$ regions are determined by the regions' separators, i.e. as per loci of points that are equidistant from the most distant edge. Since also the angle bisector is the locus of points that are equidistant from the lines forming the angle, the lemma follows. 
  \end{proof}



\subsection{Triangle Visitation with 2 Robots - The $R_2$ Regions}
\label{sec: regions 2}

In this section we show how to subdivide the region of any non-obtuse triangle $\tr \in \trs$ into subregion with respect to the optimal trajectories and their costs, for a fleet consisting of 2 robots. The following lemma describes a geometric construction. 

\begin{lemma}
\label{lem: angle ABC R2 divisor}
Consider non-obtuse $\triangle ABC$ along with its incenter $I$. Let $K,M$ be the intersections of angle bisectors of $A,C$ with segments $BC,AB$ respectively.
From $K,M$ we consider cones of angles $A,C$ respectively, having direction toward the interior of the triangle, and placed so that their bisectors are perpendicular to $BC,AB$, respectively. Then, the extreme rays of the cones intersect in line segment $BI$.
\end{lemma}

\begin{proof}
Consider non-obtuse triangle $\tr=ABC$ and angle $A$ bisector $AK$, where $K \in BC$. Let $A'$ be the reflection of $A$ across $BC$. Without loss of generality we may assume that $\angle B \geq \angle C$ or, in other words, that $\angle AKB \leq \pi/2$. 
First we consider the case that $\angle B > \angle C$, see Figure~\ref{fig: Cost1and2equalizer}, and in particular that $\triangle ABC$ is not isosceles (the case of isosceles triangle is much easier and can be treated similarly). 
Then
$$\angle A'BA + \angle BAC = 2 \angle B + \angle A > \angle C + \angle B + \angle A = \pi,$$ 
and therefore $BA'$ and $AC$ are not parallel. Moreover the extensions of these line segments (on the directions of $B,A$, respectively) meet at a point, call it $D$. 
\begin{figure}[h!]
\centering
  \includegraphics[width=12cm]{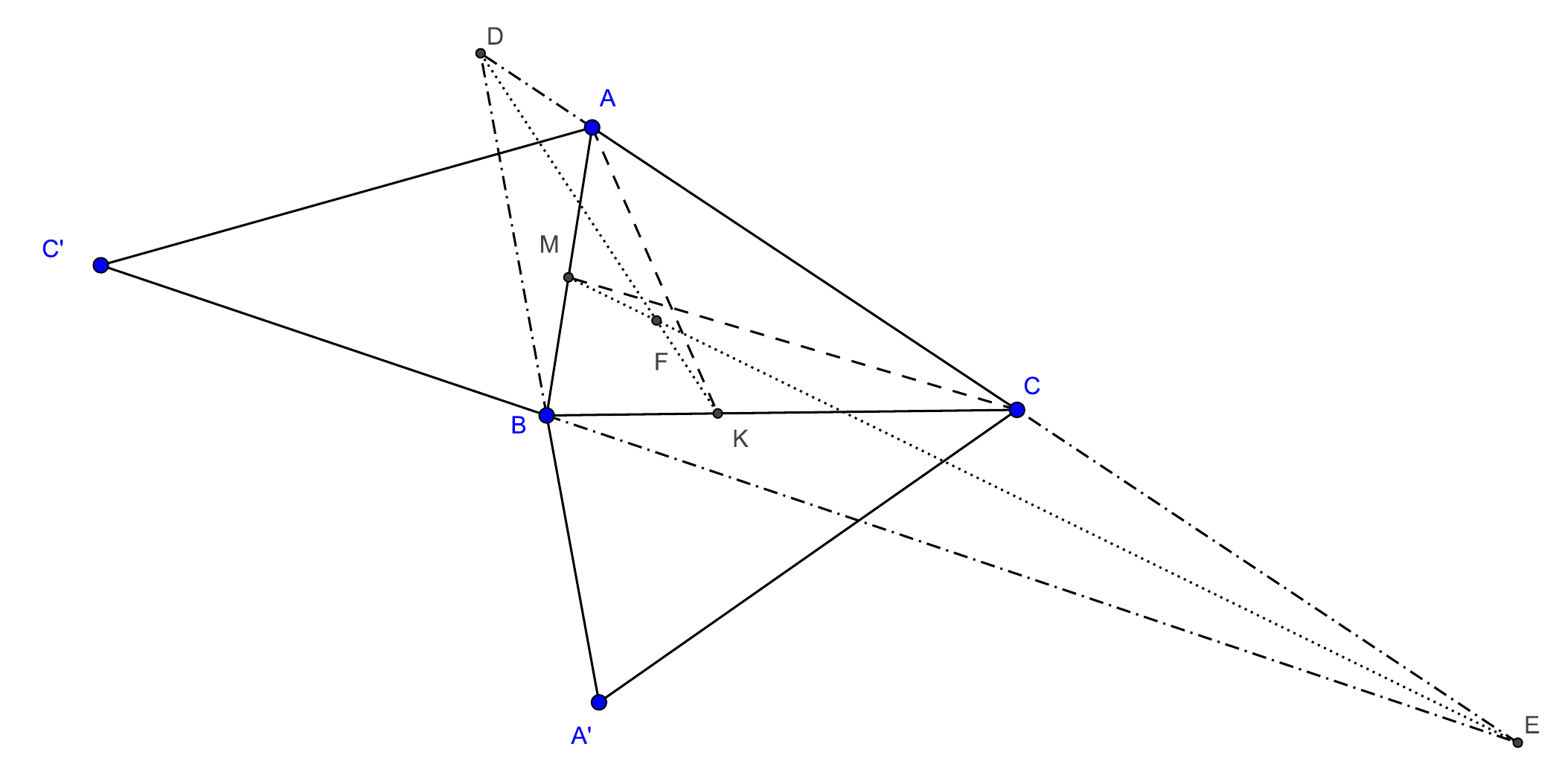}
\caption{
The case of $\angle B > \angle C$, in the proof of Lemma~\ref{lem: angle ABC R2 divisor}.
}
\label{fig: Cost1and2equalizer}
\end{figure}

Next we claim that $DK$ is the angle bisector of $\angle A'DC$. 
To prove this, it suffices to show that $K$ is the incenter of $\triangle DA'C$.
Indeed, $K$ lies on the bisector of $\angle ACA'$, since $\triangle A'BC$ was the result of the reflection of $ABC$ across $BC$. For the same reason, $K$ lies on the bisector of $\angle CA'B$ (which is an equivalent claim to that $AK$ is the bisector of $\angle BAC$). 
Hence, $K$ is the intersection of two bisectors of $\triangle DA'C$, as promised.


Next we denote by $\angle A, \angle B, \angle C$ the angles of $\triangle ABC$. 
We have the following claim: 
$\angle DKB = \pi/2-\angle A/2$. 
Indeed, 
\begin{align*}
\angle DKB
& = \pi - \angle KBD - \angle BDK \\
&= \pi - (\pi-\angle A'BC) - \angle A'DC/2 \\
&= \angle B - (\pi - \angle CA'B - \angle DCA')/2 \\
&=\angle B - \pi/2 + \angle A/2 + \angle C \\
&=\pi - \angle A - \pi/2 + \angle A/2 \\
& = \pi /2 - \angle A/2. 
\end{align*}
In particular, this shows that $KD$ is an extreme ray of a cone with tip $K$ and angle $\angle A$, having direction toward the interior of the triangle, and placed so that its bisector is perpendicular to $BC$.

Similarly, consider the reflection of $C$ across $AB$. Since $\angle A \not = \angle C$, the lines passing through pairs $C',B$ and $A,C$ are not parallel. 
So their extensions meet at some point, call it $E$, in the directions of points $B,C$ respectively. Let also $CM$ be the angle bisector of $\angle C$, where $M\in AB$. Exactly as before, $EM$ is the angle bisector of $\angle AEB$, and hence all points on $EM$ are equidistant from $EA$ and $EB$. 
It follows that $EM$ and $DK$ intersect at the incenter of $\triangle BED$, call it $F$. 

The above argument proves that $F$ is the intersection of two of the extreme rays of the two cones with tips $K,M$ and angles $A,C$, as per the description of the lemma.  It remains to prove that $F$ lies on the bisector of angle $B$. 
Indeed, $F$ is on the bisector of $\angle DEB$ and on the bisector of $\angle EDB$. Therefore, $F$ is the incenter of $\triangle BDE$. In particular, $F$ should lie on the bisector of $\angle DBE$. Note however that $\angle DBE$ and $\angle ABC$ have the same bisector, because $\angle CBE=\angle ABD=\pi-2 \angle B$.
  \end{proof}

Motivated by Lemma~\ref{lem: angle ABC R2 divisor}, we will be referring to the subject point $F$ in the line segment $BI$ as the \emph{separator of the angle $B$ bisector}.
Similarly, we obtain separators $J,H$ of angles $C,A$ bisectors, respectively, see also Figure~\ref{fig: R2triangleNEW}.
  \begin{figure}
\centering
\includegraphics[width=6cm]{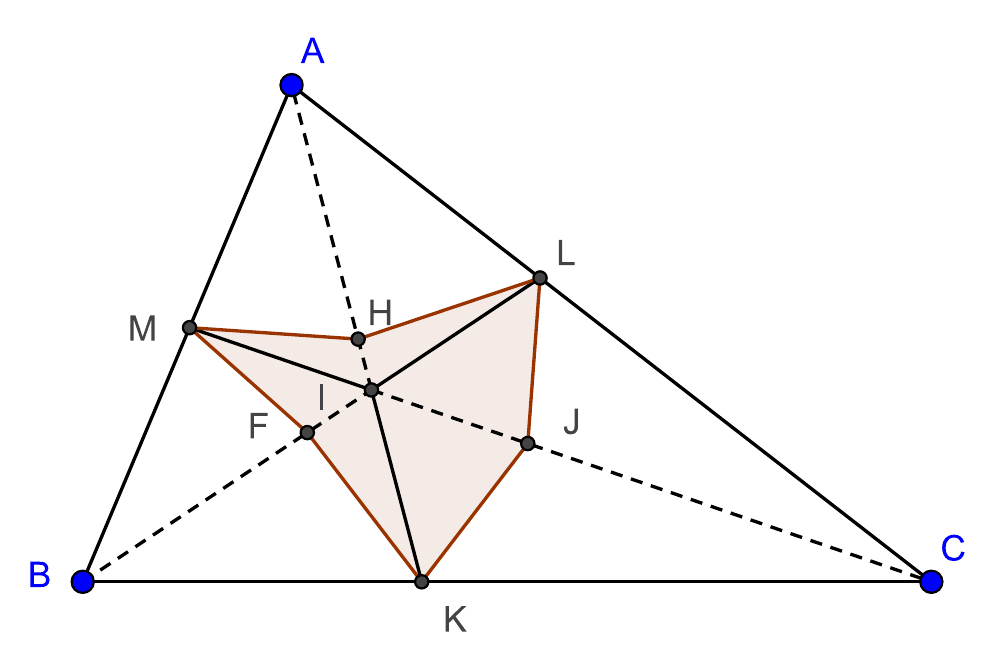}
\caption{
The $R_2$ hexagon separator of $\triangle ABC$.
}
\label{fig: R2triangleNEW}
\end{figure}
In what follows, we will be referring to the (possibly non-convex) hexagon $MFKJLH$ as the \emph{$R_2$ hexagon separator of $\triangle ABC$}. 

From the proof of Lemma~\ref{lem: angle ABC R2 divisor} we also derive a useful observation. Referring to Figure~\ref{fig: Cost1and2equalizer}, recall that $DK$ is the angle bisector of $\angle A'DC$. Therefore all points on $DK$ are equidistant from $DA'$ and $AC$. 
At the same time the distance of any point $P$ on $FK$ 
to $DA'$ equals $d(P,\{AB,BC\})=d(P,[BC,AB])$ 
unless $d(P,\{AB,BC\}) = \norm{PB}$, that is, unless the optimal strategy for visiting $\{AB,BC\}$ would be to move directly to vertex $B$ (which happens only if $\angle B \geq \pi/3$ and $P$ lies within the optimal bouncing subcone of angle $B$, see Observation~\ref{obs: Visit2Edges-ii}). 

The remaining of the section refers to non-obtuse triangle $\tr=\triangle ABC$ as in Figure~\ref{fig: R2triangleNEW}, where in particular 
$MFKJLH$ is the $R_2$ separator of $\tr$.
We are now in a position to make a preliminary observation which will be generalized soon. 
Suppose that $\angle B \leq \pi/3$. Then for every point $P$ either on $MF$ or on $FK$, we have that 
$d(P, AC) = d(P, \{BC,AB\})$. Moreover, for every $P$ within tetragon $BKFM$, we have that $R_2(P) = d(P, AC)$. 

The previous observation can be extended to larger angles. Assume that $\angle B \geq \pi/3$. Then, the same reasoning shows that for every point $P$ either on $MF$ or on $FK$, which are outside the optimal bouncing subcone of angle $B$, we have that 
$d(P, AC) = d(P, \{BC,AB\})$. For points within the subcone, the optimal trajectory to visit $\{BC,AB\}$ would be to go directly to $B$. So for points $P$ within the optimal bouncing subcone, condition $d(P, AC) = d(P, \{BC,AB\})$ translates into that $P$ is equidistant from $AC$ and $B$. Hence, $P$ lies in a parabola with $AC$ being the directrix and $B$ being the focus. Next, we refer to that parabola as the \emph{separating parabola of $B$}. 

Motivated by the previous observation, we introduce the notion of the \emph{refined $R_2$ mixed-hexagon separator of triangle $\tr$} as follows. For every angle of $\tr$ which is more than $\pi/3$, we replace the portion of the $R_2$ hexagon separator within the optimal bouncing subcone of the same angle by the corresponding separating parabola. In Figure~\ref{fig: R2triangleNEWrefined} we display an example where only one angle is more than $\pi/3$. 
\begin{figure}
\centering
\includegraphics[width=6cm]{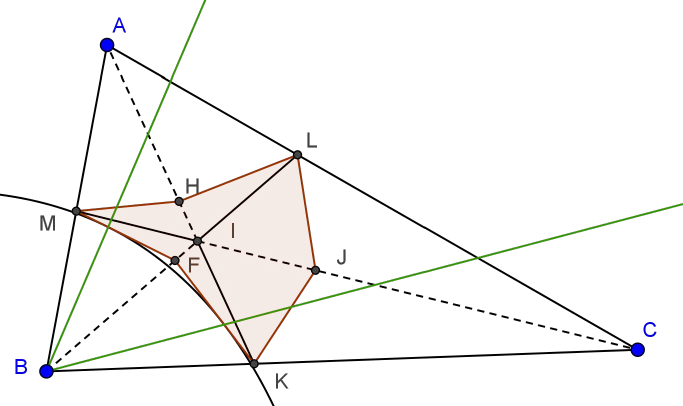}
\caption{
The refined $R_2$ mixed-hexagon separator of $\triangle ABC$, where $\angle B > \pi/3$. 
}
\label{fig: R2triangleNEWrefined}
\end{figure}
Combined with Observation~\ref{obs: cost 1,2,3} (ii), we can formalize our findings as follows. 

\begin{lemma}
\label{lem: R2 region partial}
For every starting point on the boundary of the refined $R_2$ mixed-hexagon separator of a triangle $\tr$, the cost of visiting only the opposite edge equals the cost of visiting the other two edges. For every starting point $P$ outside the $R_2$ separator, $R_2(P)$ equals the distance of $P$ to the opposite edge. 
Moreover, for every starting point $P$ in the interior of the refined $R_2$ separator, $R_2(P)$ is determined by the cost of visiting two of the edges of $\tr$. 
\end{lemma}

Lemma~\ref{lem: R2 region partial} implies the following corollaries pertaining to specific triangles $\triangle ABC$. In both statements, and the associated figures, $I$ is the incenter of the triangles, and points $K,L,M$ are defined as in Figure~\ref{fig: R3 regions}. 
 
\begin{corollary}[Hexagon separator of equilateral triangle]
\label{cor: r2 hexagon equilateral}
Consider equilateral $\tr=\triangle ABC$, see Figure~\ref{fig: R2 regions equilateral}. 
\begin{figure}[h!]
\centering
\includegraphics[width=6cm]{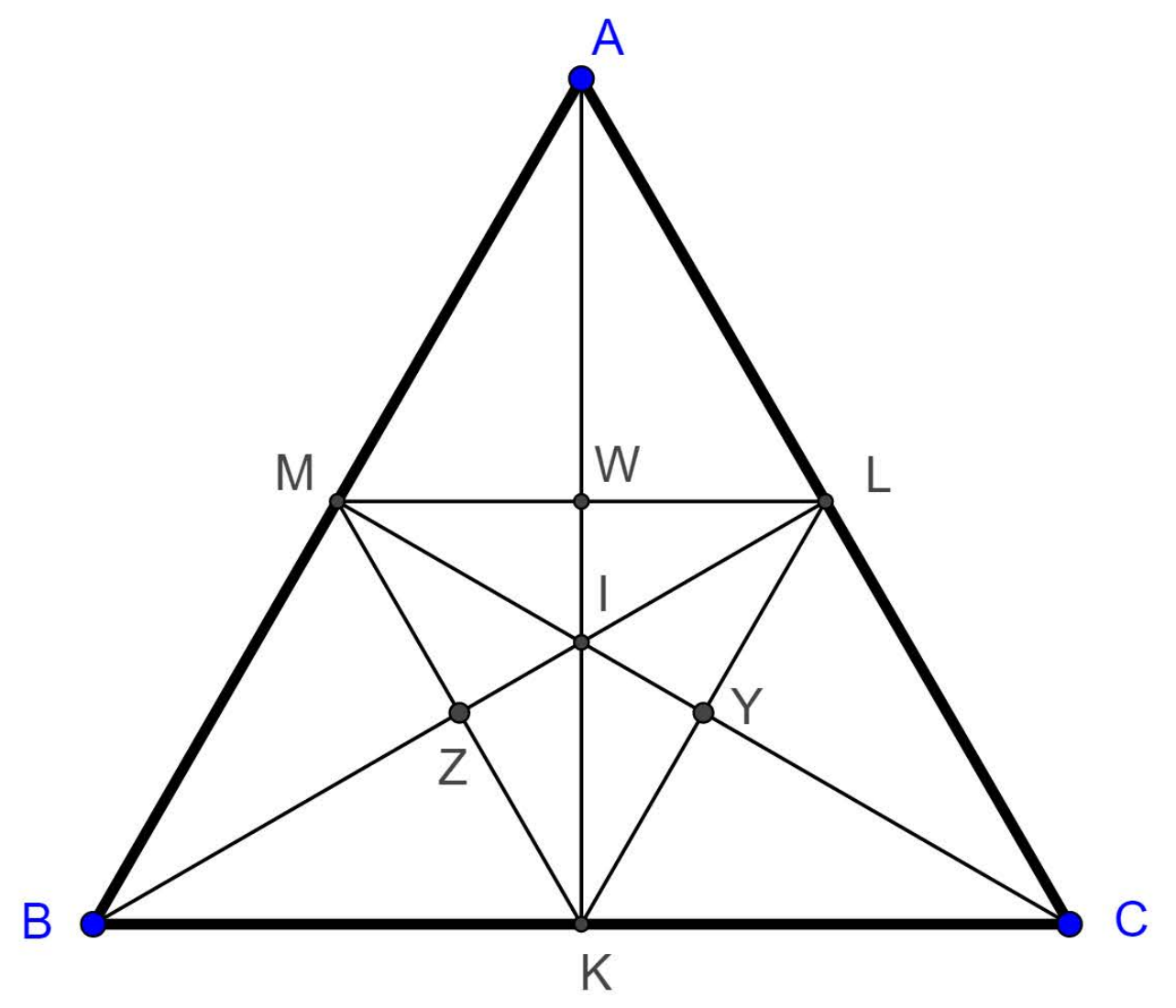}
\caption{
The $R_2$ regions of the equilateral triangle, see also Corollary~\ref{cor: r2 hexagon equilateral} and Corollary~\ref{cor: r2 equilateral within} for detailed description. 
}
\label{fig: R2 regions equilateral}
\end{figure}
Let $W,Z,Y$ be the intersections of $AK,BL,CM$ respectively (also the separators of angle $A$ bisector, angle $B$ bisector, and angle $C$ bisector, respectively). Then, the $R_2$ hexagon separator of $\tr$ is $MZKYLW$, which is also triangle $MKL$. More specifically, for all $P \in \triangle AML$, we have that $R_2(\tr,P)=d(P,BC)$.
\end{corollary}
By symmetry, one can derive $R_2(\tr,P)$ for all starting points $P$ outside the hexagon separator of isosceles $\tr$. 

\begin{corollary}[Mixed-hexagon separator of right isosceles]
\label{cor: r2 hexagon right isosceles}
Consider right isosceles $\tr=\triangle ABC$, see Figure~\ref{fig: R2 regions right isosceles}.
\begin{figure}
\centering
\includegraphics[width=5cm]{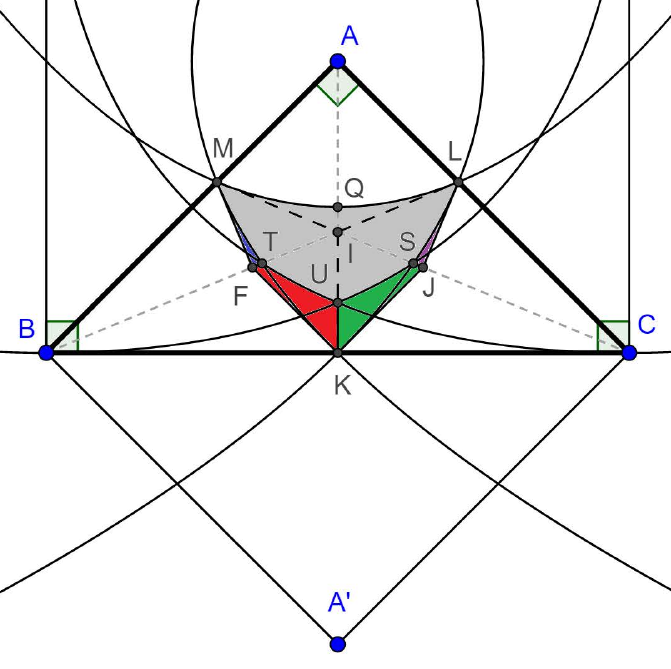}
\caption{
The $R_2$ regions of the right isosceles, see also Corollary~\ref{cor: r2 hexagon right isosceles} and Lemma~\ref{lem: r2 rightisosceles within} for detailed description. The coloured region identifies the refined $R_2$ mixed-hexagon separator.
}
\label{fig: R2 regions right isosceles}
\end{figure}
The separator of angle $A$ bisector is incenter $I$. Let also $F,J$ be the separators of angle $B$ bisector and angle $C$ bisector, respectively. Then, the $R_2$ hexagon separator of $\tr$ is $IMFKJL$. The parabola with directix $BC$ and focus $A$, intersecting $AK$ at $Q$ and passing through $M,L$ is the separating parabola of $A$. Hence, for every point $P \in \tr$ above the parabola, we have $R_2(\tr,P)=d(P,BC)$, as well as for every point $X$ in tetragon $MBKF$, we have $R_2(\tr,X)=d(X,AC)$. The remaining case of starting points in tetragon $KCLJ$ follows by symmetry.
\end{corollary}


Describing the subdivisions within the refined $R_2$ mixed-hexagon separator for arbitrary triangles is a challenging task. 
On the other hand, by Observation~\ref{obs: cost 1,2,3} (ii) and Lemma~\ref{lem: R2 region partial} the cost within the separator is determined by the cost of visiting just two edges. Also, by Observations~\ref{obs: Visit2Edges-ii},~\ref{obs: Visit2Edges-i} the cost of such visitation can be described either as a distance to a line or to a point.
We conclude that, within the $R_2$ separator, the subdivisions are determined by separators that are either parts of lines or parabolas (loci of points for which the cost of visiting some two edges are equal). 
Hence, for any fixed triangle, an extensive case analysis pertaining to pairwise comparisons of visitations costs can determine all $R_2$ subdivisions (and the challenging ones are within the refined separator). In what follows we summarize formally the subdivisions only of two triangle types, focusing on the visitation cost of all starting points within the (refined) hexagon separators. 

\begin{corollary}[$R_2$ regions of an equilateral triangle]
\label{cor: r2 equilateral within}
Consider equilateral $\tr=\triangle ABC$, as in Corollary~\ref{cor: r2 hexagon equilateral}, see Figure~\ref{fig: R2 regions equilateral}. Then for every starting point $P\in \triangle MWI$, we have that $$R_2(\tr,P)=d(P,[AB,AC]).$$
The remaining cases of starting points within the hexagon separator $MZKYLW$ follow by symmetry. 
\end{corollary}

\begin{lemma}[$R_2$ regions of a right isosceles]
\label{lem: r2 rightisosceles within}
Consider right isosceles $\tr=\triangle ABC$, as in Corollary~\ref{cor: r2 hexagon right isosceles}, see Figure~\ref{fig: R2 regions right isosceles}. 
Consider parabola with directrix the line passing through $B$ that is perpendicular to $BC$ (also the reflection of $BC$ across $AB$) and focus $A$, passing through $M,K$ and intersecting $BL$ at point $T$ (define also $S$ as the symmetric point of $T$ across $AK$). That parabola is the locus of points $P$ for which $\norm{PA}=d(P,[AB,BC])$. 
Let also $A'$ be the reflection of $A$ across $BC$. 
Consider parabola with directrix $BA'$ and focus $A$, passing through $T$ and intersecting $AK$ at point $U$. That parabola is the locus of points $P$ for which $\norm{PA}=d(P,[BC,AB])$. 
Therefore, if $P$ is a starting visitation point, we have that:
\begin{itemize}
\item $R_2(\tr,P)=\norm{PA}$, for all $P$ in mixed closed shape $MTUSLQ$ (grey shape in Figure~\ref{fig: R2 regions right isosceles}), 
\item $R_2(\tr,P)=d(P,[AB,BC])$, for all $P$ in mixed closed shape $MFT$ (blue shape in Figure~\ref{fig: R2 regions right isosceles}), 
\item $R_2(\tr,P)=d(P,[AC,BC])$, for all $P$ in mixed closed shape $LJS$ (purple shape in Figure~\ref{fig: R2 regions right isosceles}), 
\item $R_2(\tr,P)=d(P,[BC,AB])$, for all $P$ in mixed closed shape $FKUT$ (red shape in Figure~\ref{fig: R2 regions right isosceles}). 
\item $R_2(\tr,P)=d(P,[BC,AC])$, for all $P$ in mixed closed shape $JSUK$ (green shape in Figure~\ref{fig: R2 regions right isosceles}). 
\end{itemize}
\end{lemma}

\begin{proof}
For every starting point $P$ within the refined $R_2$ mixed-hexagon separator, the cost $R_2(\tr,P)$ is determined by the visitation cost of two edges. 

Note that the optimal bouncing subcone of $\angle A$ is one with extreme rays $AB,AC$ (i.e. the angle itself). 
Hence, by Observation~\ref{obs: Visit2Edges-ii}, for every $P$ in mixed closed shape $MTUSLQ$, we have that $d(P,\{AB,AC\}=\norm{PA}$. Therefore, by the definitions of the separating parabolas, we have that $R_2(\tr,P)=\norm{PA}$. 

Consider now some starting point $P$ in mixed closed shape $MFT$. By the definition of the separating parabolas, we have that $$R_2(\tr,P)=d(P, \{AB,BC\})=d(P, [AB,BC]),$$ where the last equality follows since $P$ is not below line segment $BL$, the angle bisector of $\angle B$ (and $d(P, [AB,BC])$ can be computed as the distance of $P$ to the reflection of $BC$ across $AB$). 

Finally, consider some starting point $P$ in mixed closed shape $FKUT$. By the definition of the separating parabolas, we have that 
$$R_2(\tr,P)=d(P, \{AB,BC\})=d(P, [BC,AB]),$$ where the last equality follows since $P$ is not above line segment $BL$, the angle bisector of $\angle B$. 
  \end{proof}

\subsection{Triangle Visitation with 1 Robot - The $R_1$ Regions}
\label{sec: regions 1}

In this section we show how to partition the region of an arbitrary non-obtuse $\triangle ABC$ into  sets of points $P$ with respect to the optimal strategy of $R_1(P)$. There are 6 possible visitation strategies for $d(P, \{AB,AC,BC\})$, one for each permutation of the edges indicating the order they are visited (ordered visitations). Clearly, it is enough to describe, for each two ordered visitations, the borderline (separator) of points in which the two visitations have the same cost. 
By Lemma~\ref{lem: +bounce +subopt ordered visitiation}, any such ordered visitation cost is the distance of the starting point either to a point, or to a line, or a distance to a line plus the length of some altitude. Since the $R_1$ regions are determined by separators, i.e. loci of points in which different ordered visitations induce the same costs, it follows that these separators are either lines, or conic sections. 
Therefore, by exhaustively pairwise-comparing all ordered visitations along with their separators, we can determine the $R_1$ regions of any triangle. Next, we explicitly describe the $R_1$ regions only for three types of triangles that we will need for our main results. For the sake of avoiding redundancies, we omit any descriptions that are implied by symmetries. 

The next lemma describes the $R_1$ regions of an equilateral triangle, as in Figure~\ref{fig: R1regionsEquilateral}.
\begin{figure}
\centering
\includegraphics[width=2in]{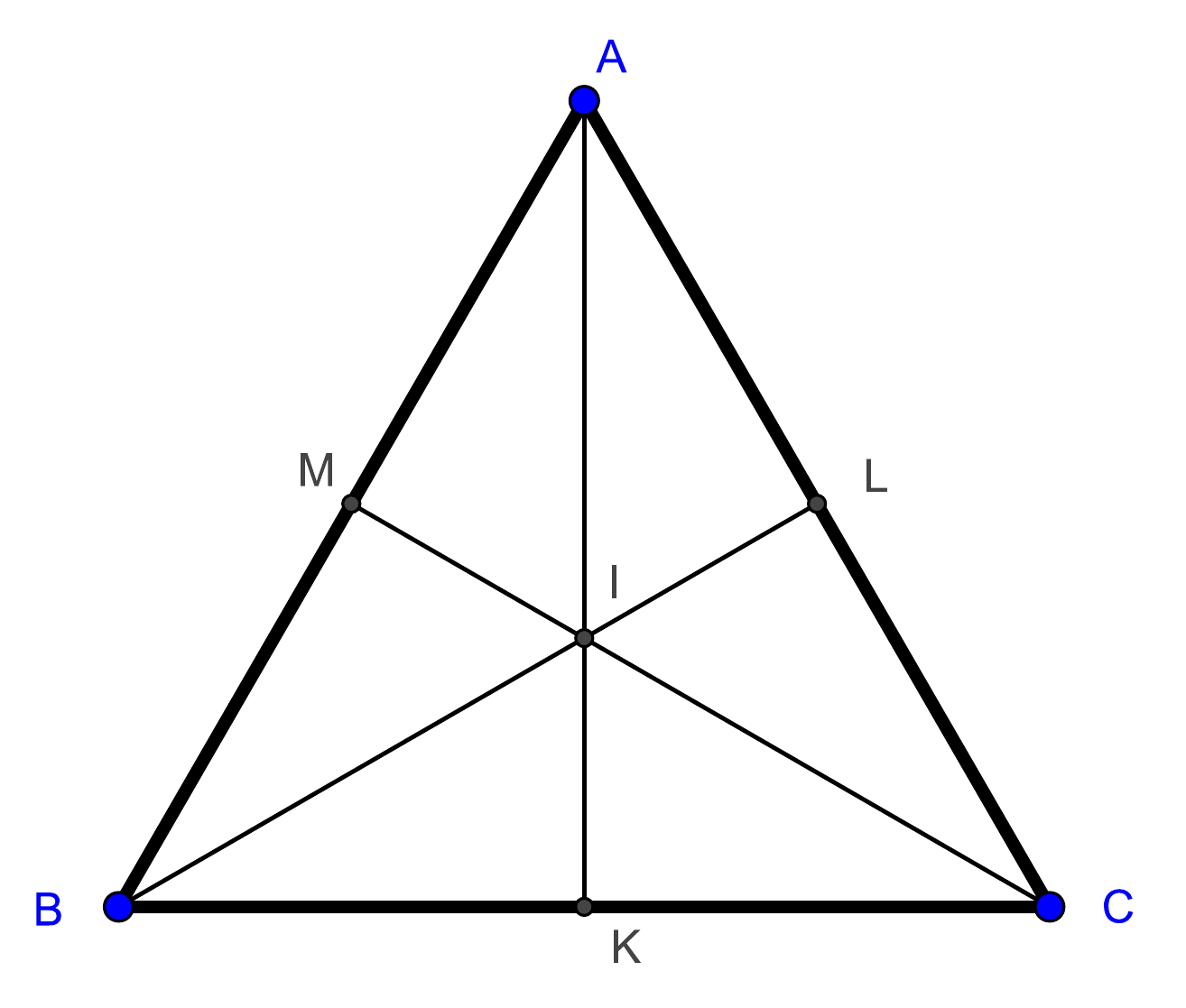}
\caption{The $R_1$ regions of an equilateral triangle.}
\label{fig: R1regionsEquilateral}
\end{figure}
\begin{lemma}[$R_1$ regions of an equilateral triangle]
\label{lem: R1regionsEquilateral}
Consider equilateral triangle $\triangle ABC$ with angle bisectors $AK,BL,CM$ and incenter $I$. 
Then, the angle bisectors are the loci of points in which optimal ordered visitations have the same cost. Moreover, for every starting point $P\in \triangle AMI$, the optimal strategy of $R_1(\tr,P)$ is LRD visitation. 
\end{lemma}

\begin{proof}
Points on $AI$ are in the positive LRD and RDL bounce and subopt halfspaces, since $AI$ is the bisector of $\angle A$. 
As such, points on segment $AI$ are the loci of points $P$, for which 
$$d(P, \{AB,BC,AC\})=d(P,[AB,AC,BC]), d(P,[AC,AB,BC]),$$ 
i.e. the points in which the optimal $R_1$ visitation is both LRD and RLD. 
\begin{figure}
\centering
\includegraphics[width=5cm]{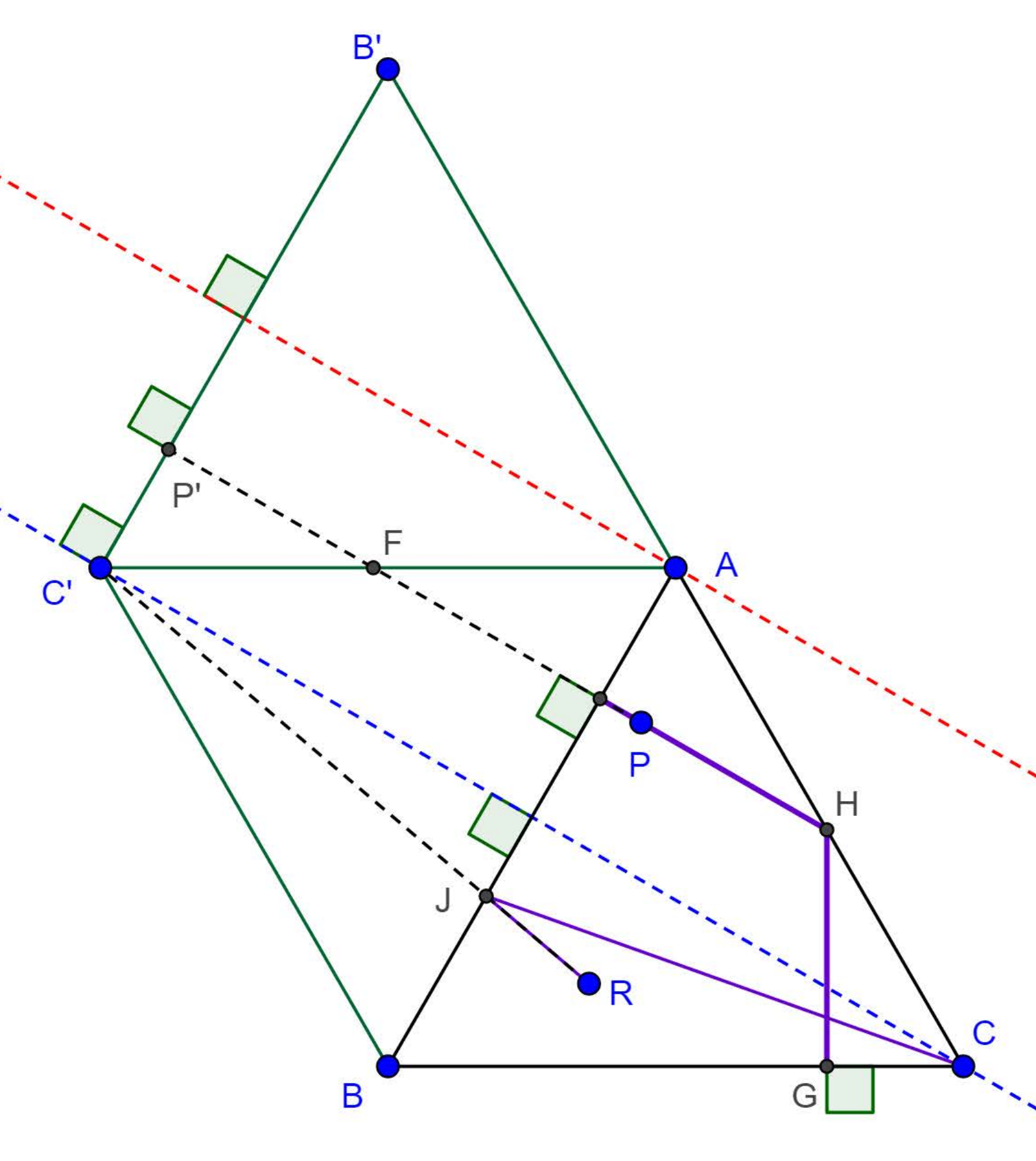}
\caption{
Equilateral $\triangle ABC$ shown with its LRD bounce indicator line (blue dotted line) and its LRD subopt indicator line (red dotted line).}
\label{fig: 3edgesPredereminedOrderSmallC-Equilateral}
\end{figure}
To see why, note that by Figure~\ref{fig: 3edgesPredereminedOrderSmallC-Equilateral}
we have that $d(P,[AB,AC,BC])=\norm{PP'}$. Considering the reflection of $C''$ of $C'$ around $A$, and the reflection $B''$ of $B$ around $A$, it is easy to see that the bisector of the line formed by $B'C'$ and $B''C''$ coincides with the altitude of $\triangle ABC$ corresponding to $BC$ (this property is actually true for every non-obtuse triangle). Since $\triangle ABC$ is equilateral triangle, its altitude corresponding to $BC$ is also the angle $\angle A$ bisector. 
  \end{proof}

The next lemma describes the $R_1$ regions of a right isosceles, as in Figure~\ref{fig: R1regionsRightIsosceles}. 
\begin{figure}
\centering
\includegraphics[width=2.3in]{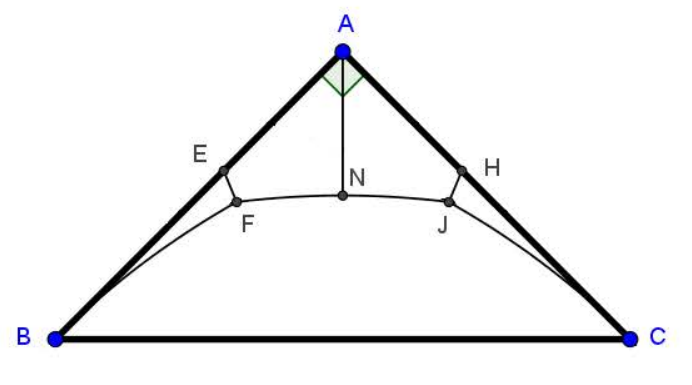}
\caption{The $R_1$ regions of a right isosceles triangle.}
\label{fig: R1regionsRightIsosceles}
\end{figure}
Curve $FJ$ is part of the parabola with directrix the relfection of $BC$ across $A$ and focus the reflection of $A$ across $BC$. 
Curve $BF$ is part of the parabola with directrix a line parallel to $AB$ which is $\norm{AB}$ away from $AB$, and focus the reflection of $A$ across $BC$. 
Curve $CJ$ is part of the parabola with directrix a line parallel to $AC$ which is $\norm{AC}$ away from $AC$, and focus the reflection of $A$ across $BC$. 
$CE$ (not shown in the figure) is the bisector of $\angle C$, and segment $EF$ is part of the reflection of that bisector across $AB$. 
$BH$ is the bisector of $\angle C$, and segment $HJ$ is part of the reflection of that bisector across $AC$. 
Segment $AN$ is part of the altitude corresponding to $A$. 
\begin{lemma}[$R_1$ regions of a right isosceles triangle]
\label{lem: R1regionsRightIsosceles}
Consider right isosceles $\tr=\triangle ABC$, and starting point $P$. Then, the optimal visitation strategy for $R_1(\tr,P)$ is: \\
- An LRD visitation, if $P \in AEFN$,\\
- An LDR visitation if $P \in BFE$, and \\
- Both an $DRL, DLR$ visitation if $P \in BCJF$ (trajectory visits $\{AB,AC\}$ at point $A$). 
\end{lemma}

\begin{proof}
The reader may consult Figure~\ref{fig: R1regionsRightIsosceles}. 
As in the proof of Lemma~\ref{lem: R1regionsEquilateral}, points on $AN$ are the loci of points in which the optimal strategy is both LRD and RLD. 
Points $P$ in $AEFN$ are in the positive LRD bounce and subopt halfspace, and so by Lemma~\ref{lem: +bounce +subopt ordered visitiation}, $R_1(\tr,P)=d(P,B'C')$, where $B',C'$ are the reflections of $B,C$ across $A$, respectively (points $B', C'$ along with the LRD bounce indicator line and the LRD subopt indicator line are depicted separately in Figure~\ref{fig: 3edgesPredereminedOrderSmallC-RightIsosceles}). 
\begin{figure}[h!]
\centering
\includegraphics[width=4cm]{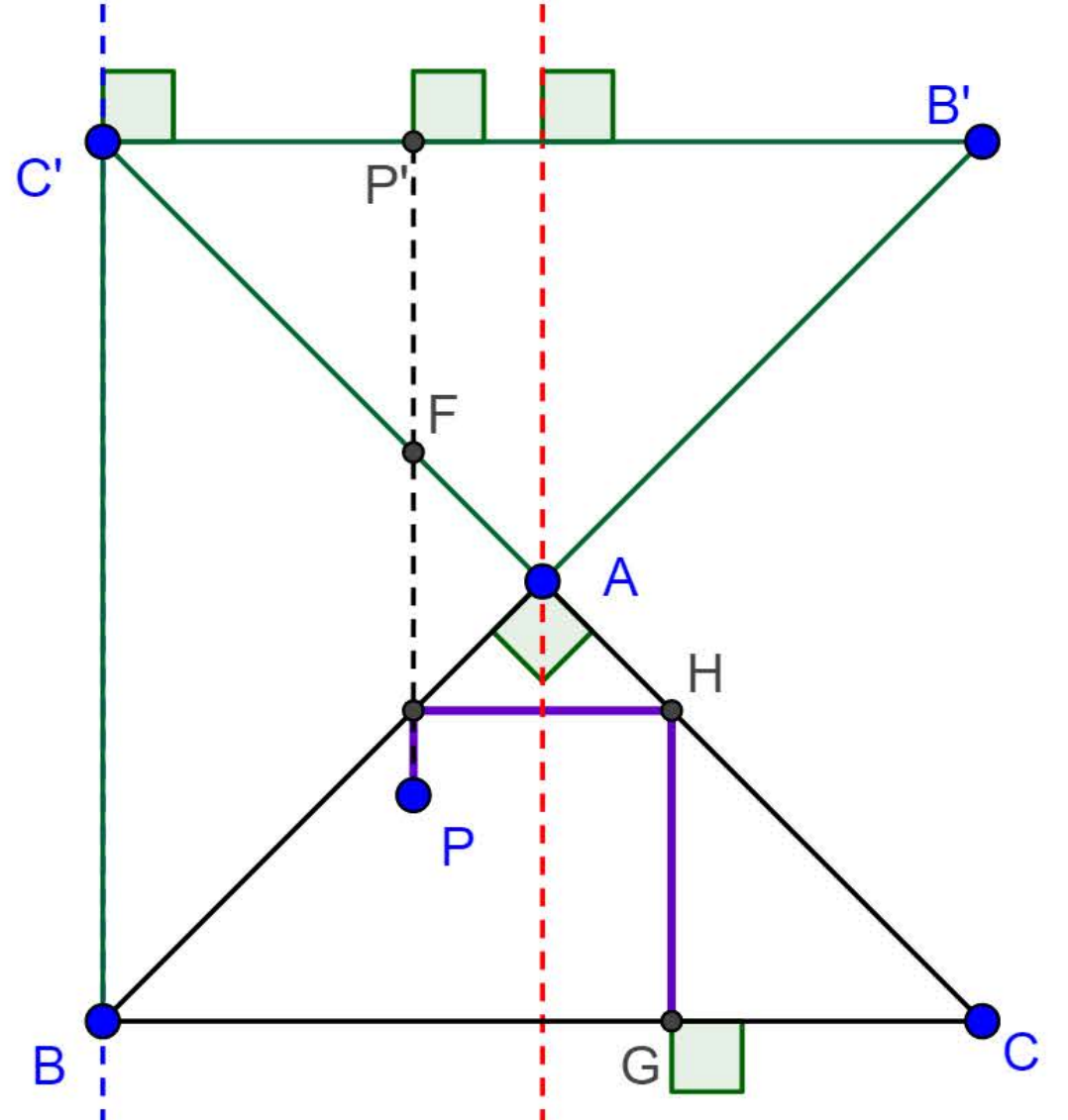}
\caption{
Right isosceles $\triangle ABC$ shown with its LRD bounce indicator line (blue dotted line) and its LRD subopt indicator line (red dotted line).}
\label{fig: 3edgesPredereminedOrderSmallC-RightIsosceles}
\end{figure}
Points $P$ in region $BCJF$ are in the negative DRL (and DLR) bounce halfspace and in the positive subopt halfspace, and so by Lemma~\ref{lem: +bounce +subopt ordered visitiation}, we have that $R_1(\tr,P)$ is obtained by a degenerate bouncing trajectory visiting both $AB,AC$ at point $A$. Considering the reflection $A'$ of $A$ across $BC$, the cost in this case would be $\norm{PA'}$. 
Hence, the loci of points $P$ for which $d(P,B'C')=\norm{PA'}$ is the parabola described in the statement of the lemma, whose portion reads as curve $FJ$.

The rest of the separators follow using similar arguments. Indeed, line segment $EF$ is the loci of points in which the optimal $R_1$ visitation is both LDR and LRD (so the separator is formed by the reflection of the angle bisector of the angle across the first visited edge, a property that can be shown for every non-obtuse triangle). 
Finally, the curve $BF$ (part of a parabola, as described in the statement of the lemma), is the loci of points $P$ for which the optimal $R_1$ visitation is LDR and has cost equal to $\norm{PA'}$. 

The analysis for the rest of the depicted separators is identical, due to the triangle's symmetry along the bisector of $\angle A$ (that coincides with $AN$). 
  \end{proof}

Next we consider a ``thin'' isosceles $\tr=\triangle ABC$ with $\angle A \leq \pi/3$, as in Figure~\ref{fig: R1regionsThinIsocleles} (Eventually we will invoke the next lemma for $\angle A \rightarrow 0$).
\begin{figure}
\centering
\includegraphics[width=1.5in]{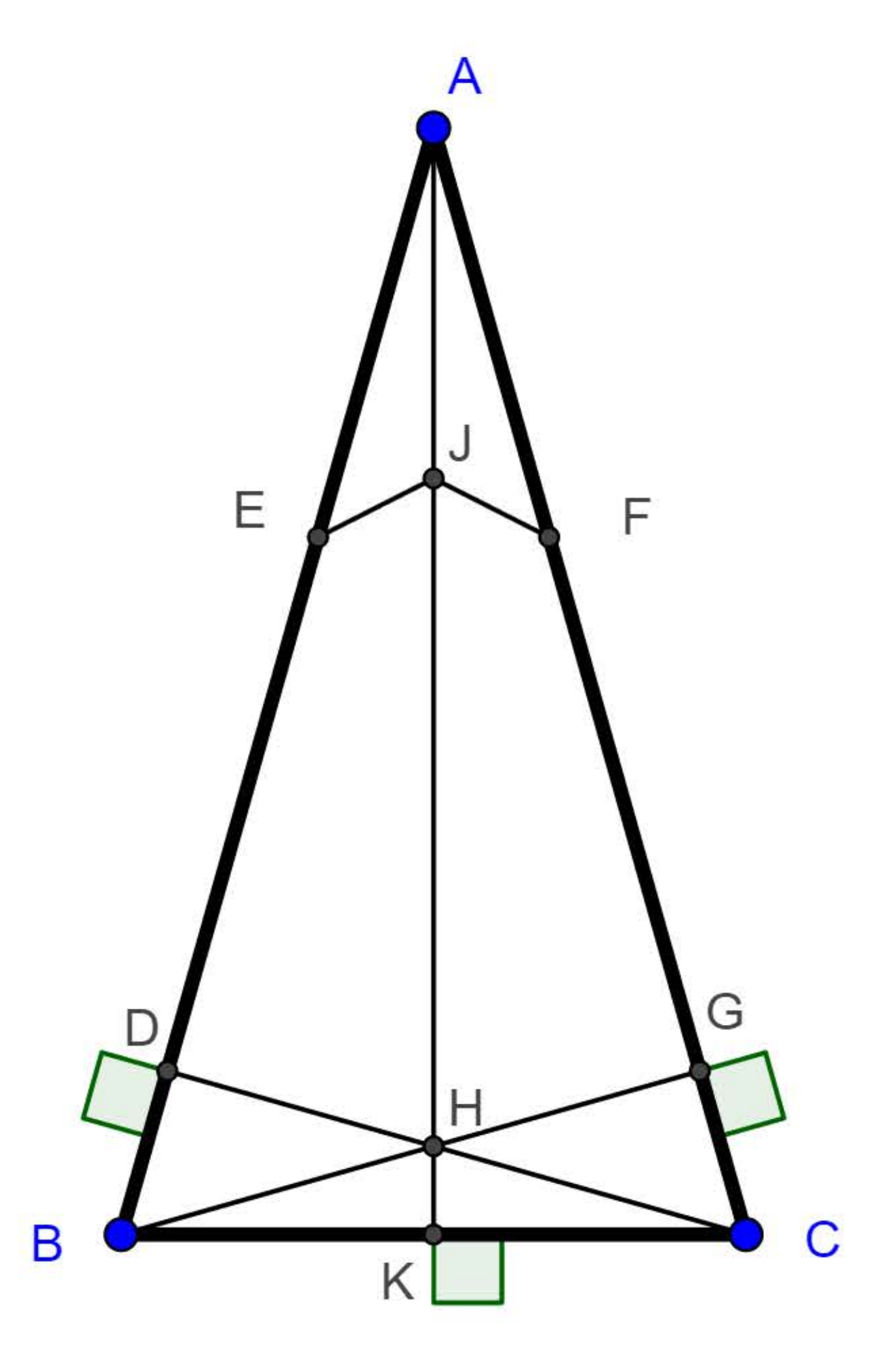}
\caption{The $R_1$ regions of an isosceles triangle $ABC$ with small $\angle A$.}
\label{fig: R1regionsThinIsocleles}
\end{figure}
$AK$ is the altitude corresponding to $A$. 
$CD, BG$ are the 
altitudes corresponding to $AB, AC$, respectively. 
$CE,BF$ (not shown) are the extreme rays of the optimal bouncing subcone corresponding to $C,B$, respectively. 
$H$ is the intersection of $AK$ with $BG$ (and $CD$), i.e. the orthocenter of the triangle. 
Segment $EJ$ (as part of a line) is the reflection of $EC$ (as part of a line) across  $AB$. 
Segment $FJ$ (as part of a line) is the reflection of $BF$ (as part of a line) across  $AC$. 
\begin{lemma}[$R_1$ regions of a thin isosceles triangle]
\label{lem: R1regionsThinIsosceles}
Consider isosceles $\tr=\triangle ABC$, with $\angle A \leq \pi/3$ and starting point $P$. Then, the optimal visitation strategy for $R_1(\tr,P)$ is: \\
- An LRD visitation if $P \in AEJ$, \\
- Both LRD and LDR (optimal strategy is to visit first $AB$ and then move to $C$), if $P \in EDHJ$, \\
- An LDR visitation, if $P \in DBH$, and \\
- A DLR visitation if $P \in BKH$. 
\end{lemma}

\begin{proof}
The separator $AK$ is justified as in the proof of Lemma~\ref{lem: R1regionsEquilateral}.
Points in region $AEJ$ are in the positive LRD bounce and subopt halfspace, and so by Lemma~\ref{lem: +bounce +subopt ordered visitiation}, the optimal $R_1$ visitation is given as a distance of $P$ to a line. 
Points in region $EDHJ$ are in the negative LRD bounce and positive subopt halfspace, and so by Lemma~\ref{lem: +bounce +subopt ordered visitiation}, the optimal $R_1$ visitation is given as a distance of $P$ to a point (the projection $C'$ of $C$ across $AB$). The transition in which the optimal visitation does not visit vertex $C$ happens exactly at segment $EJ$ whose extension passes through $C'$, and in particular has the property that $JC'$ is perpendicular to $BC'$. 
Finally, the justification of separator segments $DH, BH$ is identical to the reasoning of the equilateral triangle (see Figure~\ref{fig: R1regionsEquilateral}) as provided in the proof of Lemma~\ref{lem: R1regionsEquilateral}. 
  \end{proof}

\section{Optimal Visitations of Some Special Starting Points}
\label{sec: special visitations}

\subsection{$R_2$ Cost of the Incenter}
\begin{lemma}
\label{lem: R2costIncenter}
Consider $\triangle ABC \in \trs$ with largest angle vertex $C$ and incenter $I$. Then $R_2(I)=\norm{IC}$. 
\end{lemma}

\begin{proof}
Incenter $I$ is equidistant from all edges of $\triangle ABC$. Since in every optimal $R_2$ strategy, one robot visits an edge and the other visits the remaining two (which is at least as costly as visiting any one edge), the cost of $R_2(I)$ equals the cheapest cost of visiting any two edges of $\triangle ABC$. Now, since $\angle C$ is the largest angle, it follows that $\angle C \geq \pi/3$, and hence by Observation~\ref{obs: Visit2Edges-ii} we have that $d(I,\{AC,BC\})=\norm{IC}$. Therefore, the claim follows once we prove that 
$\norm{IC} \leq \max\{ d(I,\{AB,BC\}), d(I,\{BA,AC\})$, or equivalently 
once we prove that $\norm{IC} \leq  d(I,\{AB,BC\})$, for every $\angle B \leq \angle C$. 


 \begin{figure}[h!]
\centering
  \includegraphics[width=6cm]{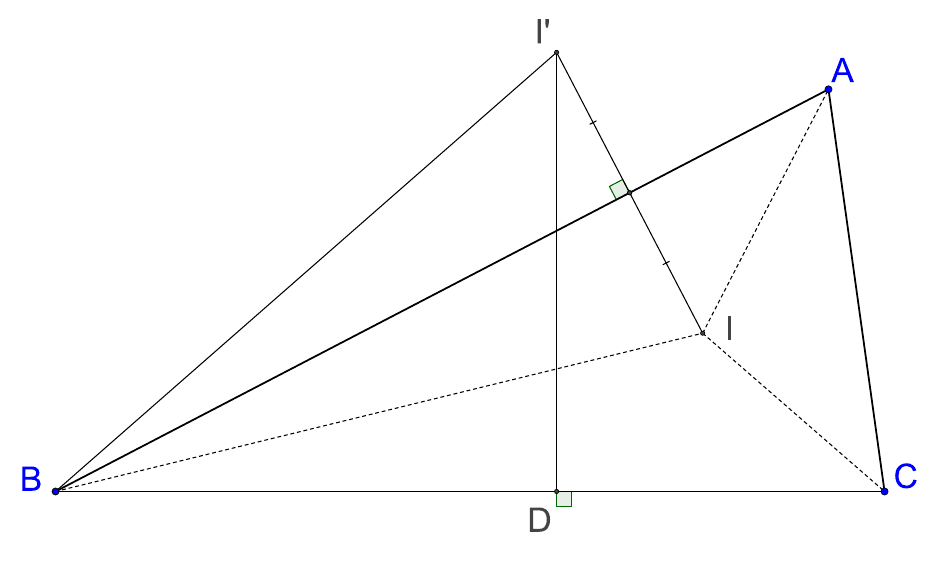}
\caption{
Showing that $\norm{IC} \leq \max\{ d(I,\{AB,BC\}), d(I,\{BA,AC\})$.
}
\label{fig: R2costIncenter}
\end{figure}

Since $\angle C$ is the dominant angle, it follows that $\angle B\leq \pi/3$, see also Figure~\ref{fig: R2costIncenter}. Consider the reflection $I'$ of $I$ around $AB$, and the projection $D$ of $I'$ onto $BC$ (because $\angle B \leq \pi/3$, point $D$ falls within segment $BC$). Then, by Observation~\ref{obs: Visit2Edges-i}, we have $d(I, \{AB,BC\}) = \norm{I'D}$.
So it remains to prove that $\norm{IC}\leq \norm{I'D}$. 
For that we employ the standard analytic form of $\triangle ABC$. 
 For notational comvenience, we introduce notation $\alpha = \norm{BC}=1$, $\beta=\norm{AC}=\sqrt{(1-p)^2+q^2}$ and $\gamma=\norm{AB}=\sqrt{p^2+q^2}$.

Note that $\angle B\leq \pi/3$, and hence $\sqrt3/2 \geq \sin( B) = q/\gamma$. Moreover, since $\angle C\geq \pi/3$, we have $\sqrt3/2 \leq \sin C = q/\beta$. Combining the two inequalities we get the following condition
\begin{equation}
\label{equa: pq condition}
3(1-p)^2 \leq q^2 \leq 3p^2,
\end{equation}
which in particular (combined with that $p\leq 1$) implies that $1/2\leq p\leq 1$. 

By Corollary~\ref{cor: incenter}, the coordinates of the incenter $I=(x,y)$ can be computed as 
$$
x=
\frac{\gamma+p}{1+\beta+\gamma}, 
~~
y=
\frac{q}{1+\beta+\gamma}.
$$
Point $I'=(x',y')$ can be computed by rotating $I$ by angle $B$ (using the Cartesian system), so it follows that $\norm{I'D} =y'= x\sin(B) +y\cos(B)$. After elementary algebraic manipulations, we obtain that 
$$
\norm{I'D} - \norm{IC}
=
\frac{1}{(1+\beta+\gamma)^2}
\left(
\frac{4p^2q^2+4pq^2\sqrt{p^2+q^2}}{p^2+q^2}
-(1+\sqrt{(1-p)^2+q^2}-p)^2
\right).
$$
It is easy to see that 
$\frac{4p^2q^2+4pq^2\sqrt{p^2+q^2}}{p^2+q^2}$ is increasing in $p$, and that $(1+\sqrt{(1-p)^2+q^2}-p)^2$ is decreasing in $p$ (when $p\in [0,1]$).
Hence, using also~\eqref{equa: pq condition}, a lower bound to $\norm{I'D} - \norm{IC}$ is obtained by setting $3(1-p)^2 =q^2$ or by setting $q^2 = 3p^2$ (and the valid lower bound would be the minimum of the two). Next we will use that $1/2\leq p \leq 1$. 

Elementary calculations show that when $3(1-p)^2 =q^2$, we have
$$
\norm{I'D} - \norm{IC}
\geq
-\frac{3 (p-1)^2 \left(2 p \left(-2 \sqrt{4 p^2-6 p+3}+4 p-9\right)+9\right)}{4 p^2-6 p+3}.
$$
The real roots of the latter function of $p$ are $p=1/2,1$. Hence, $\norm{I'D} - \norm{IC}$ preserves sign for all $p\in [1/2,1]$, and the sign is the same as, say, when $p=2/3$, in which case the value of the function becomes $\frac{9}{7} \left(\frac{8 \sqrt{7}}{27}+\frac{16}{27}\right)-1 \approx 0.76981 \geq 0$. 

Finally, when $q^2 = 3p^2$, we obtain that 
$$
\norm{I'D} - \norm{IC}
\geq
2 p-\sqrt{4 p^2-2 p+1}.
$$
The latter continuous expression of $p$ has only one real root $p=1/2$, and it is clearly positive when $p\geq 1/2$.
Hence, we conclude again that $\norm{I'D} - \norm{IC}\geq 0$, as wanted. 
  \end{proof}

\subsection{$R_1$ Cost of the Incenter}

\begin{lemma}
\label{lem: incenter R1 cost}
For non-obtuse $\triangle ABC$ with incenter $I$, let $\angle A$ be its largest angle. Then, $R_1(I) = \norm{IA'}$, where $A'$ is the reflection of $A$ across $BC$.  
\end{lemma}
\begin{proof}

The reader may consult Figure~\ref{fig: R1IncetnerDU-LRDLDR}. 
\begin{figure}
\centering
\includegraphics[width=2in]{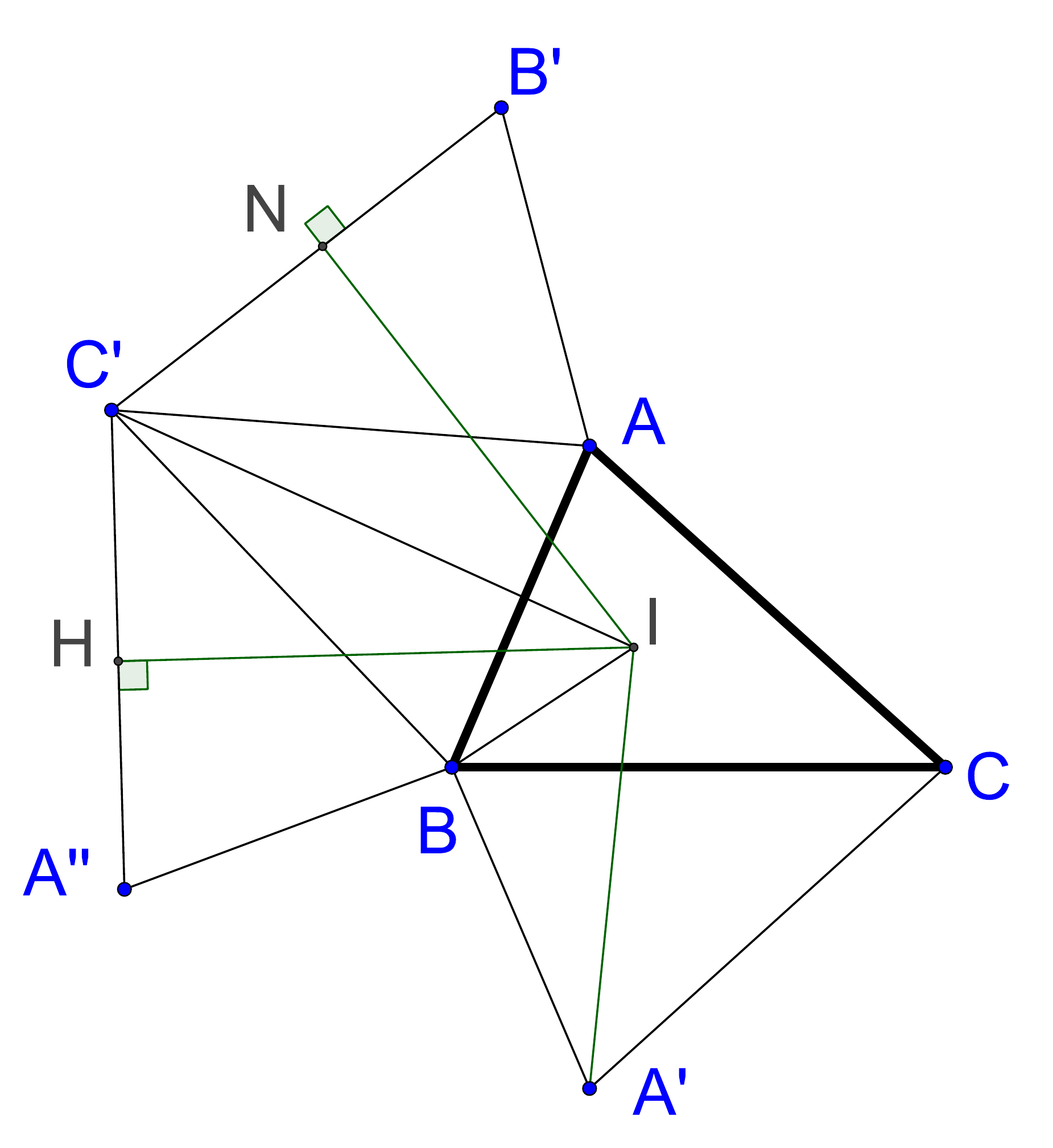}
\caption{Comparison between the optimal DRL strategy with the optimal LRD and optimal LDR strategies.}
\label{fig: R1IncetnerDU-LRDLDR}
\end{figure}
Point $C'$ is the reflection of $C$ across $AB$, and $B'$ the reflection of $B$ across $AC'$. 
First we observe that the optimal trajectory that visits first $BC$ has cost $\norm{IA'}$, that is, the optimal such strategy is both of DLR and DRL type.
Indeed, it is easy to see that $I$ is in the negative DLR bounce halfspace and in the positive subopt halfspace. Hence, by Section~\ref{sec: visiting 3 edges}, the optimal such strategy is of degenerate bouncing type, where the bouncing point $J$ on $BC$ (intersection point of $BC$ and $IA'$) lies within the optimal bouncing subcone of angle $A$ (and hence $d(J,\{AB,AC\})=JA$), and the claim follows.

In order to prove that the DLR (and DRL) type strategy with cost 
$$\norm{IA'}=d(I,\{AB,BC,AC\})$$ is optimal, we will compare it with the optimal LRD type and the optimal LDR strategy (and the symmetric argument would also imply the same comparison with the optimal RLD and RDL strategies).

\bigskip

\textbf{Comparison with optimal LRD strategy:}
Next we compare the optimal DLR strategy above with an optimal LRD strategy. There are three cases to consider.
\begin{description}
\item[Case (a)] $I$ lies in the positive LRD bounce halfspace and in the negative LRD subopt halfspaces ($I$ lies within the optimal bouncing subcone of $\angle C'$, when necessarily $\angle C \geq \pi/3$), in which case the optimal LRD strategy has cost $\norm{IC'}$, see Figure~\ref{fig: R1IncetnerDU-LRDLDR-extreme2}. 
\begin{figure}
\centering
\includegraphics[width=1.8in]{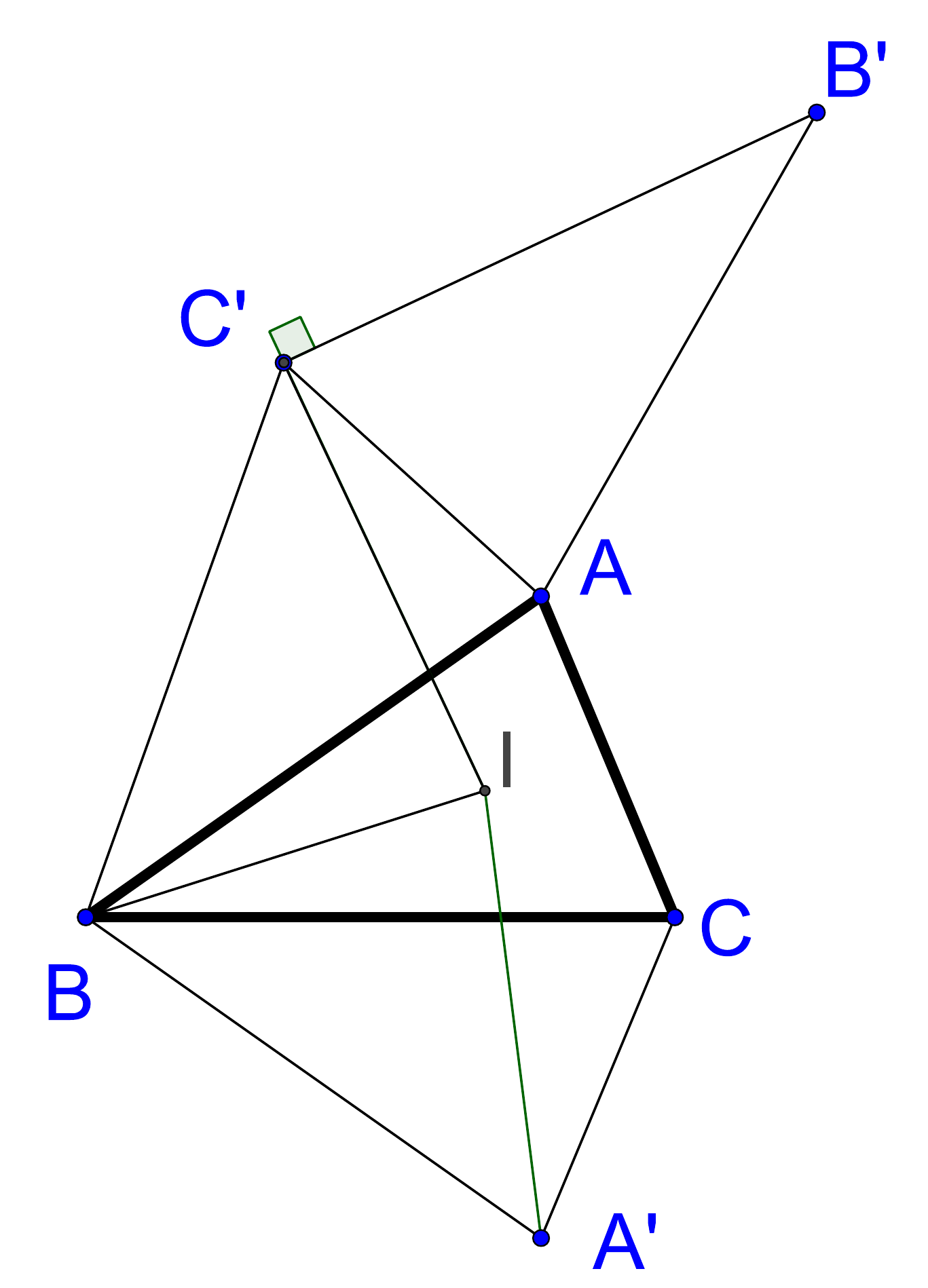}
\caption{The boundary case of Observation~\ref{obs: condition on I interior}, where the optimal LRD strategy has cost $\norm{IC'}$.}
\label{fig: R1IncetnerDU-LRDLDR-extreme2}
\end{figure}
\item[Case (b)] $I$ lies in the negative LRD subopt halfspace, hence, the optimal LRD strategy has cost $\norm{IA}+h_A$, where $h_A$ is the altitude corresponding to angle $A$, see Figure~\ref{fig: R1IncetnerDU-LRDLDR-extreme1}. 
\begin{figure}
\centering
\includegraphics[width=2.0in]{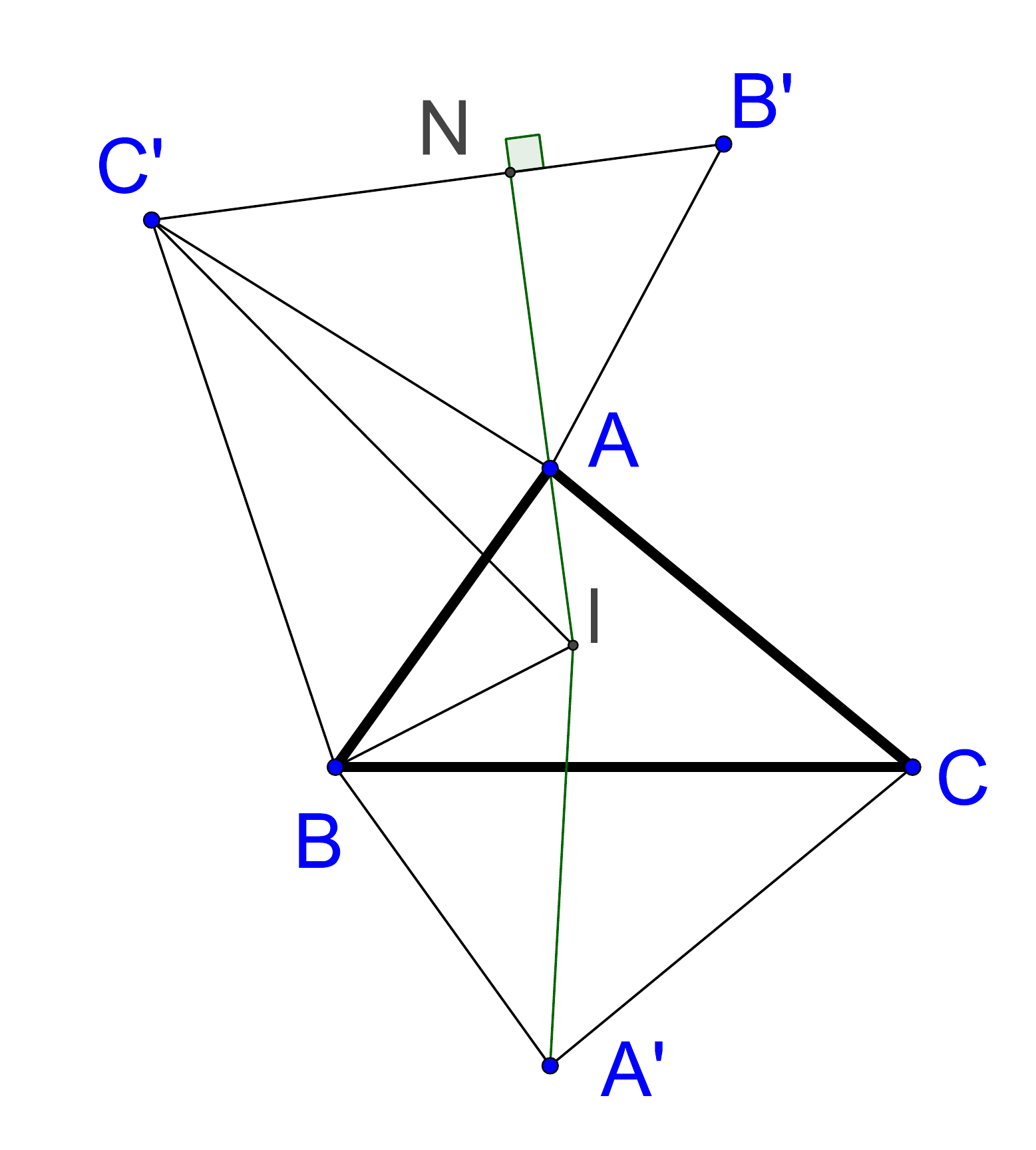}
\caption{The boundary case of Observation~\ref{obs: condition on I interior}, where the optimal LRD strategy has cost $\norm{IA}+h_A$.}
\label{fig: R1IncetnerDU-LRDLDR-extreme1}
\end{figure}
\item[Case (c)] $I$ lies in the positive LRD bounce and subopt halfspaces
and, in particular, the cost of the optimal LRD trajectory is $d(I,B'C')$ (depicted as $\norm{IN}$ in Figure~\ref{fig: R1IncetnerDU-LRDLDR}).
\end{description}

We have the following observation. 
\begin{observation}
\label{obs: condition on I interior}
$I$ lies in the positive LRD bounce and subopt halfspaces exactly when $3A-2C\leq \pi$ and 
$$
\cos ^2(C) \cos ^2\left(\frac{B+C}{2}\right)\geq \frac{\sin ^2\left(\frac{3 A}{2}\right) \sin ^2\left(\frac{C}{2}\right)}{-2 \cos (B+C)+2 \cos (B)-2 \cos (C)+3}
$$
\end{observation}

\begin{proof}
Let $N$ denote the projection of $I$ onto the line passing through $C',B'$. For $I$ to lie in the positive LRD bounce and subopt halfspaces, we need that $N$ falls within segment $B'C'$ and that $IN$ intersects segment $AB$. When $I,A,N$ become collinear (see Figure~\ref{fig: R1IncetnerDU-LRDLDR-extreme1}),  it is easy to see that $A/2+A+\pi/2-C = \pi$, or equivalently that $3A-2C=\pi$. It follows immediately that $IN$ intersects $AB$ in its interior only when $3A-2C\leq \pi$. 
Similarly, let $\rho$ denote $\angle AC'I$, see Figure~\ref{fig: R1IncetnerDU-LRDLDR-extreme2}. Point $N$ coincides with $C'$ exactly when $\rho+C=\pi/2$, and when $\rho+C\leq \pi/2$ point $N$ lies within $B'C'$. 
Now, using the Sine Law in $\triangle C'IA$, we see that 
$
\sin(\rho) = \frac{\norm{IA}}{\norm{IC'}}\sin(3A/2). 
$
Since $\rho+C\leq \pi/2$ is equivalent to that $\sin^2(\rho) \leq \sin^2(\pi/2-C)$, our condition becomes
$$
\cos^2(C) 
\geq 
\frac{\norm{IA}^2}{\norm{IC'}^2} \sin^2(3A/2).
$$
The latter expression can be simplified (after trigonometric manipulations), using also~\eqref{equa: A coordinates} and \eqref{equa: I coordinates} of Observation~\ref{obs: A,I coordinates angles}, together with that $C'=( \cos(2B), \sin(2B) )$, resulting in the promised condition. 
  \end{proof}

\begin{description}
\item{Case (a) proof: }
From Observation~\ref{obs: A,I coordinates angles}, we have that 
\begin{align}
\norm{IA'}^2 &=
(p_I-p)^2+(q_I+q)^2 \notag\\
& =
4 \sin ^2\left(\frac{B}{2}\right) \sin ^2\left(\frac{C}{2}\right) (2 \cos (B+C)+2 \cos (B)+2 \cos (C)+3) \csc ^2(B+C). \label{equa: IA' formula}
\end{align}
Observe that point $C'$ is also obtained by rotating point $C=(1,0)$ by $2\angle B$, and hence $C'=( \cos(2B), \sin(2B) )$. Therefore, 
\begin{align*}
\norm{IC'}^2 &= 
( \cos(2B) - p_I)^2 + ( \sin(2B) - q_I)^2 \\
&=
\sin ^2\left(\frac{B}{2}\right) (-2 \cos (B+C)+2 \cos (B)-2 \cos (C)+3) \csc ^2\left(\frac{B+C}{2}\right).
\end{align*}

Let $f(B,C):=\norm{IC'}^2/\norm{IA'}^2$, so that our goal is to show that $f(B,C)\geq 1$, subject to that $\angle A$ is the largest angle. We claim that $f(B,C)$ is decreasing in $C$. Indeed, elementary calculations show that 
$$
\frac{\partial}{\partial C} f(B,C)
=
-\frac{\cos \left(\frac{B}{2}\right) \csc ^3\left(\frac{C}{2}\right) \cos \left(\frac{B+C}{2}\right)}{(2 \cos (B+C)+2 \cos (B)+2 \cos (C)+3)^2}
\ g(B,C),
$$
where 
\begin{eqnarray*}
g(B,C)&=&
2 \cos (B-C)-4 \cos (B+C)-2 \cos (2 B+C)+2 \cos (B+2 C)\\
&& +4 \cos (B)+2 \cos (2 B)+4 \cos (C)+1.
\end{eqnarray*}
Note that $\frac{\partial}{\partial C} f(B,C)$ is a product of expressions (multiplied by $-1$), and clearly all of them, except possibly $g(B,C)$, are non-negative. Therefore, it remains to show that $g(B,C)\geq 0$. In order to do that we consider two sub-cases.

Sub-case i: If $B\leq \pi/4$, then the only summand of $g(B,C)$ which is negative is $2\cos(B+2C)$. Then, we have 
\begin{align*}
g(B,C) 
& \geq 4\cos(C) +4\cos(B) +2\cos(B+2C) \\
& \geq 4\cos(C) +2\cos(B) +2\cos(B+2C) \\
& = 4\cos(C) +4\cos(B+C)\cos(C) \\
& \geq  4\cos(C) -4\cos(C) \\
& \geq 0,
\end{align*}
where the second to last inequality holds because $B+C\leq 2\pi/3$ (since $A$ is the largest angle and hence $A\geq \pi/3$).

Sub-case ii: If $B> \pi/4$, then the only two summands of $g(B,C)$ which may be negative are $2 \cos (2 B)$ and $2\cos(B+2C)$. Recalling that $A\geq \pi/3$, and using the monotonicity of the cosine function, we get the following sequence of inequalities:
\begin{align*}
& 1+4\cos(B)+2\cos(2B)  \geq -1 \\
& 2\cos(B-C) \geq 2\cos(\pi/12) = \frac{\sqrt{3}+1}{\sqrt{2}} \\
& 4\cos(C) \geq 4\cos(5\pi/12) = \sqrt2( \sqrt{3}-1)  \\
& - 4 \cos(B+C) \geq - 4 \cos(\pi/2) \geq 0 \\
& -2\cos(2B+C) \geq -2 \cos(3\pi/4) =\sqrt2\\
& 2\cos(B+2C) \geq -2.
\end{align*}
So, overall we have that 
$$
g(B,C) \geq 3 \sqrt{\frac{3}{2}}+\frac{1}{\sqrt{2}}-3 >0,
$$
as wanted. Hence, $f(B,C)$ is decreasing in $C$, as promised. 

But then, since $A\geq C$, we have that $C\leq \pi/2-B/2$, and hence it follows that 
$$
f(B,C) \geq f(B,\pi/2-B/2)=1,
$$
where the last equality follows by direct substitution. 
That completes the proof of case (a).

\item{Case (b) proof:}
Let $N'$ be the projection of $A$ onto $BC$, so that $AN'$ is the altitude corresponding to angle $A$. We want to prove that $\norm{IA'} \leq \norm{IA}+\norm{AN'}$.
By triangle inequality, we have that 
$$\norm{IA'} \leq \norm{IN'}+\norm{N'A'}=\norm{IN'}+\norm{AN'}.$$
Hence, it suffices to prove that $\norm{IN'}\leq \norm{IA}$. 
Using the standard analytic form of the triangle, we have that 
$$
\norm{IA}^2 - \norm{IN'}^2 
= (q_I-q)^2-q_I^2 
 = q(q-2q_I). 
$$
Hence, it further suffices to prove that $q\geq 2q_I$. Indeed,
$$
\frac{q}{q_I}
=
2 \frac{\cos \left(\frac{B}{2}\right) \cos \left(\frac{C}{2}\right)}{\cos \left(\frac{B+C}{2}\right)} =: g(B,C)
$$
We have that
$$
\frac{\partial}{\partial C} g(B,C)= \frac{\sin(B)}{1+\cos(B+C)}\geq 0.
$$
Hence, $g(B,C)\geq g[B,0]=2$, where the last equality follows by direct substitution. That completes the proof of case (b).

\item{Case (c) proof:}
The analytic equation of the line $\ell$ passing through $B',C'$ has equation
\begin{equation}
\label{equa: B'C' line}
\ell: ~ \tan(2A) x + y - \sin(2B)-\tan(2A)\cos(2B)=0.
\end{equation}
To see why, recall that from the proof of case (a) above we have that $$C'=( \cos(2B), \sin(2B) ),$$ as well as $\ell$ form with the $x$-axis an angle of $\pi-2A$. 
But then, using the formula for the distance of point $I=(p_I,q_I)$ to $\ell$ we have that 
$$
d(I,B'C')
=
\frac{\left|
\tan(2A)p_I + q_I - \sin(2B)-\tan(2A)\cos(2B)
\right|}
{\sqrt{\tan^2(2A)+1}}.
$$
Using~\eqref{equa: I coordinates} of Observation~\ref{obs: A,I coordinates angles}, together with~\eqref{equa: IA' formula}, and after trigonometric manipulations, it follows that 
\begin{equation}
\label{equa: r1 incenter case b formula}
\left(
\frac{d(I,B'C')}{\norm{IA'}}
\right)^2
=
\frac{\cos ^2\left(\frac{B+C}{2}\right) (-2 \cos (B+2 C)+2 \cos (C)+1)^2}{2 \cos (B+C)+2 \cos (B)+2 \cos (C)+3},
\end{equation}
which we need to prove is at least 1. 
Call function~\eqref{equa: r1 incenter case b formula} $h_1(B,C)$. 
Function $h_1(B,C)$ attains values as low as $9/10$ without conditioning on that $I$ lies the positive LRD bounce and subopt halfspaces.

Consider the domain $\mathcal D_1 \subseteq \reals^2$ of \eqref{equa: r1 incenter case b formula} corresponding to non-obtuse $\triangle ABC$ with $\angle A\geq \angle B,\angle C$ and restricted to $3\angle A-2\angle C\leq \pi$ (as per Observation~\ref{obs: condition on I interior}).

\begin{lemma}
\label{lem: concavity 1}
Function $h_1(B,C)$ is concave over domain $\mathcal D$. 
\end{lemma}
\begin{proof}
We show that function~\eqref{equa: r1 incenter case b formula} is concave by verifying numerically that it's Hessian $H_{B,C}$ is negative-definite. Matrix $H_{B,C}$ is a $2\times 2$ symmetric real matrix, whose real roots can be computed analytically. The domain $\mathcal D_1$ can be re-parameterized as
$$
0 \leq B \leq  3\pi/7, ~~
 \max\{ \pi/2-B, (2\pi-3B)/5\} 
 \leq C \leq 
 \min\{
 2\pi/3-B, \pi/2-B/2, \pi-2B
 \}
$$
so that the eigenvalues can be plotted over $\mathcal D$ as it is shown in Figure~\ref{fig: eigenvalues}.
\begin{figure}[h!]
\centering
  \includegraphics[width=12cm]{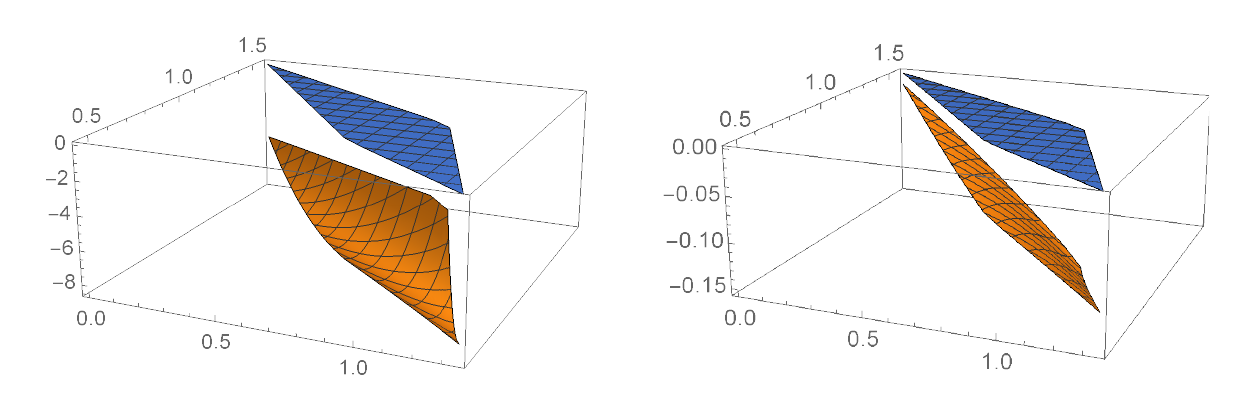}
\caption{The two eigenvalues of Hessian $H_{B,C}$ of function~\eqref{equa: r1 incenter case b formula} are compared to blue hyperplane $z=0$ over domain $\mathcal D_1$.
}
\label{fig: eigenvalues}
\end{figure}
  \end{proof}

Now, by Lemma~\ref{lem: concavity 1}, any (local) minimizers  of function~\eqref{equa: r1 incenter case b formula} are attained at the boundaries of its domain.
Subject to that any of the boundaries of Observation~\ref{obs: condition on I interior} are satisfied tightly, function $h_1(B,C)$ is at least 1, as already proven in cases (a), (b). 
The remaining constraints that might be tight are that $A\geq B,C$. 

Subject to that $B=A$, we have that 
$$
h_1(B,\pi-2B) = 
\frac{\sin ^2\left(\frac{B}{2}\right) (2 \cos (2 B)+2 \cos (3 B)-1)^2}{3-2 \cos (2 B)}.
$$
The function above can be shown to be concave when $B \in [\pi/3,\pi/2]$, hence its minima are attained at the boundaries of its domain. When $B=\pi/3$ its value is 1. When $B=3\pi/7$ its value is equal to 
$$
\frac{\sin ^2\left(\frac{3 \pi }{14}\right) \left(1+2 \sin \left(\frac{3 \pi }{14}\right)+2 \cos \left(\frac{\pi }{7}\right)\right)^2}{3+2 \cos \left(\frac{\pi }{7}\right)}\approx 1.32715.
$$ 
Finally, subject to that $C=A$, it is easy to see that the linear constraints imply that the smallest value that $C$ can attain is $\pi/3$. At the same time, the nonlinear constraint of Observation~\ref{obs: condition on I interior} becomes (for $C=A$)
$$
\frac{\cos ^2(C) (3-2 \cos (2 C)) \csc ^4\left(\frac{C}{2}\right)}{(2 \cos (C)+1)^2}\geq 1.
$$
It can be shown that the above constraint is satisfied only when $C\leq \pi/3$. 
It follows that when $A=C$, we must have $A=B=C=\pi/3$, in which case $h_1(\pi/3,\pi/3)=1$. 
\end{description}

\textbf{Comparison with optimal LDR strategy:} 
In order to describe the optimal LDR strategy, 
we also consider the reflection $A''$ of $A$ across $BC'$, so that the cost of such strategy equals the distance of $I$ to segment $C'A''$, see also Figure~\ref{fig: R1IncetnerDU-LRDLDR}. Let also $H$ denote the projection of $I$ onto $C'A''$. 
As before, we have the following cases to consider. 

\begin{description}
\item[Case (a')] $I$ lies the positive LDR bounce halfspace and in the negative LDR subopt halfspaces, in which case the optimal LDR strategy has cost $\norm{IC'}$. 
\item[Case (b')] $I$ lies in the negative LDR subopt halfspace, hence, the optimal LDR strategy has cost $\norm{IB}+h_B$, where $h_B$ is the altitude corresponding to angle $B$. 
\item[Case (c')] $I$ lies the positive LDR bounce and subopt halfspaces,
and in particular the cost of the optimal LDR trajectory is $d(I,C'A'')$ (depicted as $\norm{IH}$ in Figure~\ref{fig: R1IncetnerDU-LRDLDR}).
\end{description}

\begin{observation}
\label{obs: condition on I interior new}
$I$ lies in the positive LDR bounce and subopt halfspaces exactly when $3B-2C\leq \pi$ and 
$$
\cos ^2(2C) \cos ^2\left(\frac{B+C}{2}\right)\leq \frac{\sin ^2\left(\frac{3 A}{2}\right) \sin ^2\left(\frac{C}{2}\right)}{-2 \cos (B+C)+2 \cos (B)-2 \cos (C)+3}.
$$
\end{observation}

\begin{proof}
As in the proof of Observation~\ref{obs: condition on I interior}, let $\rho$ denote $\angle AC'I$. 
When $I,B,H$ become collinear, it is easy to see that $B/2+B+\pi/2-C = \pi$, or equivalently that $3B-2C=\pi$. It follows immediately that $IH$ intersects $BC'$ in its interior only when $3B-2C< \pi$. 
Point $H$ coincides with $C'$ exactly when $C-\rho+C=\pi/2$, and when $2C-\rho+C\leq \pi/2$ point $H$ lies within $B',C'$. 
Using 
$
\sin(\rho)
$
that was computed in the proof of Observation~\ref{obs: condition on I interior}, we obtain condition 
$$
\cos^2(2C) 
\leq 
\frac{\norm{IA}^2}{\norm{IC'}^2} \sin^2(3A/2).
$$
The latter expression can be simplified (after trigonometric manipulations), using also~\eqref{equa: A coordinates} and \eqref{equa: I coordinates} of Observation~\ref{obs: A,I coordinates angles}, together with that $C'=( \cos(2B), \sin(2B) )$, resulting in the promised condition. 
  \end{proof}

\begin{description}
\item{Case (a') proof:} This case is identical to case (a) above. 

\item{Case (b') proof:}
Let $N'$ be the projection of $A$ onto $BC$, so that $AN'$ is the altitude $h_A$ corresponding to angle $A$. By triangle inequality, we have that $\norm{IA'}\leq \norm{IN'}+h_A$. At the same time, the optimal LDR visitation in this case has cost $\norm{IB}+h_B\geq h_A$, where the inequality is due to that $\angle A$ is the dominant angle, hence $h_A$ is the shortest altitude. Hence, it suffices to argue  that $\norm{IB}\geq \norm{IN'}$, which is immediate $\norm{IB}^2= p_I^2+q_I^2$, while $\norm{IN'}^2= (p_I-p)^2+q_I^2$. 

\item{Case (c') proof:}
The analytic equation of the line $\ell$ passing through $C',A''$ has equation
$$
\ell: ~ -\tan(2B-C) x + y - \sin(2B)+\tan(2B-C)\cos(2B)=0.
$$
To see why, recall that from the proof of case (c) above we have that $$C'=( \cos(2B), \sin(2B) ),$$ as well as $\ell$ forms with the $x$-axis an angle of $2B-C$. 
But then, using the formula for the distance of point $I=(p_I,q_I)$ to $\ell$ we have that 
$$
d(I,\ell)
=
\frac{\left|
-\tan(2B-C)p_I + q_I - \sin(2B)+\tan(2B-C)\cos(2B)
\right|}
{\sqrt{\tan^2(2B-C)+1}}.
$$
Using~\eqref{equa: I coordinates} of Observation~\ref{obs: A,I coordinates angles}, together with~\eqref{equa: IA' formula}, and after trigonometric manipulations, it follows that 
\begin{equation}
\label{equa: Comparison with optimal LDR}
\left(
\frac{d(I,\ell)}{\norm{IA'}}
\right)^2
=
\frac{(2 \cos (B-C)+2 \cos (C)+1)^2 \cos ^2\left(\frac{B+C}{2}\right)}{2 \cos (B+C)+2 \cos (B)+2 \cos (C)+3},
\end{equation}
which we need to prove is at least 1. Call function~\eqref{equa: Comparison with optimal LDR} $h_2(B,C)$. As in the previous case, $h_2(B,C)$ can attain values below 1 without conditioning on that 
$I$ lies in the positive LDR bounce and subopt halfspaces (in which case $d(I,C'A'')=d(I,\ell)$, as per Observation~\ref{obs: condition on I interior new}). 

Recall that $\triangle ABC$ is non-obtuse with 
 $\angle A\geq \max\{\angle B,\angle C\}$, and that $3\angle B-2\angle C\leq \pi$ (as per Observation~\ref{obs: condition on I interior new}). The domain $\mathcal D_2$ defined by these linear constraints can be described as follows:
$$
0 \leq B \leq  3\pi/7, ~~
 \max\{ \pi/2-B, (3B-\pi)/2\} 
 \leq C \leq 
 \min\{
 2\pi/3-B, \pi/2-B/2, \pi-2B
 \}.
$$
Unfortunately, $h_2(B,C)$ is not concave over $\mathcal D_2$, however we can still show numerically that any minimizers are attained at the boundaries of the domain (and the boundary imposed by the non linear constraint of Observation~\ref{obs: condition on I interior new}).

\begin{lemma}
\label{lem: gradient no 0}
Function $h_2(B,C)$ has no minimizers in the interior of domain $\mathcal D_2$. 
\end{lemma}
\begin{proof}
We demonstrate, numerically, that the gradient of function $h_2(B,C)$ is never the zero vector over $\mathcal D_2$. In fact, as Figure~\ref{fig: gradient} shows, $\partial h_2(B,C)/ \partial B$ does not attain value 0 in domain $\mathcal D_2$.\footnote{By considering a sufficiently refined grid, we can numerically (and rigorously) verify that the function is bounded away from 0, as seen in the figure. The claim then follows by theoretical bounds on the partial derivatives of the expression.} 
\begin{figure}[h!]
\centering
  \includegraphics[width=6cm]{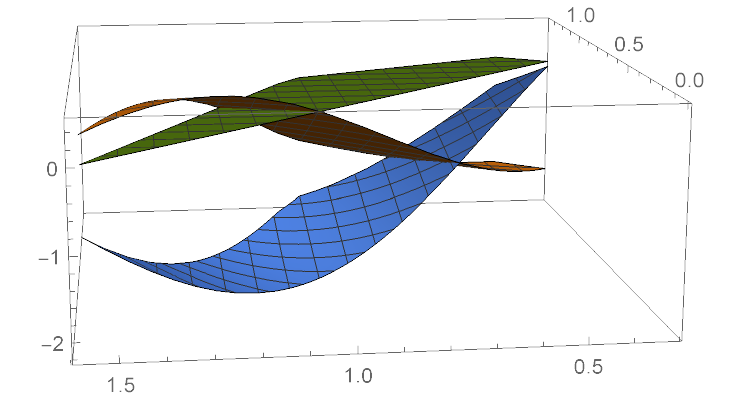}
\caption{The plot of $\partial h_2(B,C)/ \partial B$ (blue) and $\partial h_2(B,C)/ \partial C$ (orange) over domain $\mathcal D_2$, compared with the hyperplane $z=0$ (green plane). 
}
\label{fig: gradient}
\end{figure}
  \end{proof}

Now, by Lemma~\ref{lem: gradient no 0}, any (local) minimizers  of $h_2(B,C)$
 are attained at the boundaries of its domain.
Subject to that, any of the boundaries of Observation~\ref{obs: condition on I interior new} are satisfied tightly, function $h_2(B,C)$ is at least 1, as already proven in cases (a'), (b'). 
The remaining constraints that might be tight are that $\angle A\geq \max\{ \angle B,\angle C\}$. 

Subject to that $B=A$, we have that $h_2(B,\pi-2B) = h_1(B,\pi-2B)$ which was shown to be at least 1 previously, subject to that $\pi/3\leq (3\pi/7)B$ (the same bounds hold for $B$ using the current linear conditions). 
Finally, subject to that $C=A$, it is easy to see, exactly as before, that the linear constraints imply that the smallest value that $C$ can attain is $\pi/3$. At the same time, the nonlinear constraint of Observation~\ref{obs: condition on I interior new} becomes (for $C=A$)
$$
\frac{\sin ^2\left(\frac{3 C}{2}\right) \sec ^2(2 C)}{3-2 \cos (2 C)}\geq 1.
$$
It can be shown that the above constraint is satisfied only when $\pi/6\leq C\leq \pi/3$. 
It follows that when $A=C$, we must have $A=B=C=\pi/3$, in which case $h_2(\pi/3,\pi/3)=1$. 
\end{description}
  \end{proof}

\subsection{$R_1$ Cost of the Middle Point of the Shortest Altitude}

\begin{lemma}
\label{lem: half of altitude R1 cost}
For a non-obtuse $\triangle ABC$ with $\angle A \geq \angle B \geq \angle C$,  let $T$ be the middle point of the altitude corresponding to the largest edge $BC$. Then the optimal $R_1(T)$ strategy is of $LRD$ type, and has cost $\frac{1}{2} (2-\cos (2 A)) \sin (B) \sin (C) \csc (B+C)$.
\end{lemma}

Consider $\triangle ABC$ in standard analytic form, with altitude $AF$. Let $C',B'$ be the reflection of $C,B$ across $AB,AC'$, respectively (see also Figure~\ref{fig: MiddleOfheightLRD-LDR}). 
\begin{figure}
\centering
\includegraphics[width=2.6in]{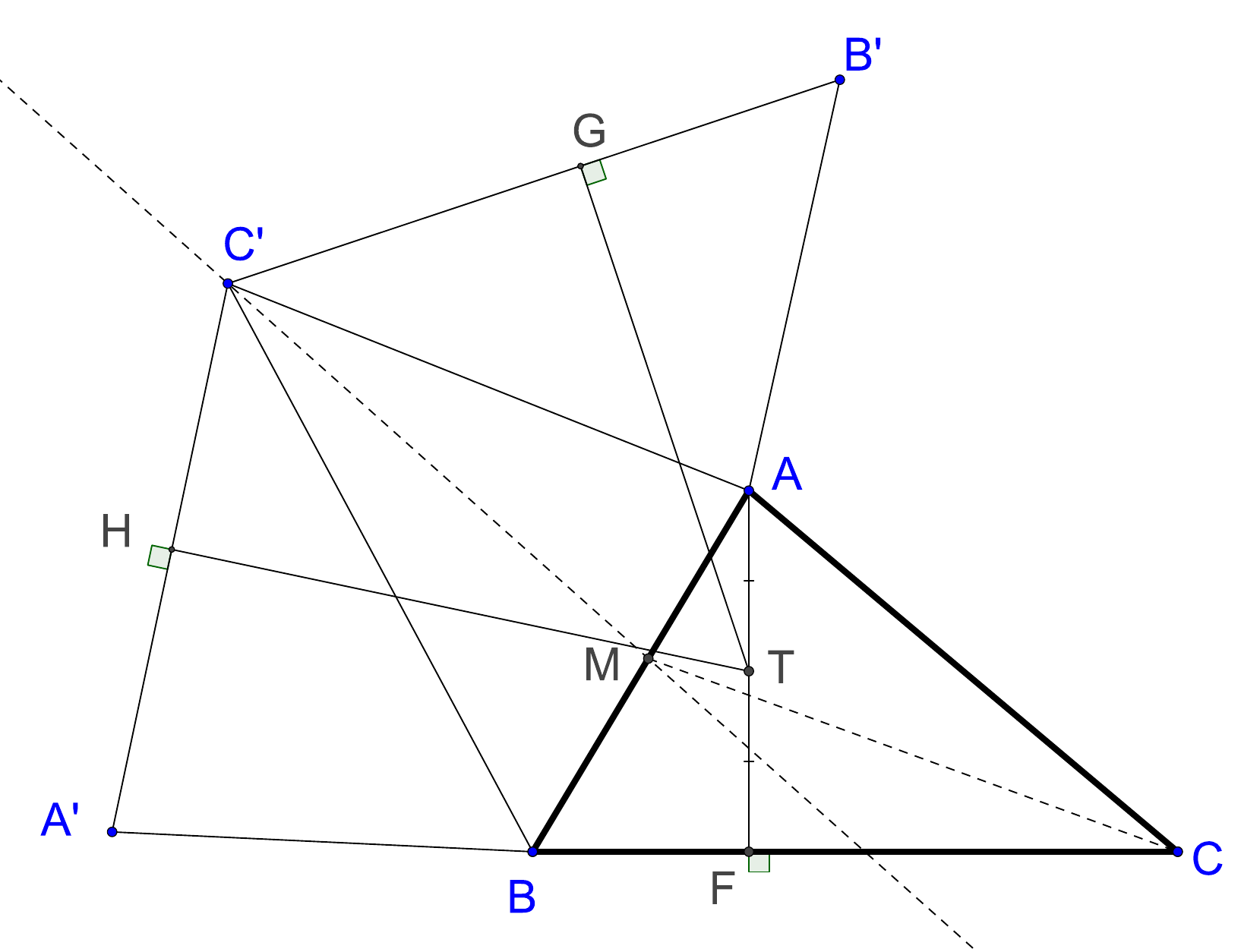}
\caption{
Comparison of optimal LRD with optimal LDR strategy.
}
\label{fig: MiddleOfheightLRD-LDR}
\end{figure}
Since $\angle A \geq \angle B \geq \angle C$, it is easy to verify that $T$ is always in the positive LRD bounce and subopt halfspaces. Hence, the projection $G$ of $T$ along $B'C'$ falls within the latter segment, and in particular by Lemma~\ref{lem: +bounce +subopt ordered visitiation}, the optimal LRD strategy has cost $d(T,B'C')=\norm{TG}$. What we show is that $R_1(T)=\norm{TG}$ by comparing $\norm{TG}$ to the cost of the remaining optimal ordered visitations the triangle edges.

\textbf{Comparison to optimal LDR strategy:}
Consider the reflection $A'$ of $A$ across $BC'$, and let $H$ be the projection of $T$ along $C'A'$, see Figure~\ref{fig: MiddleOfheightLRD-LDR}. 
Clearly the optimal LDR strategy has cost at least $d(T,C'A')=\norm{TH}$. 
The loci of points $P$ that are equidistant from $C'A',C'B'$ are exactly on the angle bisector of $\angle C'$. Moreover, there is no triangle configuration that $T$ falls on the latter angle bisector. To see why, consider the bisector $CM$ of $\angle C$ (hence $C'M$ is the bisector of $\angle C'$). Because $\angle A \geq \angle B \geq \angle C$, we have $\norm{BM}\leq \norm{MA}$, and hence $CM$ intersects $AF$ below $T$. That also shows that $T$ remains always closer to the $B'C'$ segment.

\textbf{Comparison to optimal DLR and DRL strategies:}
Consider the reflection $A'$ of $A$ across $BC$. Let also $C'', B''$ be the reflections of $C,B$ across $BA',CA'$ respectively, see also Figure~\ref{fig: MiddleOfheightLRD-DLR}.
\begin{figure}[h!]
\centering
\includegraphics[width=2.6in]{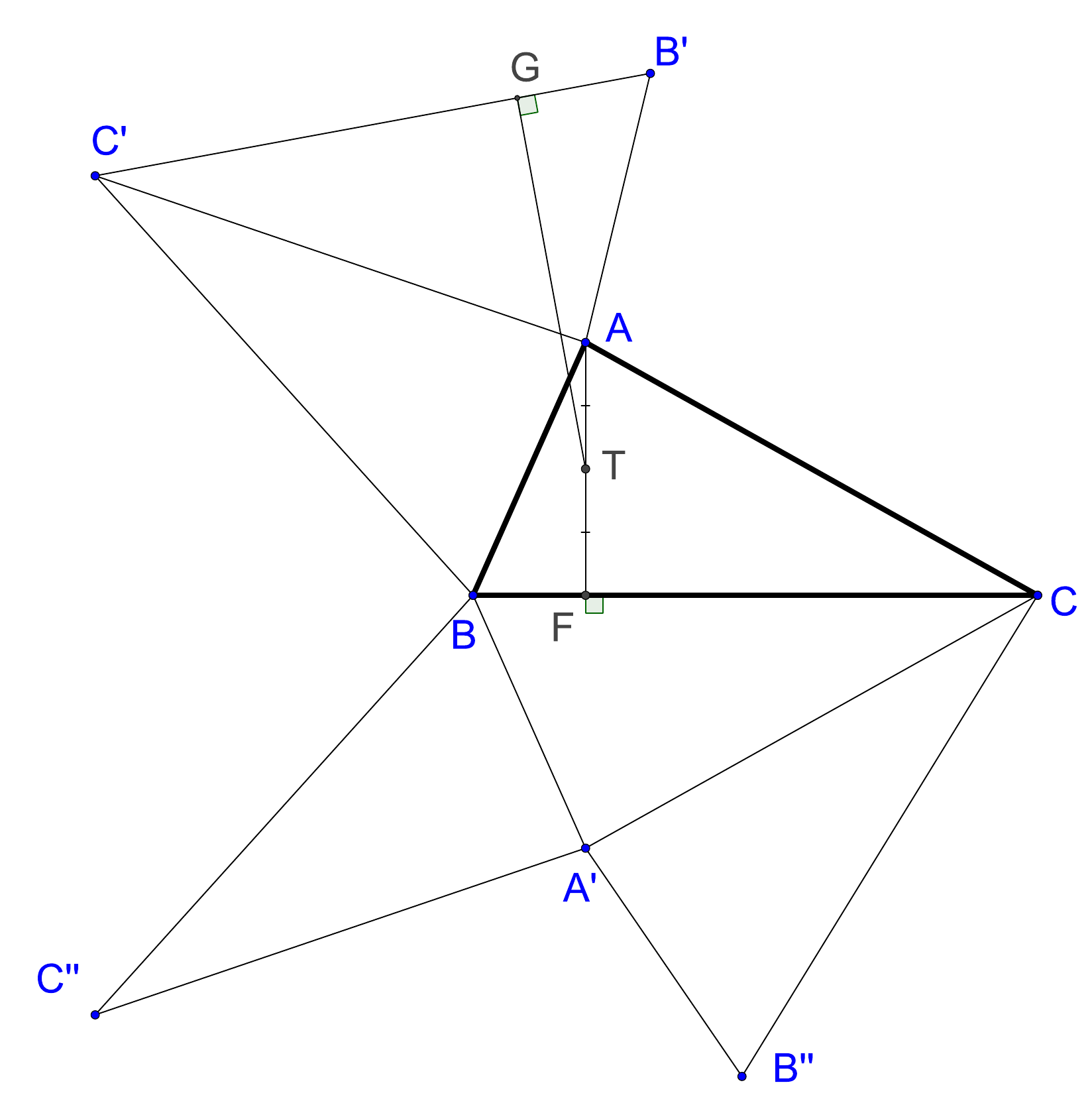}
\caption{
Comparison of optimal LRD with optimal DLR and DRL strategies.
}
\label{fig: MiddleOfheightLRD-DLR}
\end{figure}
It is easy to see that, since $\angle A \geq \angle B \geq \angle C$, point $T$ is always in the negative DLR (and DRL) bounce halfspace and in the positive subopt halfspace. 
It follows by Lemma~\ref{lem: +bounce +subopt ordered visitiation} that both optimal DLR and DRL strategies have cost $\norm{TA'}=\frac32q$. 

We compute $d(T,B'C')=\norm{TG}$ using~\eqref{equa: B'C' line}. Since $T=(p,q/2)$, it follows that 
\begin{align}
\norm{TG} 
&= 
\frac{\left|
\tan(2A)p + q/2 - \sin(2B)-\tan(2A)\cos(2B)
\right|}
{\sqrt{\tan^2(2A)+1}} \notag \\
&=
\frac{1}{2} (2-\cos (2 A)) \sin (B) \sin (C) \csc (B+C).
\label{equa: cost half altitude}
\end{align}
Now, using Observation~\ref{obs: A,I coordinates angles} we have that 
$$
\frac{\norm{TG}}
{\norm{TA'}}
=\frac{1}{3} (2-\cos (2 A))
\leq 1,$$
as wanted. 

\textbf{Comparison to optimal RDL strategy:}
Consider the reflection $B'',A'$ of $B,A$ across $AC$ and $B''C$, respectively, see Figure~\ref{fig: MiddleOfheightLRD-RDL}.
\begin{figure}[h!]
\centering
\includegraphics[width=2.7in]{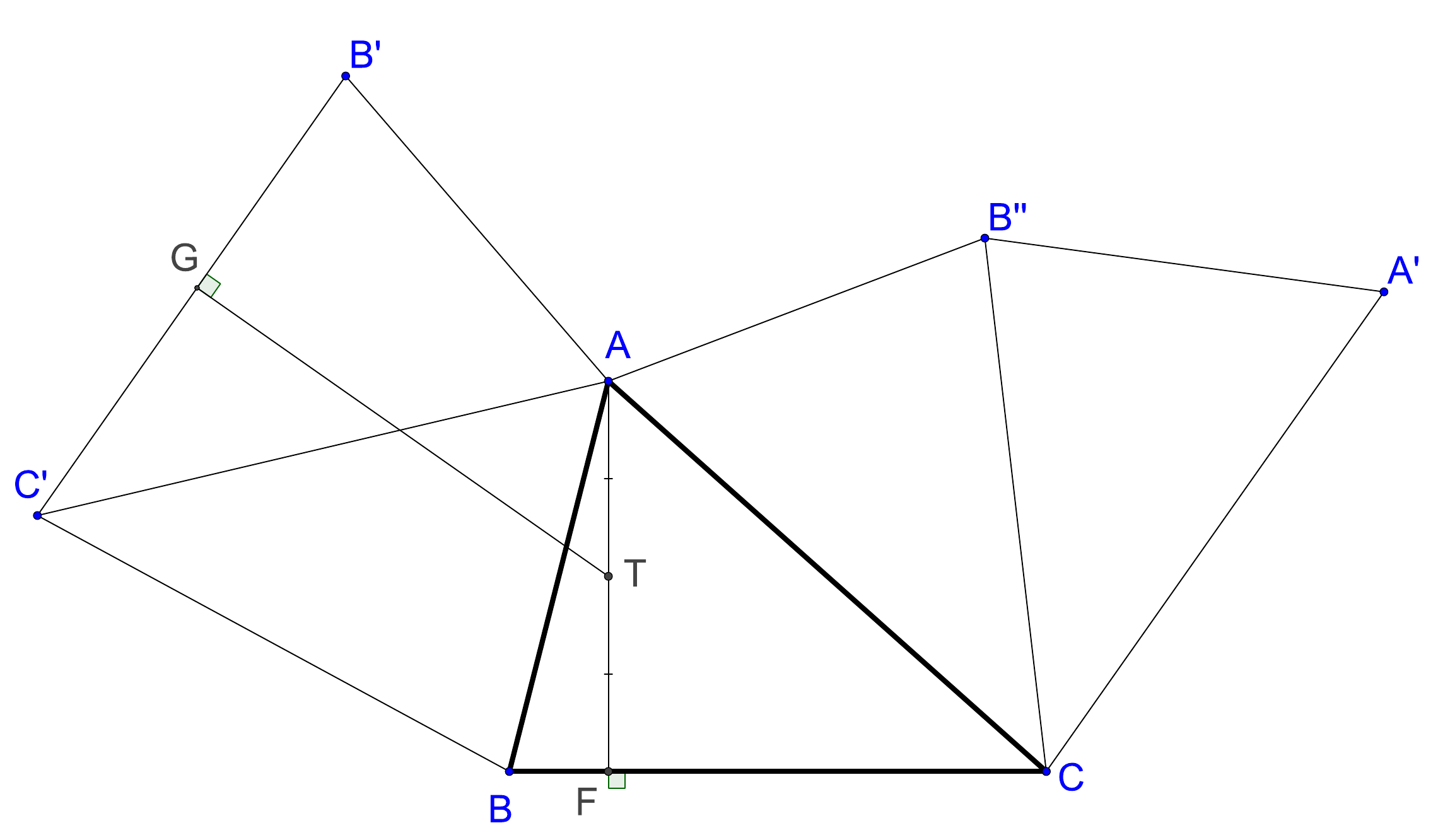}
\caption{
Comparison of optimal LRD with optimal RDL strategy.
}
\label{fig: MiddleOfheightLRD-RDL}
\end{figure}
Since $\angle A \geq \angle B \geq \angle C$, point $T$ is always in the negative RDL bounce halfspace and in the positive subopt halfspace. 
It follows by Lemma~\ref{lem: +bounce +subopt ordered visitiation} that the optimal RDL strategy has cost $\norm{TB''}$. 

Note that by shifting the origin by $(-1,0)$, point $B''$ is obtained as the rotation of $(-1,0)$ by angle $-2C$. Hence, 
$$
B''
=
\mathcal \mathcal R_{-2C}
\left(
\begin{array}{l}
-1\\
0
\end{array}
\right)
+
\left(
\begin{array}{l}
1\\
0
\end{array}
\right)
=
\left(
\begin{array}{l}
1-\cos(2C)\\
\sin(2C)
\end{array}
\right).
$$
Therefore, the optimal RDL strategy has cost
\begin{align*}
\norm{TB''} 
&= \sqrt{\left(1-\cos(2C)-p\right)^2+\left(\sin(2C)-q/2\right)^2} \\
& = \frac{1}{2\sqrt2} \frac{\sin(C)}{\sin(B+C)} \sqrt{-4 \cos (2 (B+C))-\cos (2 B)+4 \cos (2 C)+9}. 
\end{align*}
It follows, after trigonometric manipulations, that 
$$
\frac{
\norm{TB''}^2
}{
\norm{TG}^2
}
=
\frac{(-4 \cos (2 (B+C))-\cos (2 B)+4 \cos (2 C)+9)}{2 (\cos (2 (B+C))-2)^2\sin ^2(B)}. 
$$
Call the latter function $f(B,C)$. We have that 
\begin{align*}
& \frac{\partial}{\partial C}f(B,C) \\
& =
\frac{10 \sin (2 (B+C))-2 \sin (4 (B+C))-\sin (4 B+2 C)+2 \sin (2 B+4 C)+6 \sin (2 B)+7 \sin (2 C)}{(\cos (2 (B+C))-2)^3}.
\end{align*}
Recall that $\angle A \geq \angle B \geq \angle C$. Since also the triangle is non-obtuse, it follows that $0 \leq \angle C\leq \pi/4$, as well as 
$$
\angle C \leq \angle B \leq \pi/2-\angle C/2.
$$
Over this domain, it is easy to verify that 
$$
10 \sin (2 (B+C))-2 \sin (4 (B+C))-\sin (4 B+2 C)+2 \sin (2 B+4 C)+6 \sin (2 B)+7 \sin (2 C)\geq 0. 
$$
Since also $(\cos (2 (B+C))-2)^3\leq 0$, it follows that $f(B,C)$ is decreasing in $C$, and hence
\begin{align*}
f(B,C)  \geq & f(\pi/2-C/2,C) \\
= & \frac{\left(-4 \cos \left(2 \left(\frac{C}{2}+\frac{\pi }{2}\right)\right)-\cos \left(2 \left(\frac{\pi }{2}-\frac{C}{2}\right)\right)+4 \cos (2 C)+9\right) \sec ^2\left(\frac{C}{2}\right)}{2 \left(\cos \left(2 \left(\frac{C}{2}+\frac{\pi }{2}\right)\right)-2\right)^2} \\
= &
\frac{8}{\cos (C)+1}-\frac{27}{(\cos (C)+2)^2}.
\end{align*}
The last function of $C$ can be seen to be increasing in $C$, and when $C=0$ it equals 1. This shows that $\norm{TB''}/\norm{TG}\geq 1$ as wanted. 

\textbf{Comparison to optimal RLD strategy:}
The proof of this case is more holistic. We determine the loci of points $P$ that have the property that $R_1$ costs of the optimal RLD and optimal LRD visitations are equal. In its generality, the loci of points will be a mixed curve composed by a line segment, followed by a parabola segment, followed by a line segment. The mixed curve will split $\triangle ABC$ into regions in which one of the RLD or LRD strategy is strictly more efficient. As it will follow from the proof, the middle point of the shortest altitude of $\triangle ABC$, with $\angle A \geq \angle B \geq \angle C$, will either fall on that mixed curve, or on the side where the LRD strategy is strictly more efficient.

To compare the optimal LRD and RLD visitation costs, we consider reflection $C'$ of $C$ across $AB$, and reflection $B'$ of $B$ across $C'A$, 
see also Figure~\ref{fig: LastSameDown-general}.
\begin{figure}
\centering
\includegraphics[width=2in]{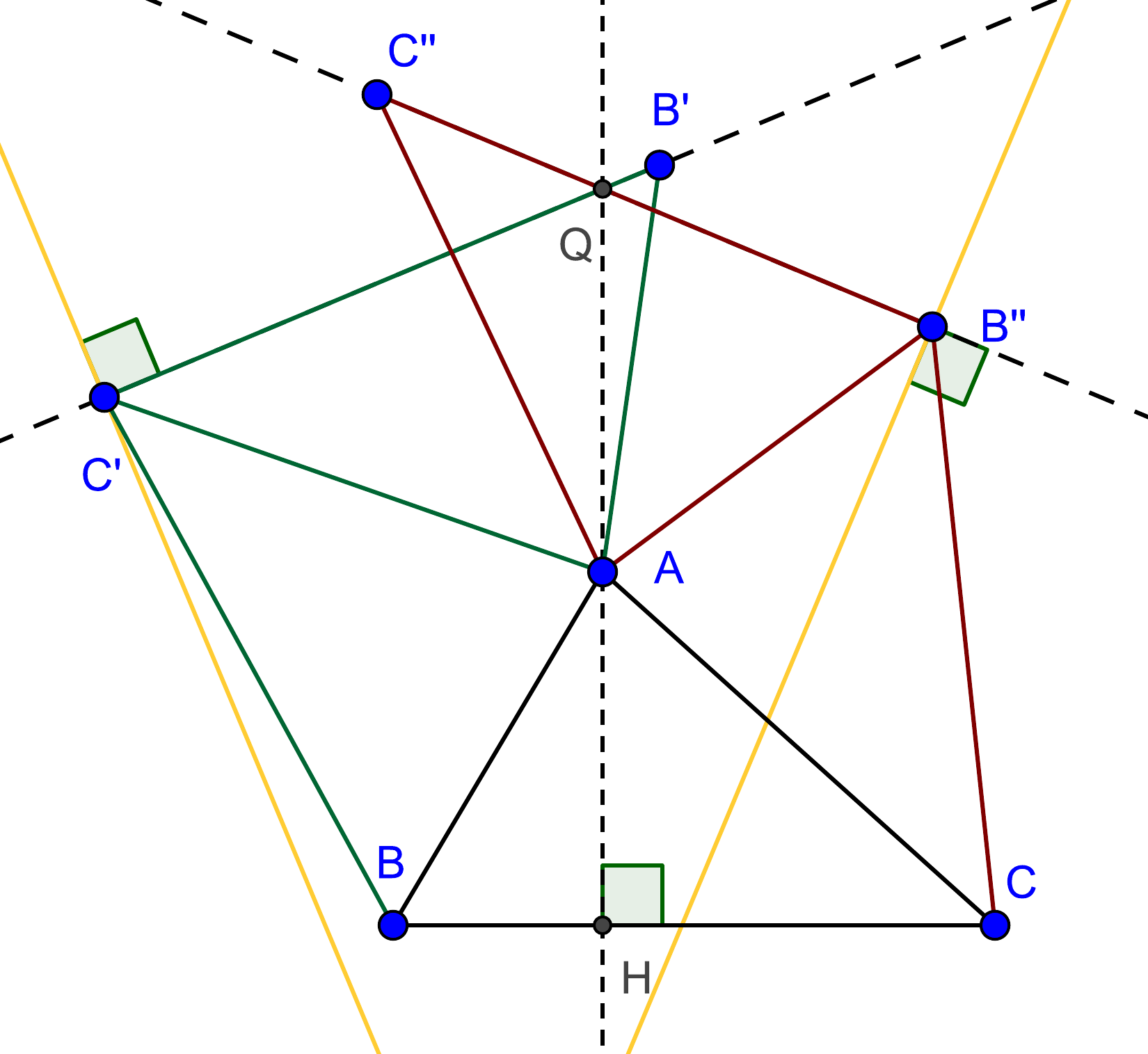}
\caption{The case of all points of altitude $AH$ falling within the positive RLD and LRD bounce halfspaces.}
\label{fig: LastSameDown-general}
\end{figure}
We also consider reflection $B''$ of $B$ across $AC$, and reflection $C''$ of $C$ across $B''A$.

Moreover, consider the LRD and RLD bounce indicator lines, that are perpendicular to $C'B'$ and $C''B''$, respectively. 
Note that $C'B'$ and $C''B''$ (or their extensions) always intersect, say at point $Q$, unless they coincide (exactly when $\angle A = \pi/2$). 
Let also $AH$ be the altitude of $\triangle ABC$ passing through $A$. 
The next observation will be useful in the following arguments. 

\begin{observation}
\label{obs: reflectors bisector}
The extension $\zeta$ of the altitude $AH$ passes also through the intersection $Q$ of (the extensions of) $C'B'$ and $C''B''$.
Moreover, unless $\angle A \not = \pi/2$, line $\zeta$ is also the bisector of one of the angles formed by (the extensions of) $C'B'$ and $C''B''$.
\end{observation}

\begin{proof}
Due to the already introduced reflections, the triangle formed by extensions of $BC, B'C'$ and $B''C''$ is isosceles. This means that $QA$ is both the altitude and the bisector of angle $Q$ of that isosceles triangle. Hence, it suffices to prove that $QA$ is a bisector of the same angle $Q$. To that end, it is enough to show that $A$ is equidistant from $B'C'$ and $B''C''$. The latter property is true since triangles $AB'C'$ and $AB''C''$ are equal, and hence have equal altitudes corresponding to $A$. 
  \end{proof}

By Observation~\ref{obs: reflectors bisector}, all points $P$ on the altitude $AH$ are equidistant from the (extensions of) segments $B'C'$ and $B''C''$. 
Combined with the findings of Section~\ref{sec: visiting 3 edges}, 
we conclude the following corollary. 

\begin{corollary}
\label{obs: R1 separator height}
If all points of altitude $AH$ are in the positive bounce and subopt halfspaces of both LRD and RLD visitations, then the loci of points whose $R_1$ visitation cost is the same for the optimal LRD and RLD visitations is the altitude $AH$. 
\end{corollary}

If the indicator lines intersect in the interior of the altitude $AH$, then for any starting point in the negative halfspace of an indicator line, the optimal bouncing trajectory of the corresponding ordered visitation ends at a vertex. 
If $\angle A \geq \angle B \angle C$, the LRD indicator line intersects altitude $AH$ at point $U$ (closer to $A$ than where the RLD indicator line might intersect $AH$), see also Figure~\ref{fig: LastSameDownHeightParabola}. 
\begin{figure}[h!]
\centering
\includegraphics[width=2.0in]{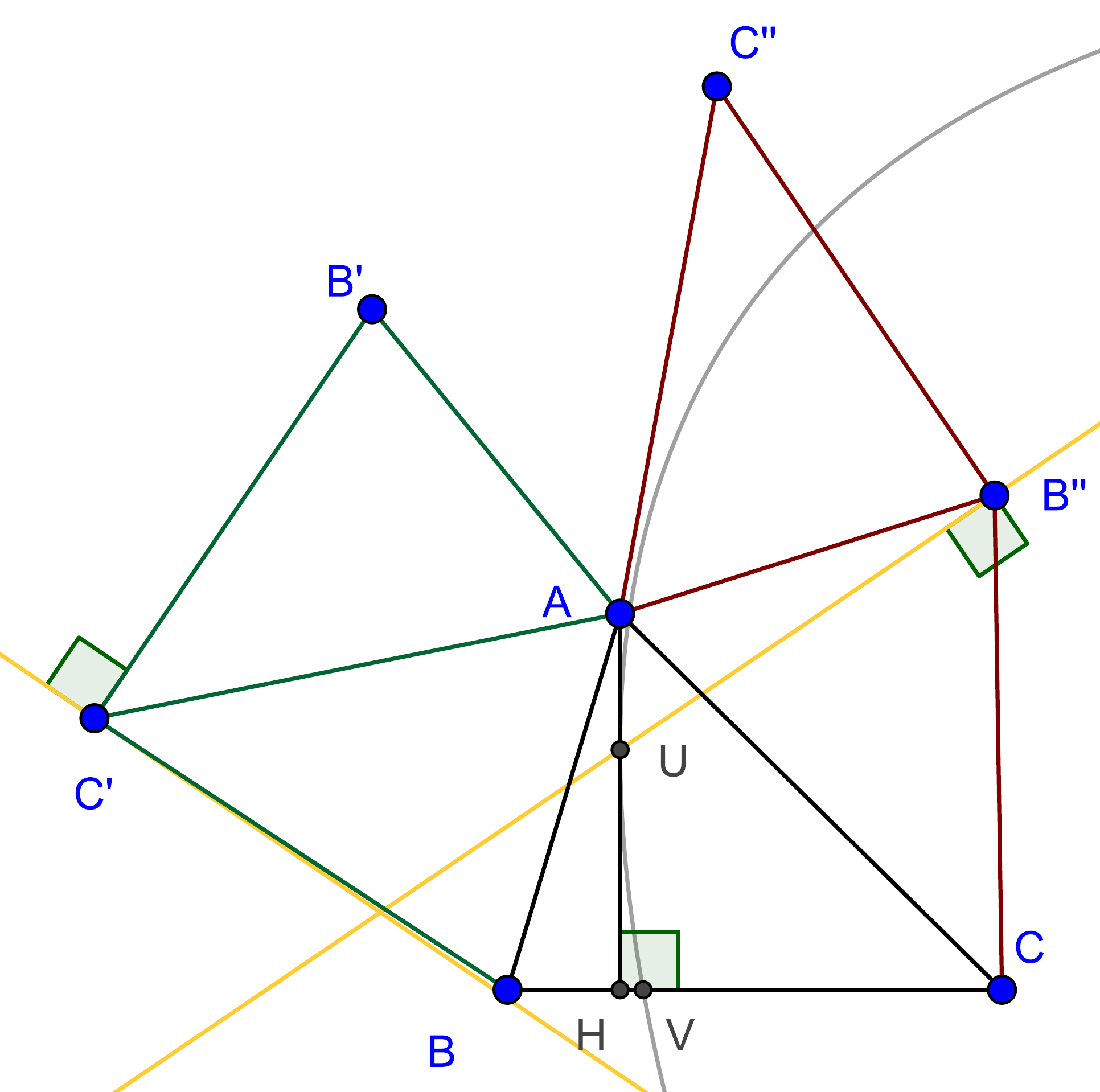}
\caption{The case of the positive RLD and negative LRD bounce halfspaces intersecting altitude $AH$.}
\label{fig: LastSameDownHeightParabola}
\end{figure}

For every starting point $P$ in the intersection of the positive LRD bounce halfspace and the negative RLD bounce halfspace, the optimal LRD visitation is computed by $d(P,B'C')$ (and the projection of $P$ onto $B'C'$ falls within $B'C'$). Also the optimal RLD visitation has cost $\norm{PB''}$. Hence, the loci of points with same LRD and RLD costs is formed by all points that are equidistant from line $B'C'$ and point $B''$. That would be the parabola with focus $B''$ and directrix $B'C'$.
\begin{corollary}
\label{obs: R1 separator height parabola}
The portion of loci of points with same LRD and RLD costs that lies within the intersection of the positive LRD and negative RLD bounce halfspaces is a (part of a) parabola with focus $B''$ and directrix $B'C'$.
\end{corollary}
If the LRD indicator line intersects the aforementioned parabola outside $\triangle ABC$, then the loci of points with same LRD and RLD costs is formed by the portion of the altitude $AU$, followed by the portion of the parabola that lies within the given triangle (curve $UV$ in Figure~\ref{fig: LastSameDownHeightParabola}). 

It remains to examine the case that the LRD bound indicator line intersects the parabola in the interior of triangle $ABC$, say at point $W$, see Figure~\ref{fig: LastSameDownHeightParabolaLineSimplified}.
\begin{figure}[h!]
\centering
\includegraphics[width=1.8in]{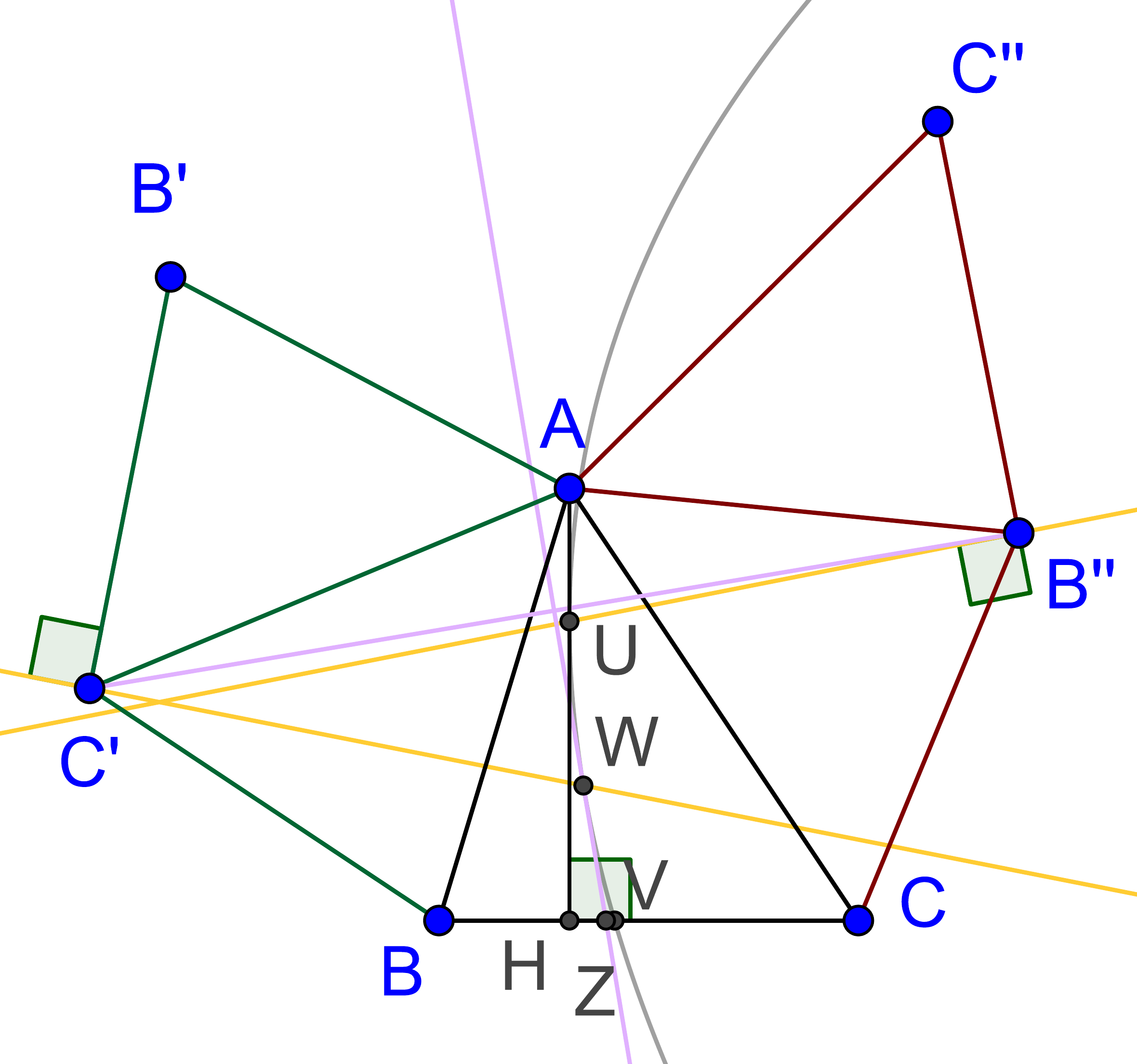}
\caption{The case of the negative RLD and LRD bounce halfspaces intersecting altitude $AH$.}
\label{fig: LastSameDownHeightParabolaLineSimplified}
\end{figure}

 Then, for all starting points $P$ in the intersection of the negative LRD and RLD bounce halfspaces, the optimal LRD and RLD visitations have cost $\norm{PC'}$ and $\norm{PB''}$, respectively. Hence, the loci of points with same LRD and RLD costs within these halfspaces are equidistant from points $B''$ and $C'$, hence they lie on the perpendicular bisector of segment $B''C'$, and let its intersection with line $BC$ be $Z$. 
 
\begin{corollary}
\label{obs: R1 separator height parabola line}
The portion of loci of points with same LRD and RLD costs that lie within the intersection of the negative LRD and RLD bounce halfspaces is part of the perpendicular bisector of segment $B''C'$ that lies within triangle $ABC$. 
\end{corollary}

To conclude, in its generality (see Figure~\ref{fig: LastSameDownHeightParabolaLineSimplified}), the loci of points with same LRD and RLD costs is defined piecewise as, first a portion of altitude $AH$ (segment $AU$), followed by a portion of parabola with focus $B''$ and directrix $B'C'$ (curve $UW$), followed by part of the perpendicular bisector of segment $B''C'$ (segment $WZ$).
If in particular $\angle A\geq \angle B \geq \angle C$, and the middle of the shortest altitude is in the negative RLD bounce halfspace, it is immediate that the optimal LRD visitation is strictly more efficient than the optimal RLD visitation (and otherwise they are equal).


\section{Visitation Trade-offs} 
\label{sec: trade-offs}

\subsection{Searching with 1 vs 3 Robots}
\label{sec: 13 sup}

\subsubsection{Supremum Proof; 1 vs 3 Robots}

In this section we prove the following theorem. 

\begin{theorem}
\label{thm: R13 sup}
$\sup_{\tr \in \trs} \mathcal R_{1,3} (\tr)=4.$
\end{theorem}

The lower bound for $\sup_{\tr \in \trs} \mathcal R_{1,3} (\tr)$ is given by the following simple lemma. 

\begin{lemma}
\label{lem: R13 sup lower bound}
Let $\tr$ be an equilateral triangle, and $I$ be its incenter. Then, $R_1(I)/R_3(I)=4$.
\end{lemma}

\begin{proof}
The reader may consult Figure~\ref{fig: R1regionsEquilateral} for a depiction of the points. 
Let $K$ be the projection of the incenter onto $BC$. 
By Lemma~\ref{lem: incenter R1 cost}, an optimal trajectory for 1 robot would be to go to $K$, and then to $A$ (following the bisector of $A$), inducing cost $R_1(I)=\norm{IK}+\norm{KA}$. 
At the same time, all angle bisectors in the equilateral triangle are also altitudes, so by Lemma~\ref{lem: r3 regions}, the cost for 3 robots equals $R_3(I)=\norm{IK}$. The claim follows by noticing that, in equilateral triangles, $\norm{IK}=\frac13\norm{KA}$.
  \end{proof}

The remaining of the section is devoted in proving a tight upper bound for $\sup_{\tr \in \trs} \mathcal R_{1,3} (\tr).$ In that direction, and for the remaining of the section, we consider a triangle $\tr=ABC$ in standard analytic form. Without loss of generality, we also assume that the starting point $P$ lies within the tetragon (4-gon) $AMIL$, see also Figure~\ref{fig: R3 regions}.

In order to provide the promised upper bound, we propose a heuristic upper bound for $R_1(P)$, as follows. Consider the projections $P_1,P_2,P_3$ of $P$ onto $AB,BC$ and $CA$ respectively. Then, three (possibly) suboptimal visitation trajectories for one robot are $T_C(P):= \langle P,P_1,P,C,\rangle$, $T_A(P):=\langle P,P_2,P,A\rangle$, $T_B(P):= \langle P,P_3,P,B\rangle$, that is
$$
R_1(P) \leq \min \{ 
T_A(P),T_B(P),T_C(P)
\}.\footnote{Note we abuse notation and we denote by $T_A(P)$ both the trajectory and its cost.}
$$ 

The proof of the upper bound follows directly by Lemma~\ref{lem: overline13, A small} and Lemma~\ref{lem: overline13, A large} below.
At a high level, we further distinguish some of the heuristics $T_A(P),T_B(P),T_C(P)$, depending on $\angle A$. 

\begin{lemma}
\label{lem: overline13, A small}
If $\angle A \leq \pi/3$, then 
$\min \{ 
T_B(P),T_C(P)
\}/R_3(P) \leq 4$. 
\end{lemma}

\begin{proof}
Consider some starting point $P$ in the tetragon $AMIL$ (see Figure~\ref{fig: R3 regions}). 
Let $\angle B', \angle C'$ denote angles $\angle CBP$ and $\angle PCB$, respectively, and note that since $I$ is the incenter of the given triangle, then $\angle B'  \geq \angle B/2$ and $\angle C'  \geq \angle C/2$.
At the same time, since $P$ lies in tetragon $AMIL$ (and as a result of the partition of the given triangle using its bisectors) we have that $\max\{ d(P,AB), d(P,AC) \} \leq d(P,BC)$, and in particular $R_3(P) = d(P,BC)$. But then, 
\begin{align*}
\min 
\frac{
\left\{ 
T_B(P),T_C(P)
\right\}
}
{
R_3(P)
}
&= 
\frac{
\min \left\{ 
2d(P,AB) + \norm{PB},
2d(P,AC) + \norm{PC}
\right\}
}
{
d(P,BC)
}
\\
&\leq 
2+
\frac{
\min \left\{ 
\norm{PB},
\norm{PC}
\right\}
}
{
d(P,BC)
}
=
2+
\min 
\left\{ 
\frac1{\sin(B')}, 
\frac1{\sin(C')}
\right\}.
\end{align*}
Since $\angle A \leq \pi/3$, it follows that $\max\{\angle B, \angle C\} \geq \pi/3$. Therefore, $\max\{\angle B', \angle C'\} \geq \pi/6$, and moreover 
$\max\{\sin  B', \sin C'\} \geq 1/2$, which implies the desired upper bound of 4. 
  \end{proof}

Note that the proof of Lemma~\ref{lem: overline13, A small} provides evidence that 
$\min \{ 
T_B(P),T_C(P)
\}/R_3(P)$
is maximized when $P$  is the incenter of the given triangle. The next lemma makes a similar  observation for heuristic trajectory $T_A(P)$.

\begin{lemma}
\label{lem: worst is incenter 13-sup}
The ratio $T_A(P)/R_3(P)$ is maximized when $P$ is either the incenter of $\tr$, or the intersections of the bisectors of $B,C$ with $AC,AB$, respectively.  
\end{lemma}

\begin{proof}
Consider an arbitrary point $P$ in the tetragon $AMIL$ (see Figure~\ref{fig: R3 regions}).
We have that 
$$
\frac{T_A(P)}{R_3(P)}
=
\frac{ 2 d(P,BC) + \norm{PA}}{d(P,BC)}
=
2+ \frac{\norm{PA}}{d(P,BC)}. 
$$
Since $\tr$ is non-obtuse, the closer $P$ is to $BC$, the larger $\norm{PA}$ is and the smaller 
$d(P,BC)$ is. In other words, $T_A(P)/R_3(P)$ attains is maximum for some point $P$ in the line segments $MI, IL$. 
So, let us consider an arbitrary point $P$ in the line segment $IL$. Clearly, it suffices to prove that 
$T_A(P)/R_3(P)$ is maximized either at $I$ or at $L$.
Equivalently, it suffices to prove that $\frac{\norm{PA}}{d(P,BC)}$ is maximized either at $I$ or at $L$.

First we show that $\angle LIA$ is strictly acute. Indeed, since $\tr$ is non-obtuse, we have 
$$
\angle LIA
= \pi- \angle A/2 - \angle ALI = 
\pi- \angle A/2 - (\pi- \angle B/2 - A) =
(\angle A + \angle B )/2 < \pi/2. 
$$
This implies that the projection $A_0$ of $A$ onto line passing through $B,L$, falls within the line segment $IL$. Now consider an arbitrary point 
$$
P_\lambda = (1-\lambda)I + \lambda L
$$
on the line segment $IL$, where $\lambda \in [0,1]$.
In particular, there exists $\lambda_0\in(0,1)$, such that $P_{\lambda_0}=A'$. 
Now, note that as $\lambda$ increases from $0$ to $1$, it is immediate that $d(P_\lambda,BC)$ increases. 
Since $\angle LIA$ is non-obtuse, $\norm{P_{\lambda} A}$ is decreasing when $\lambda \in [0,\lambda_0]$ and increasing when $\lambda \in [\lambda_0,1]$. It follows that 
$\frac{\norm{P_\lambda A}}{d(P_\lambda,BC)}$ attains its maximum either when $\lambda=0$ or when $\lambda =1$, that is, either at the incenter $I$ or point $L$. 
  \end{proof}

We are now ready to prove the lemma that complements Lemma~\ref{lem: overline13, A small}.

\begin{lemma}
\label{lem: overline13, A large}
If $\angle A \geq \pi/3$, then 
$T_A(P)/R_3(P) \leq 4$. 
\end{lemma}

\begin{proof}
By Lemma~\ref{lem: worst is incenter 13-sup}, it is enough to prove that 
$T_A(P)/R_3(P) \leq 4$, when $P$ is either $L$ or $I$, see also Figure~\ref{fig: R3 regions}. As before, we have 
$$
\frac{T_A(P)}{R_3(P)}
=
\frac{ 2 d(P,BC) + \norm{PA}}{d(P,BC)}
=
2+ \frac{\norm{PA}}{d(P,BC)},
$$
therefore our goal is to show that $\frac{\norm{PA}}{d(P,BC)}\leq 2$.

First, consider starting point $P=L$, and let $L',L''$ be the projections of $L$ onto $BC$ and $AB$, respectively. 
Note that $d(L,BC)=\norm{LL'}=\norm{L'L''}=R_3(L)$.
In right $\triangle ALL''$ we have that $\sin(A)=\norm{LL''}/\norm{AL}=d(L,BC)/\norm{AL}$. Since $\angle A\geq \pi/3$ we have $\sin(A)\geq 1/2$ and the claim follows. 

Second, we focus on the starting point $P=I$, the incenter. 
For this we consider $\tr$ in standard analytic form.
Then, 
$$
\frac{T_A(I)}{R_3(I)}
=
\frac{ 2 d(I,BC) + \norm{IA}}{d(I,BC)}
=
2+ \frac{\norm{IA}}{d(I,BC)}. 
$$
Using Corollary~\ref{cor: incenter} that gives the coordinates of incenter $I=(p_I,q_I)$ we can compute 
$$ \norm{IA}=\sqrt{(p_I-p)^2+(q_I-q)^2}, ~~d(I,BC)=q_I,$$
so that by Observation~\ref{obs: A,I coordinates angles}, we get , after some trigonometric manipulations, that 
$$
\frac{\norm{IA}}{d(I,BC)}
= 
\frac1{\cos\left( \frac{B+C}2\right)} \leq 2,
$$
where the last inequality is due to that $\angle A \geq \pi/3$, and hence $\angle B + \angle C \leq 2\pi/3$.
  \end{proof}

\subsubsection{Infimum Proof; 1 vs 3 Robots}

In this section we prove the following theorem. 

\begin{theorem}
\label{thm: R13 inf}
$\inf_{\tr \in \trs} \mathcal R_{1,3} (\tr)=\sqrt{10}.$
\end{theorem}

The next lemma shows that $\inf_{\tr \in \trs} \mathcal R_{1,3} (\tr)\leq \sqrt{10}$.

\begin{lemma}
\label{lem: R13 inf lower bound}
Let $ABC$ be an isosceles with base $BC$. Then, we have  
$$\lim_{\angle A \rightarrow 0} \max_{P \in ABC} R_1(P)/R_3(P) = \sqrt10.$$
\end{lemma}

\begin{proof}
We consider a triangle in standard analytic form, with $B=(0,0), C=(1,0)$ and $A=(1/2,q)$, with $q\rightarrow \infty$. 
For such a triangle, all the $R_1$ regions are summarized in Figure~\ref{fig: R1regionsThinIsocleles}, and the optimal visitation strategies in Lemma~\ref{lem: R1regionsThinIsosceles}.
We need to show that for all starting points $P$, we have $R_1(P)/R_3(P) \leq \sqrt{10}$. By symmetry, we may restrict starting points in $\triangle ABK$, where $K$ is the middle point of $BC$. We examine the following four regions of starting points: $BKH, DBH, DHJE$ and $EJA$. 

First we analyze regions $BKH, DBH$ together. 
Recall that $E$ is identified by the optimal bouncing subcone of $B$, and that $D$ is the projection of $C$ onto $AB$.
Note that as $q$ grows, point $D$ tends to $B$, and hence regions $BKH, DBH$ tend to line segment $BK$. 
An arbitrary point in any of these regions (as $q\rightarrow \infty$), will tend to a point $P=(x,0)$, where $0\leq x \leq 1/2$. Clearly, $R_3(P) = 1-x$. The optimal strategy in $DBH$ is of LDR-type, and hence has cost (in the limit) $x+1$. The optimal strategy in $BKH$ is of DLR-type, and hence has cost (in the limit) $x+1$. Overall, starting from $P$ in any of the regions $BKH, DBH$, we have $R_1(P)/R_3(P)$ tends to $(x+1)/(1-x) \leq 3<\sqrt{10}$. 

Next we analyze region $AEJ$. The key observation is that a robot starting from that region can move parallel to $BC$ (back and forth) visiting both $AB,AC$ before moving to $BC$ along the projection. Hence, $R_1(P) \leq 3/2+ R_3(P)$, or in other words, $R_1(P)/R_3(P) \leq 3/(2 R_3(P)) +1$. 
Next, we show that $R_3(P) \rightarrow \infty$ as $q\rightarrow \infty$ (i.e. as $\angle B \rightarrow \pi/2$), hence concluding this case as well. 
Indeed, any point in $AEJ$ lies above the intersection of $BF$ with $AK$, call it $W$. Since $\angle FBG = 3B-\pi$, it follows that 
$$\angle FBC = \angle WBC = 2\angle B-\pi/2\rightarrow \pi/2,$$
as $\angle B \rightarrow \pi/2$. Hence, $W$ lies arbitrarily away from $BC$ as $\angle B \rightarrow \pi/2$.

Finally, we analyze the region $EDHJ$. Starting from a point $P=(x,y)$, with $0\leq x\leq 1/2$ and $0\leq y \leq q$ (the latter bound is weak but does not affect our analysis), the optimal strategy for $R_1(P)$ is to bounce optimally to
$AB$ and then move to $C$. Equivalently, consider the reflection $P'$ of $P$ across $AC$. Then, as $q\rightarrow \infty$, we have $P'\rightarrow (x+2(1-x),y) = (2-x,y)$ , and therefore $R_1(P) \rightarrow \sqrt{(2-x)^2+y^2}$. 

We consider three further subcases, and our underlying assumption remains that $0\leq x\leq 1/2$ and $0\leq y \leq q$.
If $P$ is (on or) above bisector of $\angle C$, then (in the limit) we have $y \geq -x+1$ and $R_3(P)=y$. But then, we have 
$$
\frac{R_1(P)}{R_3(P)} \leq \sup_{y\geq -x+1} \frac{\sqrt{(2-x)^2+y^2}}{y} = \sqrt{10},
$$
achieved for $x,y\rightarrow 1/2$. 
If $P$ is below the bisector of $\angle C$ and (on or) above the bisector of $\angle B$, then (in the limit) we have $x\leq y \leq -x+1$, $y\geq x$, and $R_3(P)=1-x$. But then, we have 
$$
\frac{R_1(P)}{R_3(P)} \leq \sup_{x\leq y < -x+1, } \frac{\sqrt{(2-x)^2+y^2}}{1-x} = \sqrt{10},
$$
achieved for $x,y\rightarrow 1/2$. 
Finally, if $P$ is below bisector of $\angle B$, then (in the limit) we have $y \leq x$ and $R_3(P)=1-x$. But then, we have 
$$
\frac{R_1(P)}{R_3(P)} \leq \sup_{y< x} \frac{\sqrt{(2-x)^2+y^2}}{1-x} = \sqrt{10},
$$
achieved again for $x,y\rightarrow 1/2$. 
Note that worst starting point $P=(1/2,1/2)$ is the (limit of the) incenter of thin isosceles $ABC$ (as $\angle A \rightarrow 0$). 
  \end{proof}

The next lemma shows that $\inf_{\tr \in \trs} \mathcal R_{1,3} (\tr)\geq \sqrt{10}$.

\begin{lemma}
\label{lem: R13 inf upper bound}
For any triangle $\tr \in \trs$, let $I$ denote its incenter. 
Then, we have $R_1(I)/R_3(I) \geq \sqrt{10}$. 
\end{lemma}

\begin{proof}
\footnote{The provided proof uses algebraic tools of analytic geometry. An alternative approach, using Observation~\ref{obs: A,I coordinates angles}, is to show, using trigonometric manipulations, that 
$\left( R_1(I)/R_3(I)\right)^2 = \frac{4 \cos (B)+4 \cos (C)+2}{\cos (A)+1}+4$. Then, one would need to minimize the latter expression over the domain of non-obtuse triangles $\triangle ABC$ with dominant angle $A$.} 
Consider an arbitrary non-obtuse $\triangle ABC$ in standard analytic form. We assume that $\angle A$ is the largest angle. 
Let $I=(x,y)$ be its incenter, with coordinates given by Corollary~\ref{cor: incenter}. 
Since the incenter is equidistant from all triangle edges, and by Lemma~\ref{lem: r3 regions}, we have  $R_3(I)=y$.
Moreover, by Lemma~\ref{lem: incenter R1 cost}, we have that 
$$R_1(I)= \sqrt{(x-p)^2+(y+q)^2}.$$ In what follows we prove the following claims:  \\
Claim (a): $R_1(I)/R_3(I)$ is increasing in $q$, therefore it is minimized when $\angle A = \pi/2$. \\
Claim (b): The ``thinner'' a right triangle is, the smaller $R_1(I)/R_3(I)$ becomes.

\textit{Proof of Claim (a):} 
Assume that $\angle A\leq \pi/2$ is the largest angle, and without loss of generality assume also that $p\in [0,1/2]$. Note that 
$$
\frac{R_1(I)}{R_3(I)} 
= \frac{\sqrt{(x-p)^2+(y+q)^2}}{y}
=
\sqrt{
(x/y-p/y)^2+(1+q/y)^2
},
$$
where in particular, $x=x(p,q)$ and $y=y(p,q)$. Since $\sqrt{z}$ is an increasing function of $z$, it suffices to show that $(R_1(I)/R_3(I))^2$ is an increasing function of $q$. 
Note that the function we want to prove increasing in $q$ is of the form $f^2(q)+g^2(q)$, where $f(q) = (x-p)/y\geq 0$ and $g(q)=1+q/y\geq 0$. Note that, 
$$
\left( f^2(q)+g^2(q) \right)'
=
2f(q) f'(q) + 2g(q) g'(q) \geq (f(q)+g(q))' \min\{ f(q), g(q)\}.
$$
Therefore, for $(R_1(I)/R_3(I))^2$ to be increasing in $q$, it suffices to prove that $f(q)+g(q) -1 = (x-p)/y+q/y$ is increasing in $q$. To that end, we compute
$$
\frac{\partial}{\partial q} \left( (x-p)/y+q/y \right)
=
\frac{h(p,q) }{\beta  \gamma  q^2}, 
$$
where $h(p,q) := p^3 (-(\beta +\gamma ))+p^2 (\beta +2 \gamma )-\gamma  p+q^3 (\beta +\gamma )$ and $\beta = \norm{AC}=\sqrt{(1-p)^2+q^2}, \gamma=\norm{AB}=\sqrt{p^2+q^2}$. 
Therefore, it further suffices to prove that 
$h(p,q) \geq 0$. 

What we show next is that $h(p,q)$ preserves sign, condition on that $p,q> 0$, and on that $\angle A \leq \pi/2$. 
Indeed, we compute all roots of $h(p,q)$ with respect to $q$. Some tedious calculations show that $h(p,q)$ has two complex roots, and the real roots 0 and 
$$
q_{1,2}=\frac{-p^2\pm\sqrt{-3 p^4+6 p^3-4 p^2+p}+p}{2 p-1},
$$
among which only $q_1 = \frac{-p^2-\sqrt{-3 p^4+6 p^3-4 p^2+p}+p}{2 p-1}$ is non-negative for $p\in (0,1/2]$. At the same time, $\angle A \leq \pi/2$, and hence $(p-1)^2+q^2 \geq 1/4$. However, it is easy to show that 
$$
(p-1)^2+q_1^2 \leq 1/4,
$$
for all $p \in [0,1/2]$, and equality holds only when $p=0$. 
Therefore, continuous function $h(p,q)$ has no real roots in the domain $p,q>0$ and $(p-1)^2+q^2 \geq 1/4$, and hence must preserve sign. The sign is the same as the sign of $h(1/2,1)=\sqrt{5} > 0$, as wanted.

\textit{Proof of Claim (b):} 
Consider arbitrary right triangle $A=(p,q), B=(0,0), C=(1,0)$, with $\angle A = \pi/2$. Point $A$ must lie on a cycle with radius $1/2$ and center $(1/2,0)$, and hence $(p-1/2)^2+q^2=1/2^2$, from which we conclude that $q=\sqrt{p-p^2}$. Using Corollary~\ref{cor: incenter} we obtain that 
$$I=\left(\frac{1}{2} \left(-\sqrt{1-p}+\sqrt{p}+1\right),\frac{\sqrt{(1-p) p}}{\sqrt{1-p}+\sqrt{p}+1}\right).$$

Then, using the discussion above, and after elementary calculations, we see that 
$$
\frac{R_1(I)}{R_3(I)}
=\sqrt{4 \sqrt{1-p}+4 \sqrt{p}+6}.
$$
It is easy to see that $4 \sqrt{1-p}+4 \sqrt{p}+6$ preserves positive sign, and it is increasing $p$. Therefore, its square is increasing, that is $R_1(I)/R_3(I)$ is increasing in $p$. 

To conclude, using claims (a),(b) above, we have
$$
\inf \frac{R_1(I)}{R_3(I)}
\geq \lim_{p\rightarrow 0}
\sqrt{4 \sqrt{1-p}+4 \sqrt{p}+6}.
=
\sqrt{10},
$$
and the proof is finished.
  \end{proof}

\subsection{Searching with 2 vs 3 Robots}
\label{sec: 23 sup}

\subsubsection{Supremum Proof; 2 vs 3 Robots}

In this section we prove the following theorem. 

\begin{theorem}
\label{thm: R23 sup}
$\sup_{\tr \in \trs} \mathcal R_{2,3} (\tr)=2.$
\end{theorem}

The next lemma shows that $\sup_{\tr \in \trs} \mathcal R_{2,3} (\tr)\geq 2$.

\begin{lemma}
\label{lem: R23 sup lower bound}
Let $ABC$ be an equilateral triangle with incenter $I$. Then, we have  
$R_2(I)/R_3(I) = 2$.
\end{lemma}

\begin{proof}
The incenter $I$ is equidistant from all edges $AB,BC,CA$, and hence by Lemma~\ref{lem: r3 regions}, we have $R_3(I)=d(I,BC)$.
Also, by Lemma~\ref{lem: R2costIncenter}, we have $R_2(I)=\norm{IA}$,
and therefore
$$
\frac{R_2(I)}{R_3(I)} = 
\frac{\norm{IA}}{d(I,BC)}.
$$
Now recall that $\triangle ABC$ is equilateral, therefore each of the bisectors coincide with the altitudes. Moreover, $I$ is the center of the regular triangle, and hence its apothem (with length $d(I,BC)$) is 1/3 of the altitude. The main claim follows by noting that $\norm{IA}$ makes up the remaining $2/3$ of the altitude. 
  \end{proof}

The remaining of the section is devoted in proving that $\sup_{\tr \in \trs} \mathcal R_{2,3} (\tr)\leq 2$.
In that direction, we consider a triangle $\tr=ABC$ along with its incenter $I$, see also Figure~\ref{fig: R3 regions}. Without loss of generality, we also assume that the starting point $P$ lies within the $\triangle AIL$, .

In order to provide the promised upper bound, we propose a heuristic upper bound for $R_2(P)$. The two robots visit all edges as follows; one robot goes to the vertex corresponding to the largest angle (visiting the two incident edges), and the second robot visits the remaining edge moving along the projection of $P$ along that edge. 
Note that the largest angle is at least $\pi/3$.

\begin{lemma}
\label{lem: overline23, A large}
If the largest angle is $\angle A$, then 
$R_2(P)/R_3(P) \leq 2$. 
\end{lemma}

\begin{proof}
Consider an arbitrary point $P$ in $\triangle AIL$. 
Due to the heuristic strategy of $R_2(P)$, one robot goes to $A$ in time $\norm{PA}$, and the other robot goes to edge $BC$ in time $d(P,BC)$. So overall, we have
$$
R_2(P) \leq \max\{\norm{PA},d(P,BC)\}.
$$
Since $R_3(P)=d(P,BC)$, it follows that if $\norm{PA}<d(P,BC)$, then $R_2(P)/R_3(P) =1$. 
On the other hand, if $\norm{PA}\geq d(P,BC)$, then we have
\begin{equation}
\label{equa: bound R23 large A}
\frac{R_2(P)}{R_3(P)}
\leq 
\frac{\norm{PA}}{d(P,BC)}.
\end{equation}
Recall that $P$ lies in $\triangle AIL$. By the proof of Lemma~\ref{lem: worst is incenter 13-sup}, ratio~\eqref{equa: bound R23 large A} is maximized either at the incenter, or at point $L$. Then, by the proof of Lemma~\ref{lem: overline13, A large}, and since $\angle A\geq \pi/3$ the same ratio is at most $2$. 
  \end{proof}

\begin{lemma}
\label{lem: overline23, A small}
If the largest angle is either $\angle B$ or $\angle C$, then 
$R_2(P)/R_3(P) \leq 2$.
\end{lemma}

\begin{proof}
We provide the proof of the case that $\angle C$ is the largest angle, and hence at least $\pi/3$ (the other case is identical). 
For every point $P$ in $\triangle AIL$, we have that $R_3(P)=d(P,BC)$ (see Section~\ref{sec: regions 3}). Due to the heuristic we are using, the claim follows once we show that 
$
\frac{ \max\{\norm{PC},d(P,AB)\} }{d(P,BC)} \leq 2.
$
Note that $d(P,AB) \leq d(P,BC)$, and so, if $$\max\{\norm{PC},d(P,AB)\} = d(P,AB),$$ it follows that $R_2(P)/R_3(P) \leq 1$. 

It remains to examine the case $R_2(P) = \norm{PC}$. To that end, note that $\angle PCB \geq \angle C/2 \geq \pi/6$. Since moreover the $\sin$ function is increasing in $[0,\pi/2]$, we have
$$
\frac{\norm{PC}}{d(P,BC)}
= \frac1{
\sin\left( 
\angle PCB
\right)
}
\leq \frac1{\sin(\pi/6)}
= 2
$$
as wanted. 
  \end{proof}

\subsubsection{Infimum Proof; 2 vs 3 Robots}

In this section we prove the following theorem. 

\begin{theorem}
\label{thm: R23 inf}
$\inf_{\tr \in \trs} \mathcal R_{2,3} (\tr)=\sqrt2.$
\end{theorem}

The next lemma shows that $\inf_{\tr \in \trs} \mathcal R_{2,3} (\tr)\leq \sqrt2$

\begin{lemma}
\label{lem: R23 inf lower bound}
Let $\triangle ABC$ be a right isosceles with $\angle A = \pi/2$. Then, we have  
$$\max_{P \in ABC} R_2(P)/R_3(P) = \sqrt2.$$
\end{lemma}

\begin{proof}
Consider right isosceles $\triangle ABC$ with $\angle A = \pi/2$.
For such a triangle, The $R_3$ regions are summarized in Lemma~\ref{lem: r3 regions}. Also, all $R_2$ regions are summarized in Corollary~\ref{cor: r2 hexagon right isosceles} and Lemma~\ref{lem: r2 rightisosceles within}, see also Figure~\ref{fig: R2 regions right isosceles}. 

We conclude that the $R_2$ regions of right isosceles $ABC$ with $\angle A=\pi/2$ are determined by the refined $R_2$ separator $MFKJLQ$ and the incenter $I$.
More specifically, for every starting point $P$ outside $MFKJLQ$, cost $R_2(P)$ is determined by the cost of visiting the more distant edge (and hence it equals $R_3(P)$). 
Hence, $$\argmax_{P \in ABC} R_2(P)/R_3(P)$$ lies within $MFKJLQ$, and by symmetry, we may further assume that $P$ lies in $MFAQ$. 

We consider the analytic representation of right isosceles $ABC$, that is we set $A=(1/2,1/2), B=(0,0), C=(1,0)$. Using Corollary~\ref{cor: incenter}, the incenter is $I=(1/2,1/\sqrt2-1/2)$. By Lemma~\ref{lem: R2costIncenter}, we have that  $R_2(I)=\norm{IA}=1-1/\sqrt2$. Recalling also that $R_3(I)=d(I,BC)=1/\sqrt2-1/2$, it follows that $R_2(I)/R_3(I)=\sqrt2$, Below, we show that 
$$\argmax_{P \in MFAQ} R_2(P)/R_3(P) \leq \sqrt2.$$
In what follows, we commonly denote $P=(a,b)$, for any placement of $P$.  
Combined with the partition of $ABC$ that determines all costs $R_3(P)$, as per Section~\ref{sec: regions 3}, we are motivated to consider the following four cases. 

\begin{description}
\item[Case 1, $P\in MIQ$:]
We have that $R_2(P)=\norm{PA}$ and $R_3(P)=d(P,BC)$. 
Clearly, $R_2(P)$ is maximized either at $P=M$ or at $P=I$ (one of which may be a local maximizer). Note that 
$\angle AMC 
= \pi - \angle A - \angle MCA 
= \pi - \angle A - \angle C/2
= \pi - \pi/2 - \pi/8
=3\pi/8$.
Since $\angle MAI = \pi/4$, it follows that $\triangle AMI$ is isosceles, and therefore $\norm{MA}=\norm{IM}$. We conclude that $R_2(P)$ is maximized at $P=I$. 
At the same time, $R_3(P)$ is clearly minimized at $P=I$.
Hence, $\argmax_{P \in MFAQ} R_2(P)/R_3(P)=I$.

\item[Case 2, $P\in MTUI$:]
We have that 
$$R_2(P)=\norm{PA}=\sqrt{\left(a-\frac{1}{2}\right)^2+\left(b-\frac{1}{2}\right)^2}.$$ 
Also, the line passing through $A,C$ has equation$x+y-1=0$, hence
$R_3(P)=d(P,AC)= \mid a+b-1\mid/\sqrt2$. 
Let $f(a,b)=\left( R_2(P)/R_3(P) \right)^2$, and note that 
$$\nabla f(a,b) = 
2 \frac{a-b}{(a+b-1)^3}
\left(
\begin{array}{c}
2 b-1 \\
 -2 a+1\\
\end{array}
\right).
$$
Observe that $a< b$, and that $a+b-1<1$, hence directions $e_1, -e_2$ are both increasing. It follows that when $P\in MTUI$, ratio $R_2(P)/R_3(P)$ is maximized at $P=U$. At the same time, point $U$ lies on line $x=1/2$, and easy calculations show that $f(1/2,b)=2$ (that is, the function is constant), concluding this case as well.

\item[Case 3, $P\in MFT$:]
As before, we have $R_3(P)=d(P,AC)= \mid a+b-1\mid/\sqrt2$. 
The optimal strategy for $R_2(P)$ is of LD-type. Let $P'$ be the reflection of $P$ around $AB$, that is $P'=(b,a)$. It follows that $R_2(P)=d(P',BC)=a$, and therefore
$$(R_2(P)/R_3(P) )^2 = 2a^2/(a+b-1)^2.$$
It is easy to see that the last ratio is at most 2, exactly when $(b-1)(2a+b-1)\geq 0$. 
Note also that $b\leq 1/2<1$, hence it suffices to prove that $2a+b-1\leq 0$ for all $P \in MFT$. 

In region $MFT$, curve $MT$ is part of a parabola with directrix $x=0$ and focus $A=(1/2,1/2)$, hence it has equation $(y-1/2)^2=x-1/4$. It follows that point $P=(a,b)$ satisfies $(b-1/2)^2\geq a-1/4$. Since also $P$ is on or below line segment $AB$ (with line equation $y=x$), it follows $a\geq b$, whereas we also have $a\geq 0$. So we consider the following non-linear program
$$
\max\{ 2a+b-1: b\leq a \leq (b-1/2)^2+1/4, a\geq 0\},
$$
and we show that its value is bounded above by 0. 
Note the the objective is linear, so the gradient is never the zero-vector. Therefore any optimizers are attained by making some of the inequality constraints tight. 

If $a\leq 0$ becomes tight, then $2a+b-1=b-1\leq -1/2$. 
If $b\leq a$ becomes tight, then $P=M$, and that point was already considered in cases 1,2. 
It remains to examine the case that constraint $a \leq (b-1/2)^2+1/4$ is tight, that is when $P$ lies in the curve segment $MT$ which is contained in region $MTUI$ already considered in case 2. 

\item[Case 4, $P\in FKUT$:]
Again, we have $$R_3(P)=d(P,AC)= \mid a+b-1\mid/\sqrt2.$$ 
The optimal strategy for $R_2(P)$ is now of DL-type. Let $P''$ be the reflection of $P$ around $BC$, that is $P'=(a,-b)$. It follows that 
$$R_2(P)=d(P'',AC)=\mid a+b\mid/\sqrt2.$$
Taking into consideration that $a,b\geq0$, and that $a+b\leq 1$ (for any point $P=(a,b)$, within $ABC$), we have that 
$$g(a,b):=(R_2(P)/R_3(P) )^2 = (a+b)/(1-a-b),$$
which we show next is at most $\sqrt2$. 

Curve $TU$ of region $FKUT$ is part of a parabola with directrix the reflection of $AB$ around $BC$ (that is line $y=-x$), and focus $A=(1/2,1/2)$. Therefore the equation of the parabola is $(x-1/2)^2+(y-1/2)^2=(x+y)^2/2$. We conclude that points $P \in FKUT$ satisfy 
$$(a-1/2)^2+(b-1/2)^2\geq (a+b)^2/2.$$
At the same time we have that $x\leq 1/2$, so it suffices to prove that optimal value to non-linear program 
$$
\max\{ g(a,b): (a-1/2)^2+(b-1/2)^2\geq (a+b)^2/2, a\leq 1/2\}
$$
is at most $\sqrt2$. 

To that end, we observe that $\partial g(a,b)/\partial a = 1/(1-a-b)^2\geq 0$, that is $e_1$ is an increasing direction. We conclude that optimizers of the previous optimization problem within $FKUT$ happen either at line segment $UK$, or at curve segment $TU$. 

As curve segment $TU$ is contained within region $MTUI$ (already considered in case 2), it remains to examine the case that $P$ lies in line segment $UK$, that is $a=1/2$. 
But then, 
$$
g(1/2,b) =1/(1/2-y) -1 
,$$
which is further maximized when $y$ attains its maximum value, i.e. when $P$ coincides with point $U$, and that point was also considered in case 2. 
\end{description}
\end{proof}

The next lemma shows that $\inf_{\tr \in \trs} \mathcal R_{2,3} (\tr)\geq \sqrt2$.

\begin{lemma}
\label{lem: R23 inf upper bound}
For any triangle $\tr \in \trs$, let $I$ denote its incenter. 
Then, we have $R_2(I)/R_3(I) \geq \sqrt2$. 
\end{lemma}

\begin{proof}
Consider $\triangle ABC$ with largest angle $C$, and incenter $I$. Since $I$ is equidistant from all edges, and by Lemma~\ref{lem: r3 regions}, we have $R_3(I) = d(I, BC)$. Moreover, since $C\geq \pi/3$, by Lemma~\ref{lem: R2costIncenter}, we have that $R_2(I)=\norm{IC}$. But then, since $\triangle ABC$ is non-obtuse, we have $C\leq \pi/2$, and so 
$$
\frac{R_2(I)}{R_3(I)} = 
\frac{\norm{IC}}{d(I,BC)}
=
\frac1{\sin(C/2)} \geq \frac1{\sin(\pi/4)}=\sqrt2.
$$
This finishes the proof of the lemma.
  \end{proof}

\subsection{Searching with 1 vs 2 Robots}
\label{sec: 12 sup}

\subsubsection{Supremum Proof; 1 vs 2 Robots}

In this section we prove the following theorem. 

\begin{theorem}
\label{thm: R12 sup}
$\sup_{\tr \in \trs} \mathcal R_{1,2} (\tr)=3.$
\end{theorem}

As we show next, the lower bound for $\sup_{\tr \in \trs} \mathcal R_{1,2} (\tr)$ is attained for the right isosceles triangle (and for certain starting point). The upper bound is much easier and is presented next.

\begin{lemma}
\label{lem: R12 sup upper bound}
For any non-obtuse triangle $\tr \in \trs$, and any starting point $P\in \tr$, we have 
$$R_1(P)/R_2(P)\leq 3.$$
\end{lemma}

\begin{proof}
Consider a non-obtuse $\triangle ABC$, and a point $P$. Without loss of generality, assume that 
$$R_2(P) = \max\{ d(P,AB), d(P,\{BC,CA\}) \}.$$ 
Note that for the $R_2(P)$ solution, one robot follows the optimal trajectory for visiting $AB$ and the other follows the optimal trajectory for visiting $\{BC,CA\}$. Let $T_c$ be the least expensive of the two trajectories, and $T_e$ be the most expensive (and break ties arbitrarily if their costs are equal). 

It suffices to present a heuristic trajectory that one  robot could follow that does not cost more than $3R_2(P)$. 
Indeed, starting from $P$, first move along the cheapest trajectory $T_c$ and return to point $P$, followed by moving along $T_e$. 
Clearly, this trajectory visits all $\{AB,BC,CA\}$, and takes time 
$$2 \min \{ d(P,AB), d(P,\{BC,CA\}) \} + \max \{ d(P,AB), d(P,\{BC,CA\}) \},$$ 
which is at most $3 \max \{ d(P,AB), d(P,\{BC,CA\}) \} = R_2(P)$. Since the optimal $R_1(P)$ has cost at most the cost of the heuristic solution, it follows that $R_1(P)/R_2(P)\leq 3$.
  \end{proof}

In order to prove a matching upper bound, we present a non-obtuse triangle and a starting point for which the optimal $R_1$ trajectory is exactly the heuristic used in the proof of Lemma~\ref{lem: R12 sup upper bound}. Indeed, the next lemma shows that $\sup_{\tr \in \trs} \mathcal R_{1,2} (\tr)\geq 3$, and together with the previous lemma imply Theorem~\ref{thm: R12 sup}.

\begin{lemma}
\label{lem: R12 sup lower bound}
Let $ABC$ be a right isosceles triangle with right angle $A$. Let also $P$ be the middle point of the altitude corresponding to angle $A$.
Then, $R_1(P)/R_2(P)=3$.
\end{lemma}

\begin{proof}
Note that $\angle A \geq \pi/3$, so starting point $P$ lies in the intersection of angle $A$ bisector and the corresponding separating parabola. In particular $R_2(P)=d(P,BC)=d(P,\{AB,AC\})=\norm{AP}$. Without loss of generality, we may assume that $\norm{BC}=1$, and hence $\norm{AB}=\norm{AC}=\sqrt2/2$, so that $R_2(P)=1/4$.
Note also that by Lemma~\ref{lem: half of altitude R1 cost}, the optimal $R_1$ strategy is of LRD type. More specifically, using $\angle A = \pi/2, \angle B = \angle C = \pi/4$ and substituting in cost function gives $R_1(P)=3/4$, and the claim follows. 
\end{proof}

\subsubsection{Infimum Proof; 1 vs 2 Robots}

In this section we prove the following theorem. 

\begin{theorem}
\label{thm: R12 inf}
$\inf_{\tr \in \trs} \mathcal R_{1,2} (\tr)=5/2$.
\end{theorem}

The next lemma shows that $\inf_{\tr \in \trs} \mathcal R_{1,2} (\tr)\leq 5/2$.

\begin{lemma}
\label{lem: R12 inf lower bound}
Consider an equilateral triangle $\tr$. Then, we have  
$\max_{P \in \tr} R_1(P)/R_2(P) = 5/2$.
\end{lemma}

\begin{proof}
We consider equilateral $\triangle ABC$, with bisectors (and altitudes) $AK,BL,CM$, and incenter $I$, see also Figure~\ref{fig: InfR1R2LowerBoundEquilateral}. Let also $W$ be the intersection of $ML$ with $AK$. By symmetry, it is enough to show that 
$$
\max_{P \in AMI} R_1(P)/R_2(P) =5/2
.$$
Below we use the $R_1$ regions that are summarized in Lemma~\ref{lem: R1regionsEquilateral}, and the $R_2$ regions, that are summarized in Corollary~\ref{cor: r2 hexagon equilateral} and Corollary~\ref{cor: r2 equilateral within}.
          \begin{figure}[h!]
\centering
  \includegraphics[width=8cm]{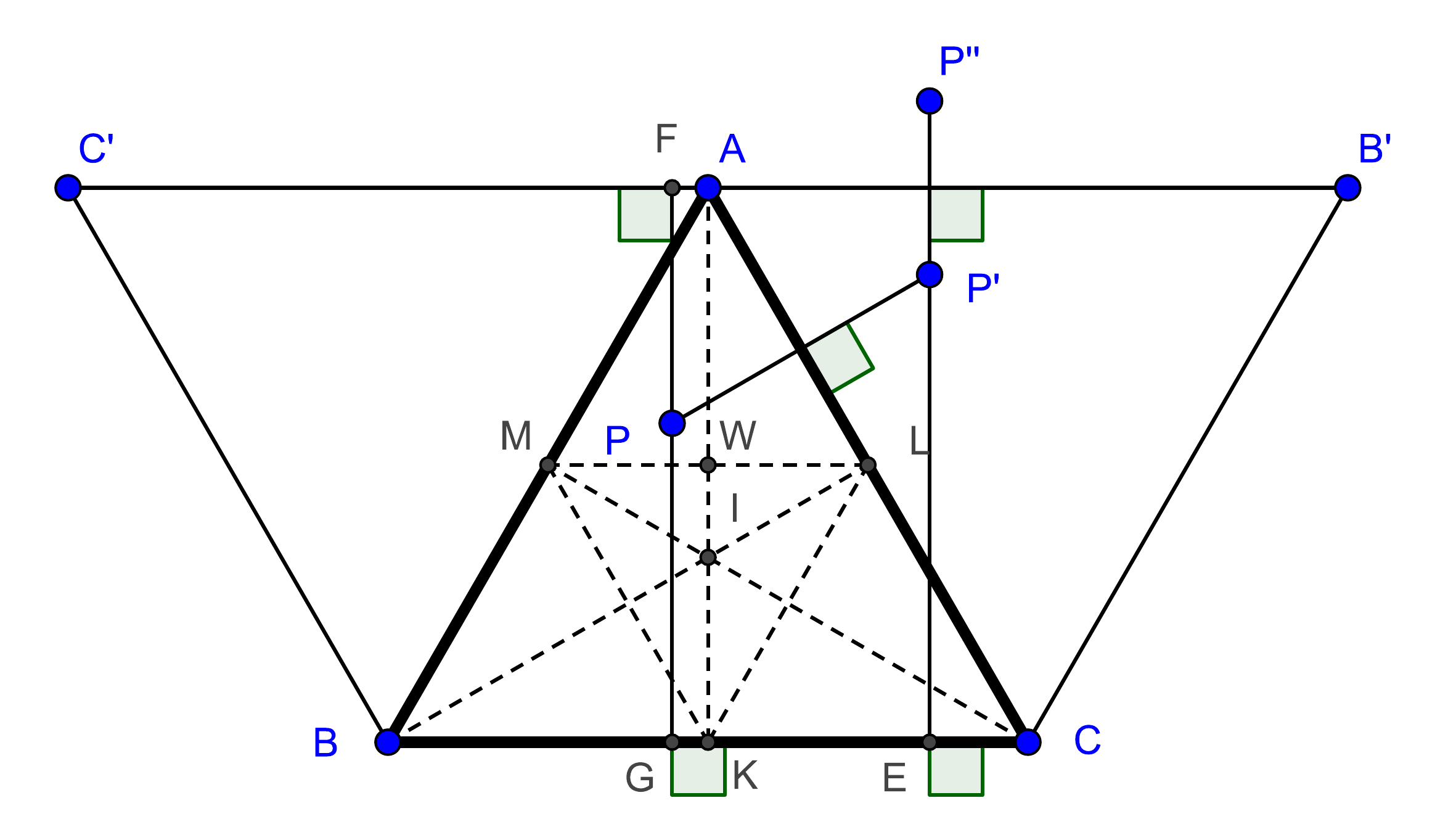}
              \caption{
              Equilateral $\triangle ABC$, and comparison of optimal $R_1, R_2$ strategies for arbitrary starting points. 
              }
              \label{fig: InfR1R2LowerBoundEquilateral}
          \end{figure}

In particular, $R_1$ and $R_2$ regions are determined by the bisectors and points $K,L,M$. More specifically, consider the reflections $C',P',B',P''$ of $C,P,B,P'$ around $AB,AC,AC,AB'$, respectively. We also assume that the triangle is in standard analytic form, i.e. we set 
$A=(p,q)^T, B=(0,0)^T, C=(1,0)^T$ in a Cartesian system (here we treat points as column vectors so at to perform some linear algebra manipulation), where $p=1/2$ and $q=\sqrt3/2$.

For every $P\in AMI$, the optimal $R_1$ strategy is LRD-type, and hence $R_1(P) = d(P'',BC)$. Note also that $P''$ is obtained by rotating $P$ by $2\pi/3$ with center $A$. Hence, if $P=(a,b)^T$ we have that
 $$
 P'' =  \mathcal \mathcal R_{2\pi/3} (P-A)+A
 =
  \mathcal \mathcal R_{2\pi/3} \left( 
\left(
 \begin{array}{c}
 a\\
 b
 \end{array}
 \right)
 -
\left(
 \begin{array}{c}
 p\\
 q
 \end{array}
  \right)
  \right)
  +
  \left(
 \begin{array}{c}
 p\\
 q
 \end{array}
 \right).
 $$
Since also $d(P'',BC)$ equals the second coordinate of $P''$, we have that 
$$
R_1(P) = \sin(2\pi/3)(a-p) +(\cos(2\pi/3)(b-q)+q
=
\frac{\sqrt3}2(a-p)-\frac12(b-q)+q. 
$$

For starting points $P\in AMI$, we have two cases regarding the cost of the optimal $R_2$ strategy. If $P\in AMW$, then we have $R_2(P)=\norm{PG}=b$. 
If $P\in MIW$, then the dominant cost for the $R_2$ strategy is due to a robot that visits $AB, AC$. Note also that $d(P, \{AB,AC\}) = d(P,AC') = q-b$. Overall, we have that for all $P \in AMI$ 
$$
R_2(P) = \max \{ b, q-b\}.
$$

Combining the above, we have that 
\begin{align*}
\max_{P \in AMI} \frac{R_1(P)}{R_2(P)}
& = 
\max_{a,b} 
\frac{\frac{\sqrt3}2(a-1/2)-\frac12(b-\sqrt3/2)+\sqrt3/2}
{\max \{ b, \sqrt3/2-b\}} \\
& \stackrel{(a\leq 1/2)}{\leq} 
\max_{b} 
\frac{-\frac12(b-\sqrt3/2)+\sqrt3/2}
{\max \{ b, \sqrt3/2-b\}} \\
&=
\frac52,
\end{align*}
i.e. the maximum is attained at $a=1/2, b=\sqrt3/4$ which is point $W$. 
  \end{proof}

The next lemma shows that $\inf_{\tr \in \trs} \mathcal R_{1,2} (\tr)\geq 5/2$.

\begin{lemma}
\label{lem: R12 inf upper bound}
For any $\triangle ABC \in \trs$, let $T$ be the middle point of the altitude corresponding to the largest edge. Then, we have $R_1(T)/R_2(T) \geq 5/2$. 
\end{lemma}

\begin{proof}
Without loss of generality, we assume that $BC$ is the largest edge of non-obtuse $\triangle ABC$, hence $\angle A$ is the largest angle. 
Consider the standard analytic representation of $ABC$.
By Lemma~\ref{lem: R2 region partial}, we have that $R_2(T)=d(T,BC)=q/2$, which can be expressed by Observation~\ref{obs: A,I coordinates angles} using only triangle angles.
By Lemma~\ref{lem: half of altitude R1 cost}, we also have that  the optimal strategy of $R_1(T)$ is of LRD type and the cost is expressed as a function of the angles. But then, it is immediate that 
$${R_1(T)}/{R_2(T)}=2 -\cos(2A)\geq 5/2,$$
 where the last inequality is due to that $\angle A\geq \pi/3$.
\end{proof}

\section{Conclusions}

We considered a new vehicle routing-type problem in which (fleets of) robots visit all edges of a triangle. We proved tight bounds regarding visitation trade-offs with respect to the size of the available fleet. In order to avoid degenerate cases of visiting the edges with 3 robots, we only focused our study on non-obtuse triangles. The case of arbitrary triangles, as well as of other topologies, e.g. graphs, remains open. We believe the definition of our problem is of independent interest, and that the study of efficiency trade-offs in combinatorial problems with respect to the number of available processors (that may not be constant as in our case), e.g. vehicle routing type problems, will lead to new, deep and interesting questions. 

\nocite{*}
\bibliographystyle{abbrvnat}
\bibliography{triangletradeoffrefs}
\label{sec:biblio}

\end{document}